\documentclass[aos]{imsart}

\RequirePackage[OT1]{fontenc}
\RequirePackage{amsthm,amsmath,amsfonts,amssymb}
\RequirePackage[round]{natbib}
\RequirePackage[colorlinks,citecolor=blue,urlcolor=blue]{hyperref} 

\usepackage{graphicx}
\usepackage{enumerate}
\usepackage[latin1]{inputenc} 
\usepackage{verbatim}

\usepackage{paralist}
\usepackage{pstool}
\usepackage{booktabs}
\usepackage{tikz}
\usepackage{mathtools}
\usepackage{bbm}
\usepackage{url}
\usepackage{nicefrac}

\newtheorem{Proposition}{Proposition}
\newtheorem{Lemma}{Lemma}
\newtheorem{Theorem}{Theorem}
\newtheorem{Corollary}{Corollary}
\newtheorem{Remark}{Remark}
\newtheorem{Definition}{Definition}
\newtheorem{example}{Example}

\DeclareMathOperator{\eig}{\varphi}

\DeclareMathOperator{\pd}{\partial}

\DeclareMathOperator{\Cov}{Cov}

\DeclareMathOperator{\diag}{diag}

\newcommand{\proper}{\mathsf}

\newcommand{\pE}{\proper{E}}

\newcommand{\pN}{\proper{N}}
\newcommand{\mv}[1]{{\boldsymbol{\mathrm{#1}}}}

\newcommand{\trsp}{\ensuremath{\top}}

\newcommand{\<}{\langle}
\renewcommand{\>}{\rangle}
\newcommand{\transp}{\top}

\DeclarePairedDelimiter\ceil{\lceil}{\rceil}
\DeclarePairedDelimiter\floor{\lfloor}{\rfloor}

\newcommand{\kT}{\kappa\ell}

\usepackage{ulem}

\graphicspath{{.}
	{figs/}}

\newcommand{\A}{\mathcal{C}}
\newcommand{\Ac}{\mathcal{U}}

\allowdisplaybreaks

\startlocaldefs

\renewcommand{\theenumi}{\normalfont{(\roman{enumi})}}

\endlocaldefs
\usepackage{multirow}
\begin{document}
	
	\begin{frontmatter}
		
		\title{Statistical inference for Gaussian Whittle--Mat\'ern fields on metric graphs}
		\runtitle{Statistical inference for Gaussian  Whittle--Mat\'ern fields on metric graphs}

		\begin{aug}
			\author[A]{\fnms{David} \snm{Bolin}\ead[label=e1,mark]{david.bolin@kaust.edu.sa}} 
			\and 
			\author[A]{\fnms{Alexandre B.} \snm{Simas}\ead[label=e2]{alexandre.simas@kaust.edu.sa}} 
			\and
			\author[B]{\fnms{Jonas} \snm{Wallin}\ead[label=e3]{jonas.wallin@stat.lu.se}}
			
			\runauthor{David Bolin, Alexandre Simas and Jonas Wallin }

			\address[A]{Statistics Program, Computer, Electrical and Mathematical Sciences and Engineering Division, King Abdullah 
				University of 
				Science and Technology, 
				\printead{e1}, \printead{e2}}
			\address[B]{Department of Statistics,
				Lund University,
				\printead{e3}} 
		\end{aug}
		
		\begin{abstract}
			Whittle--Mat\'ern fields are a recently introduced class of Gaussian processes on metric graphs, which are  
			specified as solutions to a fractional-order stochastic differential equation. Unlike 
			earlier covariance-based approaches 
			for specifying Gaussian fields on metric graphs, the Whittle--Mat\'ern fields are well-defined for any compact 
			metric graph and can provide Gaussian processes with differentiable sample paths. 
			We derive the main statistical properties of the model class, particularly the consistency and asymptotic normality of 
			maximum likelihood estimators of model parameters and the necessary and sufficient conditions for asymptotic 
			optimality properties of linear prediction based on the model with misspecified parameters.

			The covariance function of the Whittle--Mat\'ern fields is generally unavailable in closed form, and 
			they have therefore been challenging to use for statistical inference. However, we show that for specific values of 
			the fractional exponent, when the fields have Markov properties, likelihood-based inference and spatial prediction 
			can be performed exactly and computationally efficiently. This facilitates using the Whittle--Mat\'ern fields
			in  statistical applications involving big datasets without the need for any approximations. The methods are illustrated 
			via an application to modeling of traffic data, where allowing for differentiable processes 
			dramatically improves the results.  
		\end{abstract}
		
		\begin{keyword}[class=MSC]
	\kwd[Primary ]{62M30} 
	\kwd[; secondary ]{} 
	\kwd{62M05} 
	\kwd{60G60}
	\kwd{35R02}
\end{keyword}

\begin{keyword}
	\kwd{Networks,  Gaussian processes, stochastic differential equations, Gaussian Markov random fields, quantum graphs}
\end{keyword}
		
\end{frontmatter}

\section{Introduction and preliminaries}
\subsection{Introduction}
There is a growing interest in the statistical modeling of data on compact metric graphs, such as street or river networks, 
based on Gaussian random fields \citep{okabe2012spatial,baddeley2017stationary,cronie2020, moller2022lgcp, porcu2022}. 
One approach to formulating Gaussian processes on metric graphs is to specify them in terms of a covariance function $\widetilde{\varrho}(d(\cdot,\cdot))$, 
where $d(\cdot,\cdot)$ is a metric on the graph and $\widetilde{\varrho}(\cdot)$ is an isotropic covariance function. 
The difficulty with this approach is to ensure that the  function $r(\cdot,\cdot) = \widetilde{\varrho}(d(\cdot,\cdot))$ is positive semi-definite. 
However, \citet{anderes2020isotropic} demonstrated this is the case for several covariance functions if the graph has Euclidean edges, 
and $d$ is chosen as the so-called resistance metric. For example, a valid choice is the Mat\'ern covariance function:
\begin{equation}\label{eq:matern_cov}
	r(s,s') = \varrho_M(d(s,s')), \quad \varrho_M(h) = \frac{\tau^{-2}}{ 2^{\nu-1}\Gamma(\nu + \nicefrac{n}{2}) (4\pi)^{\nicefrac{n}{2}}\kappa^{2\nu}}(\kappa |h|)^{\nu}K_\nu(\kappa |h|),
\end{equation}
where $n=1$ and the parameters $\tau$, $\kappa>0$ and $0<\nu\leq \nicefrac1{2}$ control the variance, practical correlation range, 
and the sample path regularity, respectively. 
Further, $K_\nu(\cdot)$ is a modified Bessel function of the second kind and $\Gamma(\cdot)$ denotes the gamma function.  			
	
As $\nu\leq \nicefrac1{2}$ is required in \eqref{eq:matern_cov}, this approach cannot create differentiable Gaussian Mat\'ern-type processes on metric graphs, 
even if they have Euclidean edges. 
Because of this, \citet{BSW2022} proposed to instead create Whittle--Mat\'ern Gaussian fields on a compact metric graph $\Gamma$ 
by considering the differential equation 
\begin{equation}\label{eq:Matern_spde}
	(\kappa^2 - \Delta_\Gamma)^{\alpha/2} (\tau u) = \mathcal{W}, \qquad \text{on $\Gamma$},
\end{equation}
where $\alpha = \nu + \nicefrac1{2}$, $\Delta_\Gamma$ is the so-called Kirchhoff--Laplacian, and $\mathcal{W}$ is Gaussian white noise. 
The motivation for considering this particular equation is that when \eqref{eq:Matern_spde} is considered on $\mathbb{R}^n$ 
(with $\alpha = \nu + \nicefrac{n}{2}$ and $\Delta_\Gamma$ is replaced by the standard Laplacian), 
it has Gaussian random fields with the covariance function \eqref{eq:matern_cov}, where $d(\cdot,\cdot)$ is the Euclidean distance on $\mathbb{R}^n$, 
as stationary solutions \citep{whittle63}. 
As for Euclidean domains,  \citet{BSW2022} proved that the parameter $\alpha$ controls sample path regularity of the process
in the metric graph setting. 
In particular, the solution $u$ is a well-defined Gaussian random field if $\alpha>\nicefrac12$, 
which has a modification with almost surely (a.s.) $\gamma$-H\"older continuous sample paths for 
$0<\gamma < \min\{\alpha-\nicefrac12, \nicefrac12\}$, and if $\alpha>\nicefrac32$, the sample paths of the process are a.s.~weakly differentiable. 
The Whittle--Mat\'ern fields thus provide a natural analog to the Gaussian Mat\'ern fields on Euclidean domains. 
However, their statistical properties have so far not been studied, 
and no methods for inference of the processes have been proposed because their covariance function has not been available in closed form. 
We aim to fill this gap. 
	
\begin{figure}
	\centering
	\includegraphics[width=\linewidth]{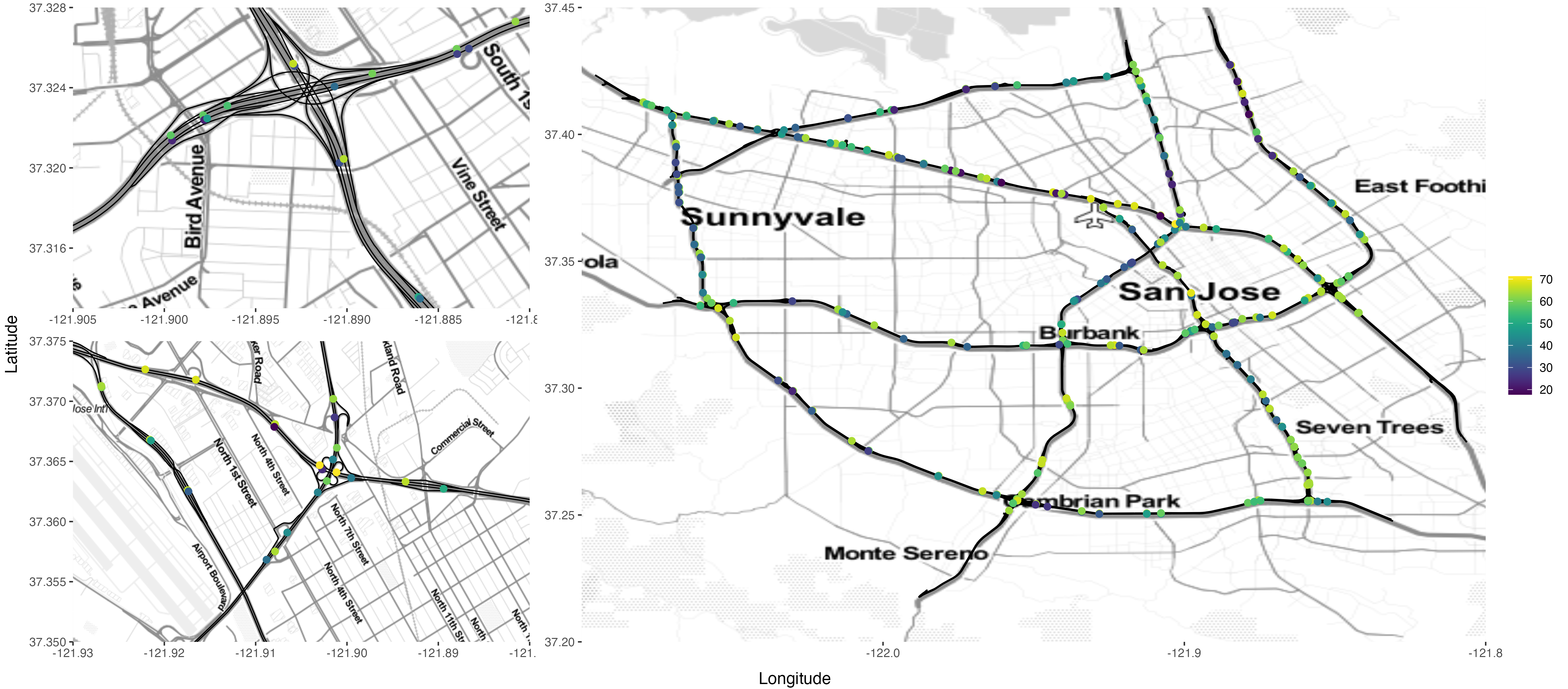}
	\caption{Average speeds observed every Monday at 17:30 during the first half of 2017 on the San Jose highway.
		Left panels: Zoomed-in areas of the panel to the right. The graph edges are shown as black curves. }
	\label{fig:data}
\end{figure}

Specifically, we demonstrate that the Whittle--Mat\'ern fields have statistical properties similar to their Euclidean counterparts.
For a fixed $\alpha$, the parameter $\tau$ (but not $\kappa$) can be estimated consistently under infill asymptotics, 
and the maximum-likelihood estimator for $\tau$ is asymptotically normal. 
We also consider kriging prediction of the fields under misspecified parameters, 
and show that one obtains uniform asymptotic optimality of the optimal linear predictor as long as $\alpha$ is correctly specified. 
Thus, accurately estimating $\kappa$ or $\tau$ is not important  to obtaining good predictions for large datasets. 
Further, an important feature of Gaussian Mat\'ern fields on $\mathbb{R}^d$ with $\alpha\in\mathbb{N}$ is that they are Markov random fields. 
This fact also holds for the Whittle--Mat\'ern fields on metric graphs \citep{BSW_Markov}, 
and we demonstrate that this can be used to evaluate finite-dimensional distributions of the fields, 
perform  likelihood-based inference, and do spatial prediction exactly and computationally efficiently. 
Thus, one can use the fields in applications involving large metric graphs and large datasets. 
	
In particular, we prove that the Whittle--Mat\'ern field (and the derivatives of the process, if they exist) 
evaluated at any finite number of locations on $\Gamma$ is a Gaussian Markov random field (GMRF) with a sparse precision matrix, 
allowing for numerically efficient likelihood-based inference, prediction and interpolation. 
Using GMRF approximations of Gaussian fields for spatial data is common. 
For example, the stochastic partial differential equation (SPDE) approach for Whittle--Mat\'ern field on Euclidean domains \citep{lindgren11, lindgren2022spde} 
results in one such approximation. 
The difference with the proposed method is that the GMRF is not an approximation; therefore,  we can perform exact and computationally efficient inference.
	
As a motivating example of why these fields are useful for data analysis, 
we consider a data set of traffic speed observations on highways in the city of San Jose, California (Figure~\ref{fig:data}), 
studied earlier in \cite{borovitskiy2021matern}. 
Due to the complicated graph structure of the highway network, this is an example of a graph with non-Euclidean edges.  
Further, as observed later, differentiability of the field dramatically improves the model fit, 
revealing the importance of differentiable fields even if the graph has Euclidean edges.
As an illustration of the computational efficiency of the model, 
the computation time for exact log-likelihood evaluation of the model for $n$ observations on the street network is illustrated in Figure~\ref{fig:timings}. 
The green curve corresponds to a covariance-based approach, 
whereas the red and black curves depict two methods we derive that take advantage of sparsity. 
			
\begin{figure}[t]
	\includegraphics[width=0.48\linewidth]{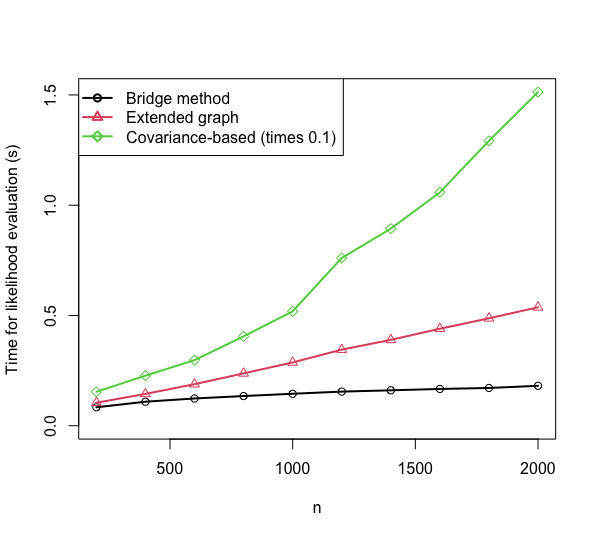}
	\vspace{-0.5cm}
	\caption{Average computation time for evaluating the log-likelihood for $n$ observations using three approaches, for the graph in Figure~\ref{fig:data}. The computation times for the covariance-based approach is multiplied by $0.1$ to more easily observe the differences between the other two approaches. See Section~\ref{sec:simulation} for details.}
	\label{fig:timings}
\end{figure}
		
The paper is organized as follows. Section~\ref{sec:construction} provides the necessary notation, 
defines Whittle--Mat\'ern fields, and summarizes their regularity properties. 
The main statistical properties of the fields are derived in Section~\ref{sec:statistical}. 
Then, Section~\ref{sec:markov} focuses on the case $\alpha\in\mathbb{N}$ and derives two alternative representations of the fields that 
can be used for exact evaluation of finite-dimensional distributions. 
Sections~\ref{sec:inference} and \ref{sec:prediction} derive the exact and computationally efficient methods 
for likelihood-based inference and spatial prediction, respectively.  
Next, Section~\ref{sec:boundary} discusses alternative boundary conditions, 
which can be used to remove boundary effects at vertices of degree 1, 
and Section~\ref{sec:graph_laplacian} compares the Whittle--Mat\'ern fields to the Mat\'ern-like models based on the graph Laplacian, 
which have also recently received much attention \citep{Alonso2021,dunson2020graph,borovitskiy2021matern}. 
Finally, Section~\ref{sec:application} considers the data in Figure~\ref{fig:data} and compares the predictive qualities of 
the Whittle--Mat\'ern fields with those of models based on the isotropic exponential covariance function and the graph Laplacian. 
Notably, the advantage of working with differentiable Gaussian processes on metric graphs is apparent for this application. 
The article concludes with a discussion in Section~\ref{sec:discussion}. 
Technical details and proofs are provided in the appendices, 
and all models and computational methods introduced in this work are implemented in the \texttt{R} software package \texttt{MetricGraph} \citep{bsw_MetricGraph_cran} available on CRAN.

\section{Whittle--Mat\'ern fields on compact metric graphs}\label{sec:construction}
This section reviews the construction and main properties of the Whittle--Mat\'ern fields on metric graphs as introduced in \citet{BSW2022}. 
However, we begin by introducing the notation for the article and critical facts about compact metric graphs.    
	
\subsection{Preliminaries and notation}	
Throughout the article, $(\Omega, \mathcal{F},\mathbb{P})$ denotes a complete probability space, 
and $\pE(Z) = \int_\Omega Z(\omega) d\mathbb{P}(\omega)$ denotes the expectation of  a real-valued random variable $Z$. 
The Hilbert space of all (equivalence classes) of real-valued random variables, $Z$, with a finite second moment, 
$\pE(Z^2) < \infty$, is denoted by $L_2(\Omega)$. 
We let $\Gamma$ denote a compact metric graph consisting of a set of finitely many vertices $\mathcal{V}=\{v_i\}$ and a 
finite set $\mathcal{E}=\{e_j\}$ of edges connecting the vertices. 
For practical purposes, we assume that $d\in\mathbb{N}$ exists such that $\mathcal{V}\subset \mathbb{R}^d$, 
and that each edge $e$ is defined by a rectifiable curve $\gamma:[0,\ell_e]\to \mathbb{R}^d$, 
where $\gamma$ is parameterized by the arc-length, so that $0<\ell_e <\infty$ is the length of the curve. 
The vertices at the start and end of the curve are denoted by $\underline{e}$ and $\bar{e}$, respectively, where $\underline{e},\bar{e}\in\mathcal{V}$.
We assume that the graph is connected so that a path exists between all vertices and write $u\sim v$ or $(u,v)\in\mathcal{E}$ 
if the vertices $u$ and $v$ are connected by an edge. 
The graph is equipped with the geodesic metric, 
denoted by $d(\cdot,\cdot)$  from now on, 
which for any two points in $\Gamma$ is defined as the length of the shortest path in $\Gamma$ connecting the two.
This metric is well-defined because $\Gamma$ is connected.
For every $v\in\mathcal{V}$, we let $\mathcal{E}_v$ denote the set of edges incident to the vertex $v$, 
and define the degree of $v$ by $\deg(v) = \#\mathcal{E}_v$. 
A location $s\in \Gamma$ is a position on an edge and can thus be represented as a pair $(e,t)$, 
where $t\in[0,\ell_e]$ and $e\in\mathcal{E}$. 
For a function $f$ on $\Gamma$, we let $f_e = f|_e$ denote the restriction of the function to the edge, 
and write $f_e(t)$ for $t\in[0,\ell_e]$, or $f_e(s)$, to denote the value of $f(s)$ with $s=(e,t)$.
	
We let $L_2(e)$ denote the space of square-integrable functions on the edge $e\in\mathcal{E}$, 
which is equipped with the Lebesgue measure. 
The space $L_2(\Gamma) = \bigoplus_{e \in \mathcal{E}} L_2(e)$ is defined as the direct sum of the $L_2(e)$ spaces 
and is equipped with the norm $\|f\|_{L_2(\Gamma)}^2 =  \sum_{e\in\mathcal{E}}\|f_e\|_{L_2(e)}^2$.
That is, $f = \{f_e\}_{e\in \mathcal{E}} \in L_2(\Gamma)$ if $f_e \in L_2(e)$ for each $e\in \mathcal{E}$.
The space of continuous functions on $\Gamma$ is denoted by $C(\Gamma) = \{f\in L_2(\Gamma): f\hbox{ is continuous}\}$, 
which is equipped with the supremum norm $\|\phi\|_{C(\Gamma)} = \sup\{|\phi(x)|: x\in\Gamma\}$. 	
For any $k\in\mathbb{N}$, we introduce the decoupled Sobolev space $\widetilde{H}^k(\Gamma) = \bigoplus_{e\in \mathcal{E}} H^k(e)$, 
endowed with the norm $\|f\|_{\widetilde{H}^k(\Gamma)}^2 = \sum_{e\in\mathcal{E}} \|f\|_{H^k(e)}^2,$ 
where $H^k(e)$ is the Sobolev space of order $k$ on $e$ \citep[see ][Appendix A for details]{BSW2022}.
We define the Sobolev space $H^1(\Gamma) = C(\Gamma)\cap \widetilde{H}^1(\Gamma)$ as the space of all continuous functions on $\Gamma$ such that 
$\|f\|_{H^1(\Gamma)} = \|f\|_{\widetilde{H}^1(\Gamma)} < \infty$.
It follows from standard Sobolev space theory that, if $f\in\widetilde{H}^1(\Gamma)$, then, $f|_e$ is continuous for each edge $e\in\mathcal{E}$. 
However, $f$ might have discontinuities at the vertices. 
Therefore, the continuity assumption of $H^1(\Gamma)$ guarantees that $f$ is uniquely defined at the vertices. 
For $u\in H^1(\Gamma)$, the weak derivative $u'\in L_2(\Gamma)$ is defined as the function whose restriction to any edge $e$ coincides 
almost everywhere (a.e.) with the weak derivative of $u|_e$, 
which is well-defined because ${u|_e\in H^1(e)}$. 

\subsection{Model construction and properties}
To define the Whittle--Mat\'ern fields, we introduce 
$K(\Gamma) = \{f\in\widetilde{H}^2(\Gamma): \forall v\in\mathcal{V},\, \sum_{e\in\mathcal{E}_v} \partial_e f(v) = 0\}$, 
where $\partial_eu(v)$ is the directional derivative of $u$ on $e$ in the directional away from $v$ (i.e., 
if $e=[0,\ell_e]$, then ${\partial_e u(0) = u_e'(0)}$, and ${\partial_eu(\ell_e) = -u_e'(\ell_e)}$).
We define 
${\Delta_\Gamma : \mathcal{D}(\Delta_\Gamma) = C(\Gamma) \cap K(\Gamma) \subset L_2(\Gamma) \rightarrow L_2(\Gamma)}$
by ${\Delta_\Gamma := \oplus_{e\in \mathcal{E}}\Delta_e}$, where $\Delta_e u_e(x) = u_e''(x)$ with $u_e\in H^2(e)$. 
This operator acts as the second derivative on the edges and enforces the Kirchhoff vertex conditions on $f\in \widetilde{H}^2(\Gamma)$:
$\forall v\in\mathcal{V},  f \mbox{ is continuous at $v$ and} \sum_{e \in\mathcal{E}_v} \partial_e f(v) = 0$. 
These vertex conditions are the natural extension of Neumann boundary conditions to the graph setting and 
coincide with homogeneous Neumann boundary conditions for vertices of degree 1. 
Therefore they are often called Kirchhoff--Neumann, Neumann, or even ``standard'' boundary conditions in the physics literature \citep{Berkolaiko2013}. 
	
Next, we let $\kappa^2>0$ and define the operator 
$L : \mathcal{D}(L) = \mathcal{D}(\Delta_\Gamma)\subset L_2(\Gamma) \rightarrow L_2(\Gamma)$ 
as the shifted Kirchhoff-Laplacian $L = \kappa^2 - \Delta_\Gamma.$
According to \citet[][Theorem 1.4.4]{Berkolaiko2013}, $\Delta_\Gamma$ is self-adjoint, which implies that $L$ is also self-adjoint. 
Further, $L$ is densely defined, strictly positive-definite, and has a discrete spectrum where each eigenvalue has finite multiplicity \citep{BSW2022}.
We let $\{\hat{\lambda}_i\}_{i\in\mathbb{N}}$ denote the (nonnegative) eigenvalues of $-\Delta_{\Gamma}$, sorted in nondecreasing order, 
and let $\{\eig_j\}_{j\in\mathbb{N}}$ denote the corresponding eigenfunctions. 
Then, $L$ clearly diagonalizes with respect to the eigenfunctions of $\Delta_{\Gamma}$, 
and has eigenvalues $\{ \lambda_i\}_{i\in\mathbb{N}} = \{\kappa^2 + \hat{\lambda}_i\}_{i\in\mathbb{N}}$. 
By Weyl's law \citep{Odzak2019Weyl}, we have that $\hat{\lambda}_i \sim i^2$ as $i\rightarrow\infty$. 
Hence, constants $c_\lambda$ and $C_\lambda$ exist such that ${0<c_\lambda< C_\lambda<\infty}$ and
\begin{equation}\label{eq:weyl}
	\forall i\in\mathbb{N}, \quad c_\lambda i^2 \leq \lambda_i \leq C_\lambda i^2.
\end{equation}
	
For $\beta>0$, we introduce the fractional operator $L^{\beta}$ in the spectral sense. 
Start by defining the space $\mathcal{D}(L^\beta) = \dot{H}^{2\beta} := \{ \phi \in L_2(\Gamma) : \|\phi\|_{2\beta} < \infty\}$,
where, for $\phi \in L_2(\Gamma)$, we have $\|\phi\|_{2\beta} := \left(\sum_{j\in\mathbb{N}} \lambda_j^{2\beta} (\phi,\varphi_j)_{L_2(\Gamma)}^2\right)^{\nicefrac{1}{2}}$.
We let the action of $L^\beta : \mathcal{D}(L^\beta) \rightarrow L_2(\Gamma)$ be defined by 
$
L^\beta \phi := \sum_{j\in\mathbb{N}} \lambda_j^\beta (\phi,\eig_j)_{L_2(\Gamma)}\eig_j, \phi\in \mathcal{D}(L^\beta).
$
For	$\phi,\psi\in\dot{H}^{2\beta}$, we define the inner product $(\phi,\psi)_{2\beta} := (L^\beta\phi, L^\beta \psi)_{L_2(\Gamma)}$, and note that $\|\phi\|_{2\beta} = \|L^\beta \phi\|_{L_2(\Gamma)}$. Observe that $(\dot{H}^{2\beta}, (\cdot,\cdot)_{2\beta})$ is a Hilbert space. 
	
We define the Whittle--Mat\'ern fields through the fractional-order equation \eqref{eq:Matern_spde}, that is, 
$L^{\nicefrac{\alpha}{2}} (\tau u) = \mathcal{W}$, where $\mathcal{W}$ denotes Gaussian white noise on $L_2(\Gamma)$, 
which can be represented as a family of centered Gaussian variables $\{\mathcal{W}(h) : h\in L_2(\Gamma)\}$ satisfying ${\forall h,g\in L_2(\Gamma)}$, $\pE[\mathcal{W}(h)\mathcal{W}(g)] = (h,g)_{L_2(\Gamma)}$.
Given that $\alpha>1/2$, \eqref{eq:Matern_spde} has a unique solution $u\in L_2(\Gamma)$ $\mathbb{P}$-a.s. \citep[][Proposition 1]{BSW2022}, which is a centered Gaussian random field satisfying
\begin{equation*}
	\forall \psi\in L_2(\Gamma),\quad (u,\psi)_{L_2(\Gamma)} = \mathcal{W}(\tau^{-1}L^{-\alpha/2}\psi) \quad \mathbb{P}\text{-a.s.}.
\end{equation*}
The following theorem summarizes the main properties of the Whittle--Mat\'ern fields needed later (see Appendix \ref{app:proofs_edge} for auxiliary definitions).
	
\begin{Theorem}[\citet{BSW2022}]\label{thm:regularity}
Suppose that $\alpha>\nicefrac12$, let $u$ be the solution of \eqref{eq:Matern_spde}, 
 let $\lceil a \rceil$ denote the smallest integer larger than or equal to $a\in\mathbb{R}$, 
and let $\lfloor a \rfloor$ denote the largest integer less than or equal to $a\in\mathbb{R}$. 
\begin{enumerate}
	\item $u$ has a modification with continuous sample paths. Further, the covariance function of $u$,
	${\varrho(s,s') = \pE(u(s)u(s'))}$ for $(s,s')\in \Gamma\times \Gamma$, is continuous (i.e., $u$ is $L_2(\Omega)$-continuous).
	\item Let $\alpha \geq 2$. Then, for every edge $e\in\mathcal{E}$, 
	the weak derivatives in the $L_2(\Omega)$ sense of $u_e$ up to order $\lfloor \alpha\rfloor-1$ exist and are weakly continuous on $e$ in the $L_2(\Omega)$ sense.
	\item Let $\alpha\geq 2$, fix any edge $e\in\mathcal{E}$, and let $u_e^{(k)}(\cdot)$, $k=0,\ldots,\floor{\alpha}-1$, 
	be the $k$th order directional weak derivative of $u$. 
	Then, for any $t_1,t_2\in e$, and $j,k \in \{0,\ldots, \lfloor\alpha\rfloor-1\}$, 
	$$\pE\left(u_e^{(j)}(t_1) u_e^{(k)}(t_2)\right) = \frac{\partial^{j+k}}{\partial t_1^j \partial^kt_2}\pE(u_e(t_1) u_e(t_2)).$$
	\item The derivatives of $u$ in the weak $L_2(\Omega)$ sense agree with the weak derivatives of $u$ in the Sobolev sense, 
	whenever both of them exist simultaneously.
	\item Let $\alpha\geq 2$ and let $k \in \{0,\ldots, \ceil{\alpha - \nicefrac{1}{2}}  -1\}$. 
	If $k$ is odd, $\sum_{e\in\mathcal{E}_v} \partial_{e}^{k} u(v) = 0$, for each $v\in\mathcal{V}$. If $k$ is even, 
	for each $v\in\mathcal{V}$ and each pair $e,e'\in\mathcal{E}_v$, $\partial_e^{k} u(v) = \partial_{e'}^{k}u(v)$.\label{thm:regularity:item:Kirchhoff}
	\item Suppose that $\Gamma$ has a vertex $v_i$ of degree 2, connected to edges $e_k$ and $e_\ell$, 
	and define $\widetilde{\Gamma}$ as the graph where $v_i$ is removed and $e_k,e_\ell$ are merged to a new edge $\widetilde{e}_k$. 
	Let $u$ and $\widetilde{u}$ be the solutions to \eqref{eq:Matern_spde} on $\Gamma$ and $\widetilde{\Gamma}$, respectively. 
	Then $u$ and $\widetilde{u}$ have the same covariance function.\label{prop:join}
\end{enumerate}
\end{Theorem}

Further, let $\varrho(\cdot,\cdot)$ be the covariance function of $u$, where $u$ is given by the solution to \eqref{eq:Matern_spde}. 
Regarding $\varrho(\cdot,\cdot)$, in \cite{BSW2022} it was only known that $\varrho$ is positive semi-definite, as it is a covariance function. 
However, to simplify the usage of the models in applications, and to derive the statistical properties in the next section, we need $\varrho(\cdot,\cdot)$ to be strictly positive-definite. This is indeed the case and is a new result, stated in Proposition \ref{prp:strposdef}, whose proof is provided in Appendix \ref{app:strposdef}. 
\begin{Proposition}\label{prp:strposdef}
	The function $\varrho:\Gamma\times\Gamma\to\mathbb{R}$ is strictly positive-definite.
\end{Proposition}

\begin{Remark}\label{rem:strposdef}
	Observe that Proposition \ref{prp:strposdef} is not trivial. Indeed, we have trivially that for every function $f\in L_2(\Gamma)$, $f\neq 0$,
	\begin{equation}\label{eq:covoppos}
		\sum_{e,e'\in\mathcal{E}} \int_e \int_{e'} \varrho(x,y) f(x)f(y) \, dxdy 
		= \|\tau^{-1}L^{-\alpha/2} f\|_{L_2(\Gamma)}^2 > 0,
	\end{equation}
	where the equality follows from \citet[Lemma 3]{BSW2022} and the inequality comes from the fact that $L^{-\alpha/2}$ is injective.
	This condition, however, does not imply that $\varrho(\cdot,\cdot)$ is strictly positive-definite. 
	For example, consider the solution to ${(\kappa^2 - \Delta_D)^{\nicefrac{1}{2}} (\tau u_D) = \mathcal{W}}$ on the interval $[a,b]$, $a<b$, where $\Delta_D$ is the Dirichlet Laplacian. Let $\varrho_D(\cdot,\cdot)$ be the covariance function of $u_D$. In this case, by the same reasons, $\varrho_D(\cdot,\cdot)$ satisfies \eqref{eq:covoppos} (considering the metric graph as a single interval). Nevertheless, $\varrho_D(\cdot,\cdot)$ is not strictly positive-definite as ${\varrho_D(a,\cdot) = \varrho_D(\cdot,b) = 0}$.
\end{Remark}

\section{Statistical properties}\label{sec:statistical}
The goal in this section is to derive some of the most vital statistical properties of the  Whittle--Mat\'ern fields. 
Let $u$ be the solution to \eqref{eq:Matern_spde} on $\Gamma$. Then, $u$ has covariance operator $\mathcal{C} = \tau^{-2}L^{-\alpha}$ satisfying $(\mathcal{C}\phi,\psi)_{L_2(\Gamma)} = \pE[(u,\phi)_{L_2(\Gamma)}(u,\psi)_{L_2(\Gamma)}]$ for all $\phi,\psi \in L_2(\Gamma)$.
Let $\mu(\cdot; \kappa,\tau,\alpha) = \pN(\cdot;0,\tau^{-2}L^{-\alpha})$ denote the Gaussian measure corresponding to $u$ on $L_2(\Gamma)$,
with zero mean and covariance operator $\tau^{-2}L^{-\alpha}$. That is, for every Borel set ${B\in \mathcal{B}(L_2(\Gamma))}$, we have 
$\mu(B; \kappa,\tau,\alpha) = \mathbb{P}(\{\omega\in\Omega : u(\cdot, w)\in B\}).$
In the following subsections, we investigate the consistency of maximum-likelihood parameter estimation of Whittle--Mat\'ern fields, 
and the asymptotic optimality of kriging prediction based on misspecified model parameters. 
		
First, we introduce some additional notation. 
For a Gaussian random field $u$ on $\Gamma$, such as a Whittle--Mat\'ern field, and $S\subset \Gamma$, 
$U(S)$ denotes the vector space of all linear combinations of $u(s)$, $s\in S$ (i.e., elements of the form 
${\gamma_1 u(s_1) + \ldots + \gamma_N u(s_N)}$, where $N\in\mathbb{N}$ and $\gamma_j \in \mathbb{R}$, $s_j \in S$ for all $j\in\{1,\ldots,N\}$). 
We introduce the Hilbert space $H(S)$, known as the linear Gaussian space induced by $u$ at $S$, 
as the closure of $U(S)$ in $L_2(\Omega)$ endowed with the norm $\|\,\cdot\,\|_{H(S)}$ induced by the $L_2(\Omega)$ inner product. 
That is, if  $g = \sum_{j=1}^{N}  \gamma_j  u(s_j)$
and $h = \sum_{k=1}^{N'} \gamma_k' u(s_k')$, for some $N,N'\in\mathbb{N}$, $s_j, s_{k}\in S$, and $\gamma_j, \gamma_k'\in\mathbb{R}$,
with $j=1,\ldots,N$ and $k=1,\ldots, N'$, then
$
(g,h)_{_{H(S)}} := \sum_{j=1}^{N}  \sum_{k=1}^{N'} \gamma_j \gamma_k' \pE\bigl[ u(s_j) u(s_k') \bigr].
$

Assuming that $u$ has a continuous covariance function $\rho:\Gamma\times\Gamma\to\mathbb{R}$, given $S\subset\Gamma$, we define the Cameron--Martin space $\mathcal{H}(S)$, which is isometrically isomorphic to $H(S)$, as
${\mathcal{H}(S) = \{h(s) = \pE(u(s)u): s\in \Gamma\hbox{ and } u\in H(S)\}}$,
with inner product $\langle h_1, h_2\rangle_{\mathcal{H}} = \pE(u_1 u_2),$
where $h_j(s) = \pE(u(s)u_j)$, $u_j \in H(S)$, and $j=1,2$.
		
\subsection{Equivalence of measures and parameter estimation}
Two probability measures $\mu$ and $\widetilde{\mu}$ on $L_2(\Gamma)$ are equivalent if for any Borel set $B$, $\mu(B)=0$ holds 
if and only if $\widetilde{\mu}(B)=0$. 
In contrast, if a Borel set $B$ exists such that $\mu(B)=0$ and $\widetilde{\mu}(B)=1$, the measures are orthogonal. 
Equivalence and orthogonality play a crucial role in studying asymptotic properties of Gaussian random fields. 
The Feldman--H\'ajek theorem \citep[][Theorem 2.25]{daPrato2014} provides necessary and sufficient conditions for equivalence; 
however, the conditions of this theorem are given in terms of the corresponding covariance operators and are generally difficult to verify. 
Nevertheless, the covariance operators for two Whittle--Mat\'ern fields diagonalize with respect to the eigenfunctions of the Kirchhoff-Laplacian, 
allowing us to derive the following result concerning the equivalence of measures. 
\begin{Proposition}\label{prop_measure}
Suppose that $\mu(\cdot; \kappa,\tau,\alpha)$ and $\mu(\cdot; \widetilde{\kappa},\widetilde{\tau},\widetilde{\alpha})$ are 
two Gaussian measures on $L_2(\Gamma)$ as defined above, 
with parameters $\kappa,\tau>0, \alpha>1/2$ and $ \widetilde{\kappa},\widetilde{\tau}>0, \widetilde{\alpha}>1/2$ respectively. 
Then $\mu(\cdot; \kappa,\tau,\alpha)$ and $\mu(\cdot; \widetilde{\kappa},\widetilde{\tau},\widetilde{\alpha})$ are equivalent if and only if 
$\alpha=\widetilde{\alpha}$ and $\tau=\widetilde{\tau}$.
\end{Proposition}
		
\begin{proof}
The two Gaussian measures can be written as $\mu(\cdot; \kappa,\tau,\alpha) = \pN(\cdot;0,\tau^{-2} L^{-\alpha})$ and 
$\mu(\cdot; \widetilde{\kappa},\widetilde{\tau},\widetilde{\alpha}) = \pN(\cdot;0,\widetilde{\tau}^{-2} \widetilde{L}^{-\widetilde{\alpha}})$, 
where $L = \kappa^2 - \Delta_\Gamma$ and $\widetilde{L} = \widetilde{\kappa}^2 - \Delta_\Gamma$.
Recall that $\{\hat{\lambda}_{i}\}$ and $\{\eig_i\}$ denote the eigenvalues and corresponding eigenvectors of the Kirchhoff-Laplacian and that 
$\{\eig_i\}$ forms an orthonormal basis of $L_2(\Gamma)$. 
Further, this is also an eigenbasis for both $L$ and $\widetilde{L}$, 
with corresponding eigenvalues $\lambda_j = \kappa^2 + \hat{\lambda}_j$ and  $\widetilde{\lambda}_j = \widetilde{\kappa}^2 + \hat{\lambda}_j$, 
respectively, for ${j\in\mathbb{N}}$. Define $\delta := \widetilde{\alpha}/\alpha$ and 
$c_j := \widetilde{\tau}^{2/\alpha}\tau^{-2/\alpha}\widetilde{\lambda}_j^{\delta}\lambda_j^{-1}$, $j\in\mathbb{N}$. 
Then, the asymptotic behavior of the eigenvalues $\hat{\lambda}_j$ in \eqref{eq:weyl} shows that $c_{-},c_{+}\in (0,\infty)$ exist such that 
${0<c_{-} < c_j < c_{+} < \infty}$ for all $j\in\mathbb{N}$ if, and only if, $\delta = \widetilde{\alpha}/\alpha = 1$.
In this case we have that $\lim_{j\rightarrow\infty} c_j = (\widetilde{\tau}/\tau)^{2/\alpha}$. 
Thus, the series $\sum_{j=1}^{\infty}(c_j-1)^2$ converges if, and only if, $\tau=\widetilde{\tau}$ and $\alpha = \widetilde{\alpha}$. 
The result then follows by \citet[][Corollary~3.1]{bk-measure}.
\end{proof}
		
Two Gaussian measures defined on the metric graph $\Gamma$ are either equivalent or orthogonal \citep[see, e.g.,][Theorem~2.7.2]{Bogachev1998}. 
Thus, whenever $\alpha\neq\widetilde{\alpha}$ or $\tau\neq\widetilde{\tau}$, the Gaussian measures are orthogonal. 
In contrast, if $\alpha=\widetilde{\alpha}$ and $\tau = \widetilde{\tau}$, the measures are equivalent even if $\kappa\neq\widetilde{\kappa}$.
This finding has crucial consequences for parameter estimation. 
It is well-known \citep{Zhang2004} that all parameters of Mat\'ern fields on bounded subsets of $\mathbb{R}^2$ cannot be estimated consistently under infill asymptotics,
but that one can estimate the so-called micro-ergodic parameter, corresponding to $\tau$ in our case. 
The following proposition shows that $\tau$ in the Whittle--Mat\'ern fields on metric graphs can also be estimated consistently. 
This result is the analogue of \citet[][Theorem~2]{Kaufman2013} for the metric graph setting. 
We recall that $u$ is the solution to \eqref{eq:Matern_spde} and that its law is given by $\mu = \pN\left(0, \tau^{-2}L^{-\alpha}\right)$.
\begin{Proposition}\label{prop:ml_consistency}
Suppose that $u_1,u_2,\ldots $ are observations of $u\sim \pN\left(0, \tau^{-2}L^{-\alpha}\right)$ at distinct locations $s_1,s_2,\ldots$  that accumulate at every point in 
$\Gamma$. Assume that $\alpha>\nicefrac{1}{2}$ is known and suppose that $0 < \kappa_L < \kappa < \kappa_U < \infty$. Let $(\tau_n, \kappa_n)$ 
denote the values of $\tau$ and $\kappa$ that maximize the likelihood $L(u_1,\ldots, u_n; \tau, \kappa)$ over 
$(\tau,\kappa)\in\mathbb{R}^+\times[\kappa_L,\kappa_U]$. Then:
\begin{enumerate}
	\item $\tau_n\rightarrow \tau$, $\mu$-a.s.
	\item $n^{1/2}(\tau_n^2 - \tau^2) \rightarrow \pN(0,2\tau^{4})$ in distribution. 
\end{enumerate}
\end{Proposition}
		
\begin{proof}
We start by proving (i).
\textbf{Step 1:} 
Let $f_{n,\kappa,\tau}$ be the probability density function of the vector $(u(s_1),\ldots, u(s_n))^\top$ under the measure $\mu(\cdot; \kappa, \tau, \alpha)$ 
and for $\mv{u}_n = (u_1,\ldots, u_n)^\top$ define ${\rho_n = \log f_{n,\tau,\kappa}(\mv{u}_n) - \log f_{n,\tau^*,\kappa^*}(\mv{u}_n)}$. 
First, suppose that the likelihood is evaluated with $\kappa = \kappa^*$ fixed and $\tau_n(\kappa^*)$ denotes the value of $\tau$ 
that maximizes the likelihood for this fixed value of $\kappa$. 
By Proposition~\ref{prop_measure}, the Gaussian measures $\mu(\cdot; \kappa,\tau,\alpha)$ and $\mu(\cdot; \kappa^*,\tau^*,\alpha)$ 
are orthogonal if $\tau^*\neq \tau$ and equivalent if $\tau^* = \tau$. 
In addition, by the continuity of the sample paths (Theorem~\ref{thm:regularity}), we can apply \cite[Theorem 1, p.100]{gikhmanskorohod}. 
Therefore, if $\tau^*\neq \tau$, $\rho_n \rightarrow -\infty$ as $n\rightarrow\infty$. 
If instead $\tau^* = \tau$, then $\rho_n \rightarrow \log C$ as $n\rightarrow\infty$, where $C$ is the Radon--Nikodym derivative 
$d \mu(u; \kappa,\tau,\alpha)/ d\mu(u; \kappa^*,\tau,\alpha)$. 
The result then follows along the same lines as in the proof of \citep[][Theorem 3]{Zhang2004}: 
It is sufficient to demonstrate that $\tau_n(\kappa^*) \rightarrow \tau^*$, $\mu(\cdot; \kappa^*,\tau^*,\alpha)$-a.s.
Therefore, it is therefore sufficient to establish that, for any $\epsilon>0$, 
an integer $N$ exists such that for $n>N$ and $|\tau_n(\kappa^*) - \tau^*|>\epsilon$, $\rho_n < \log C -1$. 
This result follows immediately from the fact that, for each $n$, the log-likelihood function $\tau\mapsto \log f_{n,\kappa,\tau}(\mv{u}_n)$ is strictly concave. 
	
\textbf{Step 2:} 
Next, we show that, if $\kappa_L < \kappa_1 < \kappa_2 < \kappa_U$, then $\tau_n(\kappa_2) \leq \tau_n(\kappa_1)$. 
To that end, let $\varrho_\kappa$ denote the covariance matrix of the Whittle--Mat\'ern field with $\tau=1$ and 
let $\mv{\Gamma}_{n,\kappa}$ denote the corresponding covariance matrix with elements $[\mv{\Gamma}_{n,\kappa}]_{i,j} = \varrho_\kappa(s_i,s_j)$. 
Observe that, by Proposition \ref{prp:strposdef}, for every $\kappa>0$, $\mv{\Gamma}_{n,\kappa}$ is strictly positive-definite.
Then, we have 
${\tau_n^2(\kappa) = \frac{n}{\mv{u}_n^\top\mv{\Gamma}_{n,\kappa}^{-1}\mv{u}_n}}$.
Thus, to show that $\tau_n^2(\kappa_2) \leq \tau_n^2(\kappa_1)$ it is enough to establish that for any possible realization of $\mv{u}_n$,
and any set of locations, we have 
$\mv{u}_n^\top \mv{\Gamma}_{n,\kappa_2}^{-1}\mv{u}_n - \mv{u}_n^\top \mv{\Gamma}_{n,\kappa_1}^{-1}\mv{u}_n \geq 0$, or, 
in other words, that the matrix 
$\mv{\Gamma}_{n,\kappa_2}^{-1} - \mv{\Gamma}_{n,\kappa_1}^{-1}$ is positive semi-definite. 
As $\mv{\Gamma}_{n,\kappa_2}$ and $\mv{\Gamma}_{n,\kappa_1}$ are strictly positive-definite, this holds if, and only if, 
$\mv{\Gamma}_{n,\kappa_1} - \mv{\Gamma}_{n,\kappa_2}$ is positive semi-definite, which holds for any possible realization of $\mv{u}_n$,
and any set of locations if, and only if, the function 
$r_\kappa = \varrho_{\kappa_1} - \varrho_{\kappa_2}$ is positive semi-definite. 
Using the expansion of the covariance functions $\varrho_{\kappa_1}$ and $\varrho_{\kappa_2}$ in \citet[Proposition 7]{BSW2022}, we have  
$$
	r_\kappa(s,t) = \sum_{i=1}^\infty \left(\frac1{(\kappa_1^2 + \hat{\lambda}_i)^\alpha} - \frac1{(\kappa_2^2 + \hat{\lambda}_i)^\alpha}\right)\varphi_i(s)\varphi_i(t).
$$
Next, because $\{\varphi_i\}$ are the eigenfunctions of $\varrho_\kappa$ it follows by \citet[Lemma 2.6]{Steinwart2012} 
and \citet[Theorem 10.4]{wendland} (for one direction and \citet[Theorem 10.4]{wendland} 
and \citet[p. 364-365]{Steinwart2012} for the other direction) that $r_\kappa$ is positive semi-definite if, and only if,
$(\kappa_1^2 + \hat{\lambda}_i)^{-\alpha} \geq (\kappa_2^2 + \hat{\lambda}_i)^{-\alpha}$ and
\begin{equation}\label{eq:conv_cond_like}
	\forall s\in \Gamma,\quad \sum_{i=1}^\infty \left(\frac1{(\kappa_1^2 + \hat{\lambda}_i)^\alpha} - \frac1{(\kappa_2^2 + \hat{\lambda}_i)^\alpha}\right) \varphi_i^2(s)< \infty.
\end{equation}
By \citet[Proposition 7]{BSW2022}, condition \eqref{eq:conv_cond_like} always holds. Thus, both conditions clearly hold if, and only if, $\kappa_1 \leq \kappa_2$; thus, $\tau_n^2(\kappa)$ is monotonically decreasing in $\kappa$. 
	
\textbf{Step 3:}
The result in (i) now follows by applying the result from Step 1 twice: one time for $\kappa^* = \kappa_L$ and another for $\kappa^* = \kappa_U$.  
Then, the result follows from using the fact that $\tau_n^2(\kappa)$ is monotonically decreasing in $\kappa$, which proves (i).
	
To prove (ii), set $\sigma^2_n=\frac{1}{\tau^2_n}$, $\sigma^2 = \frac{1}{\tau^2}$, and let $\kappa_n$ be an arbitrary sequence in $[\kappa_L,\kappa_U]$. 
Then 
\begin{align}
	\sqrt{n}(\sigma_n^2(\kappa_n) - \sigma^2) &= \sqrt{n}\left(\frac{\mv{u}_n^\top \mv{\Gamma}_{n,\kappa_n}^{-1}\mv{u}_n}{n} - \sigma^2\right) \notag \\
	&= \sqrt{n^{-1}}\left(\mv{u}_n^\top \mv{\Gamma}_{n,\kappa_n}^{-1}\mv{u}_n - \mv{u}_n^\top \mv{\Gamma}_{n,\kappa}^{-1}\mv{u}_n \right) + \sqrt{n}\left(\frac{\mv{u}_n^\top \mv{\Gamma}_{n,\kappa}^{-1}\mv{u}_n}{n} - \sigma^2\right). \label{eq:proof_2terms}
\end{align}
First, we establish that the first term in \eqref{eq:proof_2terms} converges to zero in probability by establishing that 
$\delta_n^{\mv{u}}(\kappa^*) = \mv{u}_n^\top \mv{\Gamma}_{n,\kappa^*}^{-1}\mv{u}_n- \mv{u}_n^\top \mv{\Gamma}_{n,\kappa}^{-1}\mv{u}_n$ 
is bounded in probability for any $\kappa^*$. 
As the measures $\mu(\cdot, \kappa, \tau, \alpha)$ and $\mu(\cdot, \kappa^*, \tau, \alpha)$  are equivalent, 
\citet[Lemma~2 and Expression~(2.9), p.76]{ibragimov_rozanov} show that the variance of $\delta_n^{\mv{u}}(\kappa^*)$ is uniformly bounded in $n$. 
In particular, by Chebyshev's inequality, the sequence $\delta_n^{\mv{u}}(\kappa^*)$ is bounded in probability. 
Next, because the variance of  $\delta^\mv{u}_n(\kappa^*)$  is a continuous function of $\kappa^*$ the result follows for any bounded sequence $(\kappa_n)$. 
Further, since $\mv{u}_n \sim \pN(0,\tau^{-2}\Gamma_{n,\kappa})$  it follows that
$$
	\sqrt{n}\Bigl(\frac{\mv{u}_n^\top \mv{\Gamma}_{n,\kappa}^{-1}\mv{u}_n}{n} - \sigma^2\Bigr) =  \sqrt{n}\sigma^2\Bigl(\frac{1}{n} \sum_{i=1}^n Z_i^2 - 1\Bigr) \rightarrow  \pN(0,2\sigma^4).
$$
Finally, $\sqrt{n}\left(n(\mv{u}_n^\top \mv{\Gamma}_{n,\kappa}^{-1}\mv{u}_n)^{-1} - \tau^2\right)  \rightarrow  \pN(0,2\tau^4)$
by the delta theorem \cite[][Theorem 2.5.2]{lehmann1999elements}.
\end{proof}
			
\subsection{Kriging prediction}
We now characterize the asymptotic properties of linear prediction for $u$ based on misspecified parameters. 
A sufficient criterion for asymptotic optimality is the equivalence of the corresponding Gaussian measures \citep{stein99}; 
thus, we obtain asymptotically optimal linear prediction as soon as $\widetilde{\alpha} = \alpha$ and $\widetilde{\tau} = \tau$ by Proposition~\ref{prop_measure}. 
However, the equivalence of measures is not necessary \citep{kb-kriging}, 
and we now establish that we only need  $\widetilde{\alpha} = \alpha$ to obtain asymptotic optimality.

To clarify the setup, recall the Gaussian linear space $H(\Gamma)$ and 
suppose  that we aim to predict $h\in H(\Gamma)$ based on a set of observations $\{y_{nj}\}_{j=1}^n$ of the process.
Then, the best linear predictor, with respect to the $H(\Gamma)$-norm, $h_n$, 
is the $H(\Gamma)$-orthogonal projection of $h$ onto the subspace $	H_n(\Gamma)  := 
\operatorname{span}\bigl\{y_{nj} \bigr\}_{j=1}^n$.
We now want to know what would happen if we replace $h_n$ with another linear predictor $\widetilde{h}_n$, 
computed based on a Whittle--Mat\'ern field with misspecified parameters. 
To answer this question, we assume that the  set of observations 
$\bigl\{ \{ y_{nj} \}_{j=1}^n : n\in\mathbb{N} \bigr\}$ 
yields $\mu$-consistent kriging prediction, that is,  
\begin{equation*}\label{eq:ass:Hn-dense}
	\lim\limits_{n\to\infty} \pE\bigl[ (h_n - h)^2 \bigr]
	= 
	\lim\limits_{n\to\infty} \| h_n - h \|_{\mathcal{H}}^2   
	= 0. 
\end{equation*}
Following \citet{kb-kriging}, we let $\mathcal{S}^\mu_{\mathrm{adm}}$ 
denote the set of all admissible sequences of observations which provide $\mu$-consistent kriging prediction 
and introduce the set $H_{-n}(\Gamma):=\bigl\{ h\in H(\Gamma):\pE\bigl[ (h_n - h)^2 \bigr]> 0 \bigr\}$. 
Observe that, in view of Proposition \ref{prp:strposdef}, we can apply the results in \cite{kb-kriging}.	
Then, we obtain the following result:
\begin{Proposition}\label{prop:A-kriging} 
Let $h_n, \widetilde{h}_n$ denote the best linear
predictors of $h\in H(\Gamma)$, with respect to the $H(\Gamma)$-norm, based on~$H_n(\Gamma)$~and the measures
$\mu(\cdot;\kappa,\tau,\alpha)$  and $\mu(\cdot;\widetilde{\kappa},\widetilde{\tau},\widetilde{\alpha})$, respectively. 
Let $\widetilde{\pE}(\cdot)$ denote the expectation under $\mu(\cdot;\widetilde{\kappa},\widetilde{\tau},\widetilde{\alpha})$. Then, 
any of the following statements,
\begin{align}
	\lim_{n\to\infty}\sup_{h\in H_{-n}(\Gamma)}
				\frac{
					\pE\bigl[ ( \widetilde{h}_n - h)^2 \bigr]
				}{
					\pE\bigl[ ( h_n - h)^2 \bigr]
				} 
				= 
				\lim_{n\to\infty}
				\sup_{h\in H_{-n}(\Gamma)}
				\frac{
					\widetilde{\pE}\bigl[ ( h_n - h)^2 \bigr]
				}{
					\widetilde{\pE}\bigl[ ( \widetilde{h}_n - h)^2 \bigr]
				} 
				= 1, 
				\label{eq:prop:A-kriging-1} 
\end{align}
\begin{align}
	\lim_{n\to\infty}
				\sup_{h\in H_{-n}(\Gamma)}
				\left|
				\frac{
					\widetilde{\pE}\bigl[ ( h_n - h)^2 \bigr]
				}{
					\pE\bigl[ ( h_n - h)^2 \bigr]
				} - c 
				\right| 
				= 
				\lim_{n\to\infty}
				\sup_{h\in H_{-n}(\Gamma)}
				\left|
				\frac{
					\pE\bigl[ ( \widetilde{h}_n - h)^2 \bigr]
				}{
					\widetilde{\pE}\bigl[ ( \widetilde{h}_n - h)^2 \bigr]
				} - \frac{1}{c} 
				\right| 
				= 0, 
				\label{eq:prop:A-kriging-2} 
\end{align} 
holds for some $c\in\mathbb{R}_+$ and all $\{H_n(\Gamma)\}_{n\in\mathbb{N}} \in \mathcal{S}^\mu_{\mathrm{adm}}$ 
if and only if $\alpha=\widetilde{\alpha}$. In this case, the constant $c$ in  \eqref{eq:prop:A-kriging-2} is $c = (\widetilde{\tau}/\tau)^{2/\alpha}$.
\end{Proposition} 
		
\begin{proof}
We use the same notation as in the proof of Proposition~\ref{prop_measure}; thus, we have 
$\mu(\cdot; \kappa,\tau,\alpha) = \pN(0,\tau^{-2} L^{-\alpha})$ and 
$\mu(\cdot; \widetilde{\kappa},\widetilde{\tau},\widetilde{\alpha}) = \pN(0,\widetilde{\tau}^{-2} \widetilde{L}^{-\widetilde{\alpha}})$, 
where $L = \kappa^2 - \Delta_\Gamma$ and $\widetilde{L} = \widetilde{\kappa}^2 - \Delta_\Gamma$. 
Both measures are centered; therefore, by \citet{kb-kriging}, the necessary and sufficient conditions for any of the statements in \eqref{eq:prop:A-kriging-1} or \eqref{eq:prop:A-kriging-2} are
\begin{enumerate}[i.]
	\item The operators $\tau^{-2} L^{-\alpha}$ and $\widetilde{\tau}^{-2} \widetilde{L}^{-\widetilde{\alpha}}$ have isomorphic Cameron--Martin spaces.
	\item There exists $c>0$ such that $(\nicefrac{\widetilde{\tau}}{\tau})^{2}L^{-\alpha/2}\widetilde{L}^{\widetilde{\alpha}}L^{-\alpha/2} -c^{-1} I$ is a compact operator on $L_2(\Gamma)$. 
\end{enumerate}
As in Proposition~\ref{prop_measure}, we define $\delta:=\widetilde{\alpha}/\alpha$ and 
$c_j := \widetilde{\tau}^{2/\alpha}\tau^{-2/\alpha}\widetilde{\lambda}_j^{\delta}\lambda_j^{-1}$. 
Then, the asymptotic behavior of the eigenvalues of the Kirchhoff--Laplacian in \eqref{eq:weyl} shows that constants 
$c_{-}, c_{+} \in (0,\infty)$ exists such that ${0<c_{-} < c_j < c_{+} < \infty}$ 
if and only if $\delta = \widetilde{\alpha}/\alpha = 1$. 
In this case, ${\lim_{j\rightarrow\infty} c_j = (\widetilde{\tau}/\tau)^{2/\alpha}}$. The result then follows from Corollary 3.1 in \citet{bk-measure}.
\end{proof}
		
\section{The Markov subclass}\label{sec:markov}
Although it is essential to be able to consider a general smoothness parameter $\alpha$ for the Whittle--Mat\'ern fields, 
we believe that the most important cases are where $\alpha\in\mathbb{N}$, corresponding to the case of a local precision operator $\mathcal{Q} = \tau^2L^{\alpha}$. 
The reason is that this results in Gaussian random fields with Markov properties \citep{BSW_Markov}. 
This section applies these Markov properties to derive two representations of the Whittle--Mat\'ern fields with  $\alpha\in\mathbb{N}$, 
facilitating computationally efficient and exact likelihood evaluations, spatial prediction, and simulation.
Section~\ref{sec:bridge} derives a ``bridge'' representation, where the process is represented as the sum of independent Gaussian processes on the edges, 
which are zero at the vertices, and  a ``low-rank'' Gaussian process defined in the vertices and interpolated to the edges. 
This representation completely characterizes the conditional finite-dimensional distributions of the fields, 
given the field and its derivatives evaluated at the graph vertices. 
The only quantity not explicitly characterized in the bridge representation is the joint distribution of the field and its derivatives at the vertices. 
To obtain this distribution, Section~\ref{sec:conditional} shows that one can define a set of independent Gaussian processes on the edges, 
with explicit covariance functions, such that the Whittle--Mat\'ern fields are obtained when conditioning on the Kirchhoff vertex conditions. 
This is referred to as the `conditional' representation and it can be used to evaluate the joint density of the process and its derivatives at the vertices. 
These two representations are used in the later sections to derive explicit and computationally efficient methods for likelihood evaluations and spatial prediction. 
	
\subsection{Bridge representation}\label{sec:bridge}	
We begin by introducing the following process, referred to as a Whittle--Mat\'ern bridge process.

\begin{Definition}\label{def:WMB}
Let $x_\alpha$ be a centered Gaussian process on an interval $[0, T]$, with covariance function \eqref{eq:matern_cov},  
where $d(x,y) = |x-y|$, $\nu=\alpha-\frac{1}{2}$ and $\alpha \in \mathbb{N}$. 
Then, the Whittle--Mat\'ern bridge process with parameters $(\kappa,\tau,\alpha)$ on the interval $[0,T]$, $x_{B,T,\alpha}$,  
with respect to $x_\alpha$ is $ {x_{B,T,\alpha}(t) = x_\alpha(t) | \{\mv{x}_{\alpha}(0)  =0, \mv{x}_{\alpha}(T) =0\}},$ 
where $ \mv{x}_{\alpha}(t)= [x_\alpha(t),x_\alpha^{(1)}(t),\ldots , x_\alpha^{(\alpha-1)}(t)]$, 
and the derivatives of $x_\alpha(\cdot)$ are taken weakly in the $L_2(\Omega)$ sense (see Appendix~\ref{app:proofs_edge}).
\end{Definition}

We consider the following operator, defined for a  sufficiently differentiable function, as follows: 
$$
	B^{\alpha} u =  \left[u(0),u^{(1)}(0),\ldots,u^{(\alpha-1)}(0),u(\ell_e),u^{(1)}(\ell_e),\ldots,u^{(\alpha-1)}(\ell_e) \right]^\top,
$$
where the derivatives are weak in the $L_2(\Omega)$ sense. 
The Whittle--Mat\'ern bridge process has the following properties:

\begin{Proposition}\label{prp:Whittle_Matern_bridge_prop}
Let $x_{B,T,\alpha}(\cdot)$ be a Whittle--Mat\'ern bridge process on $e=[0,\ell_e]$, then
$x_{B,T,\alpha}(\cdot)$
is $\alpha-1$ times weakly differentiable in the $L_2(\Omega)$ sense, its weak derivatives in $L_2(\Omega)$
are also weakly continuous 
in the $L_2(\Omega)$ sense, and 
$B^{\alpha}x_{B,T,\alpha}(\cdot) = \mv{0}$. Further, $x_{B,T,\alpha}(\cdot)$ has covariance function
\begin{equation}\label{eq:cov_func_whittle_matern_bridge}
	r_{B,\ell_e}(t_1,t_2) = \varrho_M(t_1-t_2) - 
		\begin{bmatrix}
			\mv{r}_1(t_1,0) & \mv{r}_1(t_1,\ell_e) 
		\end{bmatrix}
		\begin{bmatrix}
			\mv{r}(0,0) & \mv{r}(0,\ell_e) \\
			\mv{r}(\ell_e,0) & \mv{r}(\ell_e,\ell_e)
		\end{bmatrix}^{-1}
		\begin{bmatrix}
			\mv{r}_1(0,t_2) \\
			\mv{r}_1(\ell_e,t_2) 
		\end{bmatrix},
\end{equation} 		
where $t_1,t_2\in e$, $\mv{r}(s,t)$ is the matrix given by
\begin{equation}\label{eq:R_matrix_edge_repr}
	\mv{r} : \mathbb{R} \times \mathbb{R} \mapsto \mathbb{R}^{\alpha \times \alpha}, 
	\quad \mv{r}(t_1,t_2) = \left[ \frac{\pd^{i-1}}{\pd t_2^{i-1}}\frac{\pd^{j-1}}{\pd t_1^{j-1}}\varrho_M(t_1-t_2)\right]_{ij\in\{1,2,\ldots, \alpha\}},
\end{equation}
with $\varrho_M(\cdot)$ given in \eqref{eq:matern_cov}, and $\mv{r}_1(\cdot,\cdot)$ denotes the first row in $\mv{r}(\cdot,\cdot)$. 
\end{Proposition}
The proofs of Proposition \ref{prp:Whittle_Matern_bridge_prop}, Lemma \ref{lem:CM_edge_repr_bridge} and Theorems \ref{thm:ReprTheoremEdge_Refined} and \ref{cor:conditional_dists} are provided in Appendix~\ref{app:proofs_edge}. Further, the invertibility of the matrix whose inverse appears in \eqref{eq:cov_func_whittle_matern_bridge} follows from the proof of Lemma \ref{lem:sol_ODE_edge} in Appendix~\ref{app:proofs_edge}.

For a given edge $e\in\mathcal{E}$, we let $C_c^\infty(e)$ denote the set of infinitely differentiable functions with support compactly contained in the interior of $e$. 
The space $H_0^\alpha(e)$ is the completion of $C_c^\infty(e)$ with respect to the $\|\cdot\|_{H^\alpha(e)}$ norm, 
and $(\cdot,\cdot)_{\alpha,e}$ is the extension to $H_0^\alpha(e)\times H_0^\alpha(e)$ of the bilinear form
$(u,v)_{\alpha,e} = (u, L^\alpha v)_{L_2(e)}$, $u, v\in C^\infty_c(e)$.
\begin{Lemma}\label{lem:CM_edge_repr_bridge}
Let $V_{\alpha,0}(\cdot)$ be a Whittle--Mat\'ern bridge process with parameters $(\kappa,\tau,\alpha)$ on the interval $[0,\ell_e]$, ${0<\ell_e<\infty}$ and $\alpha\in\mathbb{N}$. The Cameron--Martin space associated with $V_{\alpha,0}(\cdot)$ is given by $(H^\alpha_0(e), (\cdot,\cdot)_{\alpha,e}).$
\end{Lemma}
	
The following representation of a Whittle--Mat\'ern field restricted to one edge of $\Gamma$ holds. 

\begin{Theorem}\label{thm:ReprTheoremEdge_Refined}
Let $\Gamma$ be a compact metric graph and let $u$ be a Whittle--Mat\'ern field on $\Gamma$
obtained as a solution to \eqref{eq:Matern_spde} for $\alpha\in\mathbb{N}$. Let $\mathcal{E}$ be
the set of edges of $\Gamma$. For any $e\in\mathcal{E}$, $e = [0,\ell_e]$, we obtain the following representation:
$$
	u_e(t) = v_{\alpha,0}(t) + \mv{S}_{e}(t) B^\alpha u_e,\quad t\in e,
$$
where 
\begin{equation}\label{eq:edge_repr_Solution}
	\mv{S}_{e}(t) = \begin{bmatrix}
		\mv{r}_1(t,0) & \mv{r}_1(t,\ell_e) 
	\end{bmatrix}
	\begin{bmatrix}
		\mv{r}(0,0) & \mv{r}(0,\ell_e) \\
		\mv{r}(\ell_e,0) & \mv{r}(\ell_e,\ell_e)
	\end{bmatrix}^{-1},
\end{equation}
$\mv{r}(s,t)$ is the matrix given by \eqref{eq:R_matrix_edge_repr},
and $\mv{r}_1(\cdot,\cdot)$ denotes the first row in $\mv{r}(\cdot,\cdot)$. 
Finally, $v_{\alpha,0}(\cdot)$ is a Whittle--Mat\'ern bridge process on $e$, which is independent of $B^\alpha u_e$. 
\end{Theorem}

\begin{Remark}
Theorem \ref{thm:ReprTheoremEdge_Refined} provides a substantial refinement of \citet[Theorem~9]{BSW_Markov} for the case of Whittle--Mat\'ern fields. 
Instead of having a description of the process $v_{\alpha,0}(\cdot)$ through its Cameron--Martin space, 
in Theorem \ref{thm:ReprTheoremEdge_Refined}, we completely identify $v_{\alpha,0}(\cdot)$ as a Whittle--Mat\'ern bridge process 
and provide explicit expressions for the functions multiplying the boundary term $B^\alpha u_e$.
\end{Remark}

The above representation indicates that given a Whittle--Mat\'ern field $u$ on a general compact metric graph $\Gamma$, 
then for any edge $e=[0,\ell_e]$, the conditional distribution of the field given the boundary data $B^\alpha u$ only depends on $\alpha, \kappa, \tau$ and $\ell_e$. 
In particular, it does not depend on the graph geometry. 
Further, this representation can also be used to obtain the conditional distribution of $u_e(t_1),\ldots,u_e(t_n)$ given $B^\alpha u$, 
where $t_1,\ldots,t_n\in [0,l_e]$, $n\in \mathbb{N}$: 

\begin{Theorem}\label{cor:conditional_dists}
Let $u$ be a solution to \eqref{eq:Matern_spde} with $\alpha\in\mathbb{N}$ and define a centered Gaussian process $u_M$ on $e\in\mathcal{E}$ 
with a Mat\'ern covariance function with parameters $(\kappa,\tau,\alpha)$. 
Fix ${\underline{\mv{u}}_0, \bar{\mv{u}}_0 \in \mathbb{R}^{\alpha}}$ and let $\mv{u}_0 = (\underline{\mv{u}}_0,\bar{\mv{u}}_0)$. 
If $\deg(\underline{e}) = 1$, we require that $\underline{\mv{u}}_0= (u_{00}, 0, u_{01}, 0, \ldots)$, with $u_{0j}\in\mathbb{R}$, $j=0,\ldots,\floor{\nicefrac{\alpha-1}{2}}$. Similarly, if $\deg(\bar{e}) = 1$, we require that every second element in $\bar{\mv{u}}_0$ is zero.
Then, the two conditional processes $u_e | \left\{ B^{\alpha}u_e =\mv{u}_0  \right\}$ and $u_M |\left\{ B^{\alpha}u_M =\mv{u}_0  \right\}$
have the same finite-dimensional distributions. 
\end{Theorem}

Finally, as a corollary of Theorem \ref{thm:ReprTheoremEdge_Refined}, we have the following bridge representation.

\begin{Corollary}\label{cor:bridge_representation}
Let $\Gamma$ be a compact metric graph and $u$ be a Whittle--Mat\'ern field obtained as solution to \eqref{eq:Matern_spde}, with $\alpha\in\mathbb{N}$. 
Then, $u$ has the following representation:	
\begin{equation}\label{eq:bridge_representation}
	u(s) = u_\Gamma(s)  + \sum_{e\in \mathcal{E}}\tilde u_{B,\ell_e,\alpha}(s), \quad s\in\Gamma,
\end{equation}
where $\widetilde{u}_{B, \ell_e,\alpha}$, $e\in\mathcal{E}$, are independent Mat\'ern bridge processes defined on each edge and extended by zero outside the edge. 
Moreover, $u_\Gamma(s) =  \sum_{e\in\mathcal{E}}\mv{S}_{e}(s) \mv{D}_e \mv{U}$, where 
\begin{align}\label{eq:dens_u2}
   \mv{U} =  [{\mv{u}}(\underline{e}_1)^\top,
   {\mv{u}}(\bar{e}_1)^\top,
   {\mv{u}}(\underline{e}_2)^\top,
   {\mv{u}}(\bar{e}_2)^\top,
   \ldots,
	{\mv{u}}(\underline{e}_{|\mathcal{E}|})^\top,
   {\mv{u}}(\bar{e}_{|\mathcal{E}|})^\top]^\top,
\end{align}
$$	
	\mv{u}(s) =  [{u}(s),{u}'(s),{u}''(s),\ldots, {u}^{(\alpha-1)}(s)]^\top = \sum_{e \in \mathcal{E}} \mathbb{I}\left(s \in e\right) {\mv{u}}_e(s),
$$
 $\mv{D}_{e}$ is the matrix that maps $\mv{U}$ to $(\mv{u}(\underline{e}),\, \mv{u}(\bar{e}))^\top$, 
 and $\mv{S}_e(\cdot), e\in\mathcal{E}$, is given in Theorem \ref{thm:ReprTheoremEdge_Refined}.
\end{Corollary}

\begin{proof}
Consider the Whittle--Mat\'ern field $u(s)$ on a metric graph $\Gamma$. 
By \cite[Theorem 5 and Remark 3]{BSW_Markov}, $u(s_1)$ and $u(s_2)$ are conditionally independent given $\mv{U}$ 
if $s_1$ and $s_2$ are locations on different edges. 
By Theorem~\ref{thm:ReprTheoremEdge_Refined}, we can express $u(s)|_e$ as 
$$
	u(s)  = \mv{S}_{e}(t) B^\alpha u_e + u_{B,\ell_e,\alpha}(t), \quad s\in e,\quad s=(e,t),
$$
where $u_{B,\ell_e,\alpha}$ is a Mat\'ern bridge process on $[0, \ell_e]$, independent of $\mv{U}$.

 By extending $\mv{S}_{e}(\cdot)$ by zero outside the edge $e$, we can then express the Whittle--Mat\'ern field $u(s)$ as follows:
\begin{align*}
u(s)  =   \sum_{e\in\mathcal{E}}\mv{S}_{e}(s) \begin{bmatrix}
\mv{u}(\underline{e}) \\
\mv{u}(\bar{e})
\end{bmatrix}   + \sum_{e\in \mathcal{E}}\tilde u_{B,\ell_e,\alpha}(s) =  \sum_{e\in\mathcal{E}}\mv{S}_{e}(s) \mv{D}_e\mv{U}  + \sum_{e\in \mathcal{E}} u_{B,\ell_e,\alpha}(s),
\end{align*}
where $u_{B, \ell_e}$ are independent Mat\'ern bridge processes defined on each edge and extended by zero outside the edge. 
Thus, by defining the low-rank process $u_\Gamma(s) =  \sum_{e\in\mathcal{E}}\mv{S}_{e}(s) \mv{D}_e \mv{U}$, 
we arrive at the bridge representation \eqref{eq:bridge_representation} for the Whittle--Mat\'ern field.
\end{proof}

With this representation, the only quantity not explicitly given is the joint distribution of $\mv{U}$. 
This distribution can be obtained through the conditional representation in Theorem~\ref{thm:representation} in the following subsection. 

\subsection{Conditional representation}\label{sec:conditional}
We will now establish that one can define independent Gaussian processes on the graph edges, which after conditioning on 
Kirchhoff constraints at the vertices results in a process that is the solution to \eqref{eq:Matern_spde}. 
To define these processes, we require the multivariate covariance function from the following proposition, 
whose proof is given in  Appendix~\ref{app:proofs_conditional}:
	
\begin{Proposition}\label{prop:multivariate_covariante}
Let $\mv{r}(\cdot,\cdot)$ be given by \eqref{eq:R_matrix_edge_repr} with $\alpha\in\mathbb{N}$.
Then, for $\ell >0$,
\begin{align}\label{eq:covmod}
	\tilde{\mv{r}}_{\ell}\left(t_1,t_2\right) = \mv{r}(t_1,t_2) +
	\begin{bmatrix}\mv{r}(t_1, 0) & \mv{r}(t_1,\ell)\end{bmatrix} 
	\begin{bmatrix}
		\mv{r}(0,0)  & 	-\mv{r}(0,\ell) \\
		-\mv{r}(\ell,0)  & 	\mv{r}(0,0)
	\end{bmatrix}^{-1} 
	\begin{bmatrix}\mv{r}(t_2,0)\\
		\mv{r}(t_2,\ell)
	\end{bmatrix}
\end{align}
is a multivariate covariance function on the interval $[0, \ell]$.
\end{Proposition}

The covariance function \eqref{eq:R_matrix_edge_repr} is the covariance function of $[x(t),x'(t),x''(t),\ldots, x^{(\alpha-1)}(t)]$ if $x$ 
is a centered Gaussian process on $\mathbb{R}$ with a Mat\'ern covariance function with $\nu=\alpha - \nicefrac{1}{2}$. Further, the invertibility of the matrix whose inverse appears in \eqref{eq:covmod} is showed in the proof of Theorem \ref{thm:CondDens} (together with Lemma \ref{lem:prop_bdlessproc} to apply it for the covariance \eqref{eq:R_matrix_edge_repr}) in Appendix~\ref{app:bdlessaux}.

Proposition \ref{prop:multivariate_covariante} allows defining a new centered Gaussian stochastic process on an interval $[0,\ell]$, 
which we call the boundaryless Whittle--Mat\'ern process. 

\begin{Definition}\label{def:cov_based_bdlessWM}
Let $\ell>0, \alpha\in\mathbb{N}$ and $\widetilde{r}_\ell(t_1,t_2) = [\tilde{\mv{r}}_{\ell}\left(t_1,t_2\right)]_{1,1}$, $t_1,t_2\in [0,\ell]$, 
that is, $\widetilde{r}_\ell(\cdot,\cdot)$ is the first entry of the matrix $\tilde{\mv{r}}_{\ell}\left(t_1,t_2\right)$. 
We define the boundaryless Whittle--Mat\'ern process on $[0,\ell]$ with parameters $(\kappa,\tau,\alpha)$ as a centered Gaussian process 
with covariance function $\widetilde{r}_\ell(\cdot,\cdot)$.
\end{Definition}

\begin{Remark}\label{rem:cov_bdless}
By Proposition \ref{prp:bdless_thm_cov_123} in Appendix \ref{app:bdlessaux}, for a given $\alpha\in\mathbb{N}$, 
the multivariate covariance function $\tilde{\mv{r}}_{\ell}(\cdot,\cdot)$ given in \eqref{eq:covmod} is the multivariate covariance function 
of the multivariate process $\widetilde{\mv{u}}(\cdot) =  [\widetilde{u}(\cdot),\widetilde{u}'(\cdot),\widetilde{u}''(\cdot),\ldots, \widetilde{u}^{(\alpha-1)}(\cdot)]^\top$, 
where $\widetilde{u}(\cdot)$ is the boundaryless Whittle--Mat\'ern process and the derivatives are weak derivatives in the $L_2(\Omega)$ sense.
\end{Remark}

We introduce the following notation for the subset of $\Omega$ corresponding to the realizations of a field $x:\Gamma\times\Omega\to\mathbb{R}$ 
such that the Kirchhoff conditions implied by taking powers of the differential operator $L$ 
(see \cite[Propositions 3 and 4]{BSW2022} and Theorem \ref{thm:regularity}.\ref{thm:regularity:item:Kirchhoff}) hold for $x$:
\begin{equation*}\label{eq:Kirchhoff_alpha}
	\begin{split}
	\mathcal{K}_\alpha(x)	= \{\omega\in\Omega:&\mbox{$ \forall v\in \mathcal{V}$ and each pair $e,\widetilde{e}\in\mathcal{E}_v$, $x_e^{(2k)}(v,\omega)=x_{\widetilde{e}}^{(2k)}(v,\omega),$ and} \\ 
	&\,\,\mbox{$\sum_{e \in\mathcal{E}_v} \partial_e x_e^{2k+1}(v,\omega ) = 0$, $k = 0,\ldots, \ceil{\alpha - \nicefrac{1}{2}}  -1$} \},
	\end{split}
\end{equation*}
where the derivatives are weak derivatives in the $L_2(\Omega)$ sense.
We can now state the conditional representation of Whittle--Mat\'ern fields. 

\begin{Theorem}\label{thm:representation}
Let $\alpha\in\mathbb{N}$ and $\{\widetilde{u}_e(\cdot): e\in\mathcal{E}\}$ be a family of independent boundaryless Whittle--Mat\'ern processes 
with parameters $(\kappa,\tau,\alpha)$, where $\widetilde{u}_e(\cdot)$ is defined on $[0,\ell_e]$, $e\in\mathcal{E}$.
Further, for any $s\in\Gamma$, set 
\begin{equation}\label{eq:utilde_proc}
	\widetilde{\mv{u}}(s) =  [\widetilde{u}(s),\widetilde{u}'(s),\widetilde{u}''(s),\ldots, \widetilde{u}^{(\alpha-1)}(s)]^\top = \sum_{e \in \mathcal{E}} \mathbb{I}\left(s \in e\right) \widetilde{\mv{u}}_e(s),
\end{equation}
and define
$\mv{u}(s) = \widetilde{\mv{u}}(s) |\mathcal{K}_\alpha(\widetilde{u})$. Then $u(\cdot)$, given by the first entry of $\mv{u}(\cdot)$, is a solution to \eqref{eq:Matern_spde}. 
\end{Theorem}
The proof of this theorem, which is one of our main results, is given in Appendix~\ref{app:proofs_conditional}. 
As an illustration of the conditioning procedure, consider the graph in Figure~\ref{fig:split_graph}. Then, for $\alpha=1$ and $\alpha=2$, we have 
\begin{align*}
	\mathcal{K}_1(\widetilde{u}) &= 
	\{\omega \in \Omega : \widetilde{u}_1(\underline{e}_1,\omega) = \widetilde{u}_2(\underline{e}_2,\omega) = \widetilde{u}_3(\underline{e}_3,\omega),\,\, \widetilde{u}_1(\bar{e}_1,\omega) = \widetilde{u}_2(\bar{e}_2,\omega) = \widetilde{u}_3(\bar{e}_3,\omega) \},\\ 
	\mathcal{K}_2(\widetilde{u}) &= \mathcal{K}_1(\widetilde{u}) \cap 
	\left\{\omega \in \Omega : \sum_{i=1}^3\widetilde{u}_i^{(1)}(\underline{e}_i,\omega) =0,\,\, -\sum_{i=1}^3\widetilde{u}_i^{(1)}(\bar{e}_i,\omega) =0 \right\}. 
\end{align*}

\begin{figure}[t]
	\includegraphics[height=0.35\linewidth]{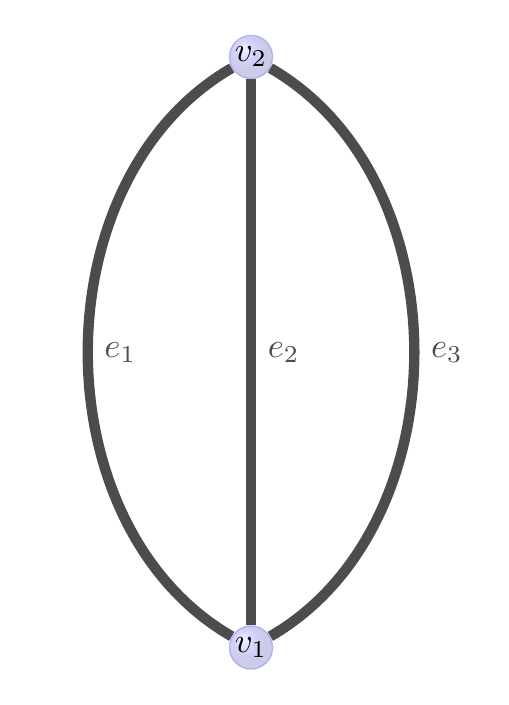}
	\includegraphics[height=0.35\linewidth]{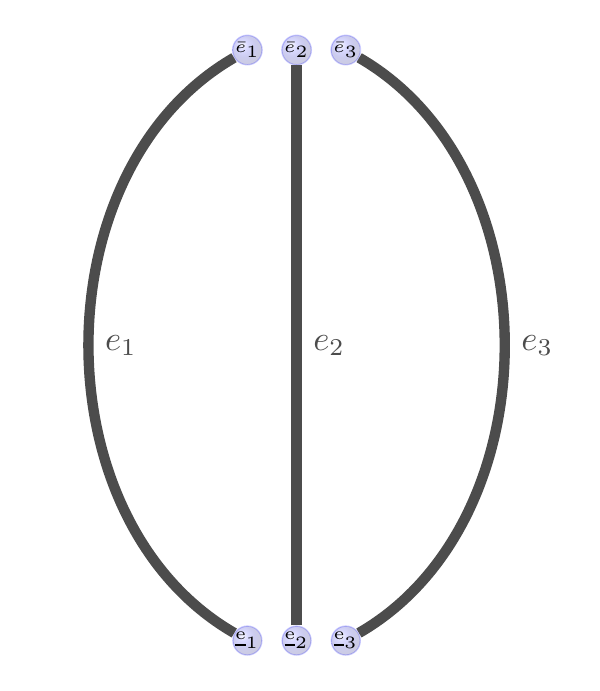}
	\caption{A graph $\Gamma$ (left) and the split used to define the independent edge processes (right).}
	\label{fig:split_graph}
\end{figure}
	
For a metric graph with two vertices and a single edge (i.e., an interval), we obtain the following result:
\begin{example}\label{ex:exp_line}
Consider the stochastic differential equation $(\kappa^2 - \Delta)^{\alpha/2} \tau u = \mathcal{W}$
on the interval $[0,\ell_e]$, where $\Delta$ is the Neumann--Laplacian. If $\alpha=1$, then the covariance function of $u$ is given by 
\begin{align}
	\varrho(t_1,t_2)  &= r(t_1-t_2) + \begin{bmatrix}
				r(t_1) &
				r(t_1-\ell_e) 
	\end{bmatrix}
	\begin{bmatrix}
		r(0) & -r(\ell_e) \\
		-r(\ell_e) & r(0)
	\end{bmatrix}^{-1}
	\begin{bmatrix}
		r(t_2) \\
		r(t_2-\ell_e) 
	\end{bmatrix}\nonumber\\
	&= 
	\frac1{2\kappa \tau^2\sinh(\kappa \ell_e)}(\cosh(\kappa(\ell_e-|t_1-t_2|)) + \cosh(\kappa(t_1+t_2-\ell_e))),\label{eq:exp2}
\end{align}
where $r(h) = (2\kappa\tau^2)^{-1}\exp(-\kappa |h|)$ and $t_1,t_2\in [0,\ell_e]$. 
Here, Equation~\eqref{eq:exp2} follows by simple trigonometric identities. 
If $\alpha=2$,  the covariance function of $u$ is given by
\begin{align*}
	\varrho(t_1,t_2) =& r(|h|) + \frac{r(h)+r(-h)+e^{2\kappa \ell_e}r(v)+ r(-v) }{2e^{\kappa \ell_e} \sinh(\kappa \ell_e) } 
	+
	\frac{\ell_e\cosh(\kappa t_1)\cosh(\kappa t_2)}{2\kappa^2\tau^2\sinh(\kappa \ell_e)^2 }
\end{align*}
where $h=t_1-t_2, v= t_1+t_2$ and $r(h)=(4\kappa^3\tau^2)^{-1}\left(1 + \kappa h\right)\exp(-\kappa h)$.
\end{example}
	For a metric graph with a single vertex and a single edge (i.e., a circle), we instead obtain the following result.
\begin{example}\label{ex:exp_circle}
Consider the stochastic differential equation  $(\kappa^2 - \Delta)^{\alpha/2} \tau u = \mathcal{W}$ on a circle with perimeter $\ell_e$. If $\alpha=1$, then the covariance function of $u$ is given by
\begin{align*}
	\varrho(t_1,t_2) &= \frac1{2\kappa\tau^2}e^{-\kappa |t_1-t_2|} + \frac{e^{-\kappa T}}{\kappa\tau^2(1-e^{-\kappa \ell_e})}\cosh(\kappa|t_1-t_2|)
	= \frac{\cosh(\kappa(|t_1-t_2|-\ell_e/2))}{2\kappa\tau^2\sinh(\kappa \ell_e/2)}.
\end{align*}
If $\alpha  = 2$, we instead obtain the following covariance function:
\begin{align*}
	 \varrho(t_1,t_2) =& \frac1{4\kappa^3\tau^2\sinh(\kappa \ell_e/2)}\left(\left[1+\frac{\kappa \ell_e}{2}\coth\left(\frac{\kappa \ell_e}{2}\right)\right]\cosh(w) + w\sinh(w)\right),
\end{align*}
where $w = \kappa(|t_1-t_2|-\ell_e/2)$.
\end{example}

\section{Exact log-likelihood evaluations for the Markov case}\label{sec:inference}
In this section, we show how to evaluate the log-likelihood of Whittle--Mat\'ern fields on metric graphs when sampled at locations on the graph. 
We consider the cases with direct observations of the random field and with observations under Gaussian measurement noise. 
That is, we assume that observations $\mv{y}= [y_1,\ldots, y_n]$ are available at locations $s_1,\ldots,s_n \in \Gamma$, 
which in the first setting are ${y_i = u(s_i)}$ and in the second are as follows:
\begin{equation}\label{eq:likelihood}
y_i  |u(\cdot)\sim \pN\left( u(s_i), \sigma^2\right).
\end{equation}
Before presenting the general likelihood formulas, we show how to evaluate the density of the process at the graph vertices. 

\subsection{Distribution of the process in the vertices}\label{sec:Uvertex}
The advantage of the conditional representation is that we have an explicit covariance function for each independent Gaussian process $\mv{u}_e$. 
Therefore, this facilitates exact evaluation of finite-dimensional distributions because it is easy to compute conditional distributions for Gaussian processes. 
For this approach, the following property, proved in Appendix~\ref{app:bdlessaux}, of the modified covariance $\tilde{\mv{r}}_{\ell}\left(\cdot,\cdot\right)$ is critical. 
	
\begin{Proposition}\label{cor:precfuncconddens}
Let $\alpha\in\mathbb{N}$ and define $\mv{u}$ and $\widetilde{\mv{u}}$ as two centered Gaussian processes on $[0,\ell]$ with covariance functions $\mv{r}$ defined in 
\eqref{eq:R_matrix_edge_repr}
and $\widetilde{\mv{r}}_\ell$ defined in \eqref{eq:covmod}, respectively. Then,  the precision matrix of $\left[\tilde{\mv{u}}(0),\tilde{\mv{u}}(\ell)\right]$ is 
\begin{equation}\label{eq:Qedge}
	\tilde{\mv{Q}}_e  = \mv{Q}_e - \frac{1}{2}\begin{bmatrix}
		\mv{r}(0,0)^{-1} & \mv{0} \\
		\mv{0} & \mv{r}(0,0)^{-1} 
	\end{bmatrix},
\end{equation}
where $\mv{Q}_e $ is the precision matrix of  $\left[{ \mv{u} }(0),{ \mv{u} }(\ell) \right]$.
\end{Proposition}
	
Observe that, by Lemma \ref{lem:prop_bdlessproc} in Appendix~\ref{app:bdlessaux}, $\mv{r}(\cdot,\cdot)$ is strictly positive-definite, in particular, $\mv{r}(0,0)$ is a strictly positive-definite matrix, so it is invertible. See Definition \ref{def:str_pos_def_mat} in Appendix~\ref{app:bdlessaux} for the definition of strictly positive-definite matrix-valued functions.
	
\begin{example}
In the above proposition, for $\alpha = 1$, we have 
\begin{equation}
	\label{eq:precQexp}
	\tilde{\mv{Q}}_e = \frac{\kappa\tau^2}{e^{2\kT} -1}
	\begin{bmatrix}
		e^{2\kT}+1 &&  -2 e^{\kT}\\
		-2 e^{\kT} &&  e^{2\kT}+1
	\end{bmatrix}.
\end{equation}	
For $\alpha=2$, let $c_1 = 1-2 (\kT)^2$, $c_2 = e^{2 \kT} (2+4 (\kT)^2)$, and $c_3 = \kT e^{2\kT}$. Then, 
$$		
	\tilde{\mv{Q}}_e = \frac{2\kappa\tau^2}{c_1^2- c_2+e^{4\kT}}
	\begin{bmatrix}
		\kappa^2(e^{4 \kT} -c_1^2+ 4  c_3) 	&&  4 \kappa c_3 && q_1 								&&  q_2\\
		4 \kappa c_3 						&& q_3 			&& q_2 									&& q_4\\
		q_1 								&& q_2 			&& \kappa^2(e^{4 \kT} -c_1^2+ 4  c_3) 	&& 4 \kappa c_3 \\
		q_2 								&& q_4 			&& 4 \kappa c_3 						&& q_3
	\end{bmatrix},
$$
where 
$q_{1} = 2\kappa^2e^{\kT} \left((1-\kappa\ell) c_1 - c_3 - e^{2 \kT} \right),$
$q_{2}  = -2 \kappa e^{\kT} \left(\kappa\ell c_1+c_3\right),$
$q_{3} = -\kappa ^2 c_1^2+ (\kappa^2-1)c_2-4 c_3 -\left(\kappa ^2-2\right) e^{4 \kT},$ and
$q_{4} = 2 e^{\kT} \left( c_1(\kappa\ell+1)+ c_3 -e^{2 \kT}\right)$.		
\end{example}

Next, we let $\alpha\in\mathbb{N}$ and define the vector 
\begin{align*}
	\widetilde{\mv{U}} = [\widetilde{\mv{u}}_1(\underline{e}_1)^\top,\widetilde{\mv{u}}_1(\bar{e}_1)^\top,\widetilde{\mv{u}}_2(\underline{e}_2)^\top,\widetilde{\mv{u}}_2(\bar{e}_2)^\top,\ldots, \widetilde{\mv{u}}_{|\mathcal{E}|}(\underline{e}_{|\mathcal{E}|})^\top,\widetilde{\mv{u}}_{|\mathcal{E}|}(\bar{e}_{|\mathcal{E}|})^\top]^\top
\end{align*}
of the process $\widetilde{\mv{u}}(s)$ from \eqref{eq:utilde_proc} evaluated in the endpoints of each of the $|\mathcal{E}|$ edges in $\Gamma$. 
We let $\mv{U}$ denote the corresponding vector for $\mv{u}(s) = \widetilde{\mv{u}}(s) |\mathcal{K}_{\alpha}(\widetilde{u})$. 
The processes $\widetilde{\mv{u}}_e$, $e\in\mathcal{E}$, are mutually independent; thus, we have 
$\widetilde{\mv{U}} \sim \pN\left(\mv{0},\widetilde{\mv{Q}}^{-1}\right)$,
where $\widetilde{\mv{Q}}= diag(\{\mv{Q}_e\}_{e\in\mathcal{E}})$, and $\mv{Q}_e$  is the precision matrix of $\{\widetilde{\mv{u}}_e(\underline{e}),\widetilde{\mv{u}}_e(\bar{e})\}$ which is given by \eqref{eq:Qedge}.
The Kirchhoff vertex conditions are linear on $\widetilde{\mv{U}}$; therefore, a $k \times 2\alpha|\mathcal{E}|$ matrix $\mv{K}$ exists 
such that the conditions can be written as $\mv{K}\widetilde{\mv{U}} = \mv{0}$,
where the number of constraints, $k$, is $\frac{\alpha}{2}\sum_{v\in\mathcal{V}} \deg(v)$ if $\alpha$ is even and $ \lceil\frac{\alpha}{2}\rceil \sum_{v\in\mathcal{V}} (\deg(v)-1)$ if $\alpha$ is odd. Thus,  $\mv{U}$ in  \eqref{eq:dens_u2} is
\begin{equation}
	\label{eq:dens_u}
	\mv{U}= \widetilde{\mv{U}} |  \left\{\mv{K} \widetilde{\mv{U}} = \mv{0}\right\},
\end{equation}
which is a GMRF under linear constraints. 
Hence, we arrive at the following corollary providing the distribution of $u(s)$ evaluated 
at the vertices. The result follows from standard formulas for GMRFs under linear constraints \citep[see, e.g.,][]{bolin2021efficient}.
\begin{Corollary}\label{cor:Uv_conditional}
Let $\mv{A}$ be a $|\mathcal{V}| \times 2\alpha|\mathcal{E}|$ matrix such that 
\begin{equation}\label{eq:Uv}
	\mv{U}_v  = (u(v_1), u(v_2), \ldots, u(v_{|\mathcal{V}|}))^\top = \mv{A}\mv{U}.
\end{equation}
That is, $\mv{A}$ acts on $\mv{U}$ by selecting the components corresponding to the solution of \eqref{eq:Matern_spde} evaluated at the vertices.
Then, $\mv{U}_v \sim \pN(\mv{0}, \mv{\Sigma})$, where 
$$
\mv{\Sigma} = \mv{A}\left(\widetilde{\mv{Q}}^{-1} - \widetilde{\mv{Q}}^{-1}\mv{K}^\top (\mv{K}\widetilde{\mv{Q}}^{-1}\mv{K}^\top)^{-1}\mv{K}\widetilde{\mv{Q}}^{-1}\right)\mv{A}^\top.
$$
\end{Corollary}
	
\begin{example}
For the graph in Figure~\ref{fig:split_graph}, for $\alpha=1$, we have $\mv{A} = \begin{bmatrix}\mv{I}& \mv{0} & \mv{0}\end{bmatrix}$ and
$$
	\mv{K} = \begin{bmatrix}
		\mv{I} & -\mv{I} & \mv{0} \\
		\mv{0} & \mv{I} & - \mv{I}
	\end{bmatrix},
$$
where $\mv{I}$ denotes a $2\times 2$ identity matrix and $\mv{0}$ a $2\times 2$ matrix with zero elements.
Note that $\mv{A}$ is not unique because  $\mv{A} = 
\begin{bmatrix}
 \mv{0}& \mv{I}& \mv{0}
\end{bmatrix}$
or
$\mv{A} = 
\begin{bmatrix}
\mv{0} & \mv{0} & \mv{I}
\end{bmatrix}$ would provide the same result.
For $\alpha=2$, we instead have
\setcounter{MaxMatrixCols}{20}
$$
\mv{K} = \begin{bmatrix}
	1 & 0 & 0 & 0  & -1 & 0  & 0 & 0 & 0 & 0 & 0 & 0\\
	0 & 0 & 1    & 0    & 0  & 0   & -1   & 0  & 0    & 0   & 0   & 0\\
	0 & 0 & 0   & 0    & 1   & 0   & 0   & 0   & -1   & 0   & 0   & 0\\
	0 & 0 & 0   & 0    & 0   & 0   & 1  & 0    & 0    & 0   & -1  & 0\\
	0 & 1 & 0   & 0    & 0   & 1   & 0  & 0   & 0    & 1    & 0   & 0\\
	0 & 0 & 0   &- 1    & 0   & 0   & 0  & -1    & 0    & 0    & 0  & -1
\end{bmatrix}, \quad \text{and} \quad
\mv{A} = \begin{bmatrix}
	1 & 0 & 0 & 0  & 0 & 0  & 0 & 0 & 0 & 0 & 0 & 0\\
	0 & 0 & 1    & 0    & 0  & 0   & 0   & 0  & 0    & 0   & 0   & 0
\end{bmatrix}.
$$		
\end{example}
	
If $\alpha = 1$, the conditioning in \eqref{eq:dens_u} only enforces continuity, and we can derive an explicit expression of the precision matrix of the process at the vertices in the graph.
	
\begin{Corollary}\label{cor:exp_Q}
For $\alpha=1$, $\mv{U}_v$ in \eqref{eq:Uv} satisfies $\mv{U}_v \sim \pN \left( \mv{0}, \mv{Q}^{-1}\right)$, where 
\begin{equation}\label{eq:expprec}
	Q_{ij} =  2\kappa\tau^2\cdot \begin{cases}
	\sum\limits_{e \in \mathcal{E}_{v_i} } 
	\left(	\frac{1}{2} + \frac{e^{-2\kappa \ell_e}}{1-e^{-2\kappa \ell_e}}  \right) \mathbb{I}\left(\bar{e}\neq \underline{e}\right)+	\tanh(\kappa \frac{l_e}{2} )  \mathbb{I}\left(\bar{e}= \underline{e}\right) & \text{if $i=j$},\\
	\sum\limits_{e \in  \mathcal{E}_{v_i}  \cap  \mathcal{E}_{v_j} } -\frac{e^{-\kappa \ell_e}}{1-e^{-2\kappa \ell_e}} & \text{if $i\neq j$.}
	\end{cases}
\end{equation}
\end{Corollary}
\begin{proof}
Without loss of generality, assume that $\tau^2=\frac{1}{2\kappa}$.  
For each circular edge $e^*$, with $\bar{e}^* = \underline{e}^*$, 
add a vertex $v^*$ at an arbitrary location $s^*$ on the edge $e^*$ and let $\Gamma^*$ denote this extended graph. Recall that these additions do not change the distribution of the process. Let $n_v^*$  and $n_e^*$ denote the number of vertices and edges in $\Gamma^*$, respectively. Let ${\tilde{\mv{U}}= \left[\tilde{u}_1(0),\tilde{u}_1(\ell_{e_1}),\ldots, \tilde{u}_{n_e^*}(0),\tilde{u}_{n_e^*}(\ell_{n_e^*}) \right]}$  be a vector of independent edges with covariance given by \eqref{eq:covmod} with $r$ is the stationary Mat\'ern covariance with $\alpha=1$.
By  \eqref{eq:precQexp} the density of  $\tilde{\mv{U}}$ is
\begin{align*}
f_{\tilde{\mv{U}}}\left(\tilde{\mv{u}}\right) \propto \exp \left( - \frac{1}{2} \sum_{i=1}^{n_e^*} q_{e_i,1}  \tilde{u}^2_i(0) +  q_{e_i,1} \tilde{u}^2_i(\ell_i)- 2  q_{e_i,2}\tilde{u}_i(0)\tilde{u}_i(\ell_i) \right),
\end{align*}
where $q_{e,1} = \frac{1}{2} + \frac{e^{-2\kappa \ell_e}}{1-e^{-2\kappa \ell_e}} $ and $q_{e,2} = -\frac{e^{-\kappa \ell_e}}{1-e^{-2\kappa \ell_e}}$.
Next, the density $\mv{U}_v = \mv{AU}$ from the statement can be obtained as follows.  Let $\mv{B}$ be the $2n_e^* \times n_v^*$ matrix that links each vertex in $\Gamma$, $v_i$, to its corresponding vertices on the independent edges. Then, by the continuity requirement,
\begin{align*}
	f_{\mv{U}_v}\left(\mv{u}_v\right)  = f_{\tilde{\mv{U}}}\left(\mv{B}\mv{u}_v\right)  &\propto \exp \left(  - \frac{1}{2} \sum_{e\in\mathcal{E}} \left[ q_{e,1}  u^2(\underline{e}) +  q_{e,1} u^2(\bar{e})- 2  q_{e,2}u(\underline{e})u(\bar{e})\right] \right) \\
	&= \exp \left(- \frac{1}{2} \mv{u}_v ^\trsp \mv{Q} \mv{u}_v\right).
\end{align*}
For each vertex $v^*$ created to remove a circular edge, we have the two edges $e^*_0$ and $e^*_1$ that split $e^*$ at $s^*$. 
Hence, if $v_i$ is the vertex from the circular edge $e^\ast$ (recall that $v^*$ is only connected to $v_i$), then
\begin{align*}
	f_{\mv{U}_v}\left(\mv{u}_v\right)  \propto \exp \left( - \frac{1}{2} \left( q_{e^*_0,1} + q_{e^*_1,1} \right)^2 \left( u_{v^*}^2 +u_{v_i}^2  \right)- 2    \left( q_{e^*_0,2} + q_{e^*_1,2} \right) u_{v^*} u_{v_i} , \right),
\end{align*}
and integrating out $u_{v^*}$  results in $\exp  \left( - \frac{1}{2}  \tanh\left(\kappa \frac{l_{e^*}}{2} \right) u_{v_i}^2  \right)$.
\end{proof}

The total number of non-zero elements in the matrix $\mv{Q}$ in Corollary~\ref{cor:exp_Q} is $|\mathcal{V}| + \sum_{v\in\mathcal{V}} d_v$. 
Therefore, the matrix is very sparse for a graph where the vertices have low degrees compared to the total number of vertices $|\mathcal{V}|$. 
Hence, we can easily sample $\mv{U}_v$ and compute its density via sparse Cholesky factorization \citep[see, e.g.,][]{rue2005gaussian}.  
		
For $\alpha>1$, we do not have such a simple expression for the precision matrix. 
Moreover, the expression in Corollary~\ref{cor:Uv_conditional} should not be used for large graphs because it does not take advantage of the sparsity of $\mv{Q}$. 
The problem is that, although the GMRF $\widetilde{\mv{U}}$ has a sparse precision matrix, sparsity is lost when conditioning on the Kirchhoff constraints. 
Fortunately, \cite{bolin2021efficient} recently derived a computationally efficient method for working with GMRFs under sparse linear constraints that 
directly applies to the Whittle--Mat\'ern fields. 
The idea is to change the basis of $\widetilde{\mv{U}}$ in such a way that the constraints imposed via $\mv{K}$ become noninteracting, 
because the likelihood could then be trivially evaluated. 
Thus, we transform $\widetilde{\mv{U}}$ into $\widetilde{\mv{U}}^* = \mv{T}\widetilde{\mv{U}}$ such that the $k$ constraints of $\mv{K}$ act only on $\widetilde{\mv{U}}^{*}_{\A}$, where $\A= \{1,\ldots,k\}$. 
The constraints in $\mv{K}$ are sparse and non-overlapping in the sense that each constraint acts on only one vertex. 
Therefore, the computational complexity of creating the change-of-basis matrix $\mv{T}$ through Algorithm~2 in  \cite{bolin2021efficient} is 
linear in the number of vertices. 
Further, the matrix is independent of the parameters. Hence, it must only be computed once. 
By this strategy, we obtain the following representation of $\mv{U}$, which follows from Theorem~1 and Corollary~1 in \cite{bolin2021efficient}:
	
\begin{Corollary}\label{cor:sample_Uv}
Let  $\widetilde{\mv{Q}}= diag(\{\mv{Q}_e\}_{e\in\mathcal{E}})$ denote the precision matrix of $\widetilde{\mv{U}}$ and let ${\widetilde{\mv{Q}}^* = \mv{T}\widetilde{\mv{Q}}\mv{T}^\top}$. Further, let $\mv{A}$ be as in Corollary \ref{cor:Uv_conditional}, let $\mv{T}_{\Ac}$ denote the matrix obtained by removing the first $k$ rows from $\mv{T}$, and let $\widetilde{\mv{Q}}^*_{\Ac\Ac}$ denote the matrix obtained by removing the first $k$ rows and the first $k$ columns of $\widetilde{\mv{Q}}^*$. Then
$\mv{U} \sim \pN\bigl(\mv{0}, \mv{T}_{\Ac}^\top \bigl(\widetilde{\mv{Q}}_{\Ac\Ac}^* \bigr)^{-1}\mv{T}_{\Ac}\bigr)\mathbb{I}\bigl(\mv{KU}=\mv{0}\bigr)$ and
\begin{equation}\label{eq:Uv_sparse}
	\mv{U}_v \sim \pN\left(\mv{0}, \mv{A}\mv{T}_{\Ac}^\top \left(\widetilde{\mv{Q}}_{\Ac\Ac}^* \right)^{-1}\mv{T}_{\Ac}\mv{A}^\top \right).
\end{equation}
\end{Corollary}

The matrices $\mv{A}$, $\mv{T}_{\Ac}$ and $\widetilde{\mv{Q}}_{\Ac\Ac}^*$ in \eqref{eq:Uv_sparse} are sparse. 
Thus, we can simulate $\mv{U}_v$ efficiently by simulating $\mv{v} \sim \pN(\mv{0}, (\widetilde{\mv{Q}}_{\Ac\Ac}^*)^{-1})$ through 
sparse Cholesky factorization of $\widetilde{\mv{Q}}_{\Ac\Ac}^*$ and then computing the sparse matrix vector product $\mv{U}_v = \mv{A}\mv{T}_{\Ac}^\top\mv{v}$.
	
\subsection{Log-likelihood evaluation}
To evaluate the log-likelihood, we differentiate between the case in which $\alpha=1$ and higher values of $\alpha\in\mathbb{N}$. 
The approach is much simpler for $\alpha = 1$ because we have an explicit expression for the precision matrix of the process evaluated in the vertices. 
Recall that $\mv{y}= [y_1,\ldots, y_n]$ is a vector of observations at locations $s_1,\ldots,s_n\in \Gamma$, which are 
either direct observations of $u$ or observations under Gaussian measurement noise. 
	
\subsubsection{The case of \texorpdfstring{$\alpha=1$}{alpha equal to 1}}\label{sec:alpha1likelihood}
The easiest method for computing the log-likelihood is to expand the graph $\Gamma$ to include the observation locations $s_1,\ldots,s_n$ as vertices. 
We let $\bar{\Gamma}$ denote the extended graph where $\{s_1,\ldots,s_n\}$ are added as vertices 
(duplicate nodes, which occur when $s_i\in \mathcal{V}$ for some $i$, are removed) and the edges containing those locations are subdivided. 
Let $\mv{v} = (v_1,\ldots, v_m)$ be the vector of indices of the original vertices and $\mv{s} = (s_1,\ldots,s_n)$ be the vector of indices of the locations, where $m = |\bar{\mathcal{V}}| - n$. Let, also, $\bar{\mv{U}} = (\bar{\mv{U}}_{\mv{v}}, \bar{\mv{U}}_{\mv{s}})$,
with $\bar{\mv{U}}_{\mv{v}} = (u(v_1),\ldots, u(v_{m}))$ and $\bar{\mv{U}}_{\mv{s}} = (u(s_1),\ldots, u(s_n))$. 
Then, $\bar{\mv{U}}\sim \pN\bigl(\mv{0}, \bar{\mv{Q}}^{-1}\bigr)$, where
\begin{equation}\label{eq:Qbar}
	\bar{\mv{Q}} = \begin{bmatrix}
		\mv{Q}_{\mv{vv}} & \mv{Q}_{\mv{vs}} \\
		\mv{Q}_{\mv{sv}} & \mv{Q}_{\mv{ss}}
	\end{bmatrix}
\end{equation}
denotes the corresponding complete vertex precision matrix from Corollary~\ref{cor:exp_Q}. Then, we have the following corollary.

\begin{Corollary}\label{cor:exp_density}
Suppose that $u(s)$ is the solution to \eqref{eq:Matern_spde} on a metric graph $\Gamma$ with ${\alpha=1}$, 
and that we have locations $s_1,\ldots,s_n \in \Gamma$. 
Let $\mv{U}_{\mv{s}}$ denote the joint distribution of $u(\cdot)$ evaluated at these locations. 
Then, $\mv{U}_{\mv{s}}\sim\pN(\mv{0},\mv{Q}_\mv{s}^{-1})$, where 
$\mv{Q}_{\mv{s}} = \mv{Q}_{\mv{ss}}  - \mv{Q}_{\mv{sv}} \mv{Q}_{\mv{vv}}^{-1} \mv{Q}_{\mv{vs}}$ is the precision matrix associated with the 
locations $s_1,\ldots,s_n$, with elements given in \eqref{eq:Qbar}.
\end{Corollary}	
	
Thus, for direct observations, the likelihood is directly obtained by Corollary~\ref{cor:exp_density}. 
For indirect observations, we have 
$\mv{y} | \bar{\mv{U}} \sim \pN\left(	\bar{\mv{A}}\bar{\mv{U}} , \sigma^2 \mv{I}\right)$,	
where $\bar{\mv{A}} = \left[\mv{0}_{n\times m}, \mv{I}_{n \times n} \right]$ is a matrix that maps $\bar{\mv{U}}$ to $\bar{\mv{U}}_{\mv{s}}$
and  $\mv{Q}$ is the precision matrix defined by \eqref{eq:expprec}. 
Hence, the resulting log-likelihood is
\begin{align*}
	l(\sigma,\tau,\kappa;\mv{y})=  \frac12\log\left|\mv{Q}\right| + n\log(\sigma)  - \frac12 \log\left| \mv{Q}_p\right| - \frac{ 1}{2\sigma^2}\mv{y}^\top\mv{y}  + \frac12  \mv{\mu}^{\top} \mv{Q}_p \mv{\mu},
\end{align*}
where $\mv{Q}_{p}  =  \mv{Q}+ \frac{1}{\sigma^2}  \bar{\mv{A}}^{\top}\bar{\mv{A}}$ and 
$\mv{\mu}    = \sigma^{-2}\mv{Q}_{p}^{-1} \bar{\mv{A}}^{\top}\mv{y}$. 
In either case, all matrices involved are sparse, 
so the log-likelihood is computed efficiently through sparse Cholesky factorization of $\mv{Q}$ and $\mv{Q}_p$ \citep[see, e.g.,][]{rue2005gaussian}. 
	
For indirect observations, an alternative approach is to utilize the bridge representation from Section~\ref{sec:bridge}. 
Suppose that we have $n_e$ observations on edge $e$ and let $\mv{y}_{e}$ be the vector of these observations. 
Let $\mv{t}_e$ be the vector with the positions of these locations on $e$. 
Then, the sets of observations $\{\mv{y}_{e}\}_{e\in\mathcal{E}}$ are conditionally independent given $\mv{U}$ (see \eqref{eq:dens_u}), 
with distribution $\mv{y}_{e} |\mv{U} \sim \pN\left( \mv{B}_{e}(\mv{t}_e)\mv{D}_e\mv{U}, \mv{\Sigma}_e\right)$,
where $\mv{B}_{e}(\mv{t}_e)  =   (\mv{B}_e(t_1)^\top,\ldots, \mv{B}_e(t_{n_e})^\top)^\top$ and $\mv{D}_e$ are defined in Corollary \ref{cor:bridge_representation}. In addition, $\mv{\Sigma}_e$ has elements 
$$
	(\mv{\Sigma}_e)_{ij} = \begin{cases}
		\sigma^2 + r_{B,\ell_e}(t_i, t_i) & i= j,\\
		 r_{B,\ell_e}(t_i, t_j) & i\neq j, 
	 \end{cases}
$$
where $ r_{B,\ell_e}$ is given in \eqref{eq:cov_func_whittle_matern_bridge}. 
Thus, the resulting log-likelihood is
\begin{align*}
	l(\sigma,\tau,\kappa;\mv{y})=  \frac12\log |\mv{Q}_v| - \frac12\log|\widetilde{\mv{Q}}_p| - \frac12\sum_{e\in\mathcal{E}}\log \left| \mv{\Sigma}_e\right|  + \frac12  \mv{\mu}^{\top} \mv{Q}_p \mv{\mu} - \frac{ 1}{2}\sum_{e\in\mathcal{E}}\mv{y}_e^{\top}\mv{\Sigma}_e^{-1}\mv{y}_e,
\end{align*}
where $\mv{Q}_v$ is the precision matrix of $\mv{U}_v$ (see Corollary~\ref{cor:exp_Q}), 
$$
	\widetilde{\mv{Q}}_p  =  \mv{Q}_v +  \sum_{e\in\mathcal{E}} \mv{B}_e(\mv{t}_e)^{\top} \mv{\Sigma}_e^{-1} \mv{B}_e(\mv{t}_e), \quad \text{and} \quad
	\mv{\mu}   = \mv{Q}_p^{-1} \sum_{e\in\mathcal{E}} \mv{B}_e(\mv{t}_e)^{\top}\mv{\Sigma}_e^{-1} \mv{y}_e.
$$
The advantage of this representation is that the size of $\mv{Q}_v$ is smaller than that of $\mv{Q}$ in the first approach 
(it has $|\mathcal{V}|$ rows instead of $|\mathcal{V}| + n$). 
The disadvantage is that the matrices $\mv{\Sigma}_i$ are dense; therefore, they have a cubic cost for evaluating the required log-determinant. 
However, in many applications, the number of observations for any specific edge is small, even if the total number of observations is vast. 
In such situations, this approach is more computationally efficient than the first.   
	
\subsubsection{The case of \texorpdfstring{$\alpha \in \mathbb{N} + 1$}{of integer alpha}}\label{sec:alphaNlikelihood}
For direct observations of $u$ in the case of $\alpha \in \mathbb{N}+1$, 
a straightforward approach to compute the log-likelihood is to add the observation locations as vertices, 
evaluate the density of $\mv{U}_v$ in Corollary~\ref{cor:conditional} and then integrate out the vertex locations. 
As mentioned, using Corollary~\ref{cor:conditional} is not feasible for large graphs because we cannot take advantage of sparsity. 
However, we can significantly reduce the computational costs by combining the bridge representation with the ideas developed in \citet{bolin2021efficient}. 
	
We use the same notation as in Section~\ref{sec:alpha1likelihood} and have 
$\mv{y}_{e} |\mv{U} \sim \pN\left( \mv{B}_{e}(\mv{t}_e)\mv{D}_e\mv{U}, \mv{\Sigma}_e\right)$.
Therefore, we have the joint density $\pi_{\mv{Y}|\mv{U}}(\mv{y} | \mv{u}) = \pN\left(\mv{y};\mv{B} \mv{u}, \mv{\Sigma} \right)$,
where $\mv{B}  = diag\left( \{\mv{B}_e \}_{e\in\mathcal{E}} \right)$ and 
${\mv{\Sigma} = diag\left( \{\mv{\Sigma}_e \}_{e\in\mathcal{E}} \right)}$.
As in Section~\ref{sec:Uvertex}, we perform the change of basis $\mv{U}^{*} = \mv{T}\mv{U}$ such that the $k$ constraints of $\mv{K}$ act only on $\mv{U}^{*}_{\A}$,  where $\A= \{1,\ldots,k\}$. 
The remaining unconstrained nodes are denoted by $\mv{U}^{*}_{\Ac}$. 
\cite{bolin2021efficient} derived a computationally efficient method for computing the likelihood of constrained GMRFs when $ \mv{\Sigma}_e$ are diagonal matrices. 
The following more general result can be used to compute the log-likelihood based on the change of basis in our situation with non-diagonal matrices. 
\begin{Proposition}\label{Them:piAXsoft}
Let $\mv{U} \sim \pN\left(\mv{\mu}, \mv{Q}^{-1}\right)$ be an $m$-dimensional Gaussian random variable and $\mv{Y}$ be an $n$-dimensional Gaussian random variable with
$\mv{Y} | \mv{U} \sim \pN\left( \mv{B}\mv{U} , \mv{\Sigma}\right)$ for some $n\times m$ matrix $\mv{B}$, where the matrices $\mv{Q}$  and $\mv{\Sigma}$ are strictly positive-definite.  
Further, let $\mv{K}$ be a $k \times m$ matrix of full rank and let $\mv{T}$ be the change-of-basis matrix of $\mv{K}$.
Define $\mv{Q}^*= \mv{T}\mv{Q}\mv{T}^\trsp$ and $\mv{B}^*=\mv{B}\mv{T}^\top$ and fix $\mv{b}\in \mathbb{R}^k$. 
Then, the density of $\mv{Y}|\mv{KU}=\mv{b}$ is
\begin{align*}
	\pi_{\mv{Y}|\mv{KU}=\mv{b}}(\mv{y}) = &
	\frac{|\mv{Q}^*_{\Ac\Ac}|^{\frac{1}{2}} |\mv{\Sigma}|^{-\frac12}}{\left(2\pi \right)^{\frac{n}{2}} |\widehat{\mv{Q}}^*_{\Ac\Ac} |^{\frac{1}{2}}}  
	\exp \left(-  \frac{1}{2}\left[\mv{y}^{T}\mv{\Sigma}^{-1}\mv{y}+ \widetilde{\mv{\mu}}_{\Ac} ^{*\top} \mv{Q}^*_{\Ac\Ac}\widetilde{\mv{\mu}}^*_{\Ac}- \widehat{\mv{\mu}}_{\Ac}^{*\top}\widehat{\mv{Q}} ^*_{\Ac\Ac}  \widehat{\mv{\mu}}^*_{\Ac}\right]\right),
\end{align*}
where $\widehat{\mv{Q}}^*_{\Ac\Ac} = \mv{Q}^*_{\Ac\Ac}  + \left(\mv{B}^*_{\Ac}\right)^\top \mv{\Sigma}^{-1}\ \mv{B}^*_{\Ac}$, $\left[\mv{b}^*, \mv{\mu}^*\right]=\mv{T}\left[\mv{b}\,\mv{\mu}\right]$, with 
$$
	\widehat{\mv{\mu}}^*_{\Ac} =\left(\widehat{\mv{Q}}^{*}_{\Ac\Ac} \right)^{-1}  \left(
	\mv{Q}^*_{\Ac\Ac} \widetilde{\mv{\mu}}^*_{\Ac} + \left(\mv{B}^*_{\Ac}\right)^\top \mv{\Sigma}^{-1}\ \mv{y} \right),\,\,
	\widetilde{\mv{\mu}}^* = 
	\begin{bmatrix}
		\mv{b}^* \\
		\mv{\mu}_{\Ac}^* - \left(\mv{Q}_{\Ac\Ac}^{*}\right)^{-1}  \mv{Q}^*_{\Ac\A} \left( \mv{b}^*  - \mv{\mu}^*_{\A} \right)
	\end{bmatrix}.
$$ 
\end{Proposition}
	
This result, whose proof is presented in Appendix~\ref{app:inference}, is applicable for both direct and indirect observations. 
The only difference is that one has $\mv{\Sigma} = diag\left( \{\mv{\Sigma}_e \}_{e\in\mathcal{E}}\right)$ for direct observations, 
whereas $\mv{\Sigma} = diag\left( \{\mv{\Sigma}_e \}_{e\in\mathcal{E}} \right) + \sigma^2 \mv{I}$ for indirect observations. 
The likelihood $\pi_{\mv{Y}|\mv{KU}=\mv{b}}(\mv{y})$ in Proposition \ref{Them:piAXsoft} and, therefore, 
also the log-likelihood can be evaluated solely based on operations on sparse or low-dimensional matrices, making it computationally efficient. 
	
\subsection{Simulation study}\label{sec:simulation}
To illustrate the computational benefits of using the above methods, we perform a brief simulation study in this subsection. 
We focus on the model with $\alpha=2$ for the graph of the street network in Figure~\ref{fig:data} and consider the computation time 
required to evaluate the likelihood for a set of indirect observations of the process. 
We compare three methods for evaluating the likelihood. 
The first is a covariance-based approach where the covariance matrix for the observations is computed and the standard formula for the log-likelihood of a multivariate normal distribution is used, thus ignoring sparsity. 
The second is to extend the graph with the observation locations and use Proposition~\ref{Them:piAXsoft}. 
The final is to use the bridge representation of the process, where Proposition~\ref{Them:piAXsoft} is used without extending the graph. 
We add $n$ observation locations at random on the network and consider the computation time as a function of $n$. 
The computations are repeated five times, and the average time is computed. 
The resulting timings presented in Figure~\ref{fig:timings} illustrate the benefit of taking advantage of the sparsity. 

\section{Spatial prediction for the Markov subclass}\label{sec:prediction}	
Suppose that we have observations at locations $s_1,\ldots, s_n \in\Gamma$ and that want to perform prediction at some location $p_1 \in \Gamma$. Again, we separate the cases $\alpha=1$ and $\alpha\in\mathbb{N}+1$.

We begin with the case of $\alpha=1$ and direct observations.
We let  $\widehat{\Gamma}$ be the extended graph where the observation and prediction locations are added as vertices to $\Gamma$. 
Further, we let ${\widetilde{\mv{v}} = (p_1,\mv{v}) = (p_1,v_1,\ldots,v_{|\mathcal{V}|})}$ denote the prediction location and the original vertex locations,
define $\mv{U}_{\mv{s}}$ as the Gaussian process at the observation locations, 
and let  $\mv{U}_{\widetilde{\mv{v}}}$ denote the Gaussian process evaluated at the locations in $\widetilde{\mv{v}}$. 
By Corollary \ref{cor:exp_Q}, the vector $\mv{U} = (\mv{U}_{\mv{s}}^{\top}, \mv{U}_{\widetilde{\mv{v}}}^{\top})^\top$ is a 
centered multivariate normal random variable with precision matrix $\mv{Q}$ in block form as in \eqref{eq:Qbar}. 
By standard results for conditional distributions of GMRFs \citep{rue2005gaussian}, 
$\mv{U}_{\widetilde{\mv{v}}}|\mv{U}_{\mv{s}} \sim \pN(- \mv{Q}_{\widetilde{\mv{v}}\widetilde{\mv{v}}}^{-1}\mv{Q}_{\widetilde{\mv{v}}\mv{s}}\mv{U}_{\mv{s}},\mv{Q}_{\widetilde{\mv{v}}\widetilde{\mv{v}}}^{-1}).$
We have the partitioning
$$
	\mv{Q}_{\widetilde{\mv{v}}\widetilde{\mv{v}}} = \begin{bmatrix}
		Q_{11} & \mv{q}_{\mv{v}}^{\top} \\
		\mv{q}_{\mv{v}} & \mv{Q}_{\mv{v}\mv{v}}
	\end{bmatrix}
$$
and can thus extract that $u(p_1)|\mv{U}_{\mv{s}}\sim \pN(\mu_p, \sigma_p^2)$, where $\mu_p$ is the first element in the vector $- \mv{Q}_{\widetilde{\mv{v}}\widetilde{\mv{v}}}^{-1}\mv{Q}_{\widetilde{\mv{v}}\mv{t}}\mv{U}_{\mv{s}}$ and $\sigma_p^2 = (Q_{11} - \mv{q}_{\mv{v}}^\top \mv{Q}_{\mv{v}\mv{v}}^{-1}\mv{q}_{\mv{v}})^{-1}$. 

For indirect observations and $\alpha=1$, similar standard results for GMRFs yield that 
${\mv{U}|\mv{Y} \sim \pN(\widehat{\mv{\mu}},\widehat{\mv{Q}}^{-1})}$, 
where $\widehat{\mv{Q}} = \mv{Q} + \sigma^{-2}\mv{A}^{\top}\mv{A}$, $\widehat{\mv{\mu}} = \sigma^{-2}\widehat{\mv{Q}}^{-1}\mv{A}^\top\mv{Y}$, and
$\mv{A}$ is the sparse $n\times (n + 1 + |\mathcal{V}|)$ matrix satisfying $\mv{U}_{\mv{s}} = \mv{A}\mv{U}$. 
Therefore,  $u(p_1)|\mv{Y} \sim \pN(\widehat{\mu}_p, \widehat{\sigma}_p^2)$, where $\widehat{\mu}_p$ is element $n+1$ in the vector $\widehat{\mu}$ and 
$
	\widehat{\sigma}_p^2 = (\widehat{Q}_{n+1,n+1} - \widehat{\mv{q}}^{\top}\widehat{\mv{Q}}_{-(n+1,n+1)}\widehat{\mv{q}})^{-1},
$
where $\widehat{\mv{Q}}_{-(n+1,n+1)}$ denotes the matrix $\widehat{\mv{Q}}$ with row $n+1$ and column $n+1$ removed, and $\hat{\mv{q}}$ is column $n+1$ of $\widehat{\mv{Q}}$ where element $n+1$ has been removed.

For both direct and indirect observations, all required matrix solves involve sparse matrices, 
which thus can be computed efficiently using sparse Cholesky factorization and back substitution \citep[see, e.g.,][]{rue2005gaussian}. 
An alternative approach for prediction is to apply the bridge representation, which may be more computationally efficient. 
However, for the sake of brevity, we omit the details of this approach.
	
To perform spatial prediction in the case of $\alpha \in \mathbb{N} + 1$, we use the same notation as in Section~\ref{sec:alphaNlikelihood}. 
We have the following result proved in Appendix~\ref{app:inference}:
\begin{Proposition}\label{Them:piXgby}
Let $\mv{U} \sim \pN\left(\mv{\mu}, \mv{Q}^{-1}\right)$ be an $m$-dimensional Gaussian random variable and $\mv{Y}$ be an $n$-dimensional Gaussian random variable with
$\mv{Y} | \mv{U} \sim \pN\left( \mv{B}\mv{U} , \mv{\Sigma}\right)$ for some $n\times m$ matrix $\mv{B}$, where the matrices $\mv{Q}$  and $\mv{\Sigma}$ are strictly positive-definite.  
Further, let $\mv{K}$ be a $k \times m$ matrix of full rank and let $\mv{T}$ be the change-of-basis matrix of $\mv{K}$.
Then, for $\mv{b}\in \mathbb{R}^k$, 
$\mv{U}| \left\{\mv{KU}=\mv{b},\mv{Y}=\mv{y}\right\} \sim  \pN\left(\widehat{\mv{\mu}}, \widehat{\mv{Q}}\right),$
where $\widehat{\mv{Q}} = \mv{T}^\top_{\Ac,}	\widehat{\mv{Q}}^*_{\Ac\Ac} \mv{T}_{\Ac,}$ and 
$
\widehat{\mv{\mu}} = \mv{T}^\top
		\scalebox{0.75}{$\begin{bmatrix}
				\mv{b}^* \\
				\widehat{\mv{\mu}}^*_{\Ac} 
			\end{bmatrix}$}$. Here 
$\widehat{\mv{Q}}^*_{\Ac\Ac}$ and 
$\widehat{\mv{\mu}}^*_{\Ac}$ are given in Theorem~\ref{Them:piAXsoft}.
\end{Proposition}
To use the result for spatial prediction, we can add the location we want to predict as a vertex in the graph, and obtain the prediction as the element in 
$\widehat{\mv{\mu}}$ corresponding to that element. The result can be used for both direct and indirect observations by changing $\mv{\Sigma}$, and the predictions can be computed solely based on operations on sparse matrices. 

\section{Stationary boundary conditions}\label{sec:boundary}	
As noted in \citet[][Section~8]{BSW2022}, the Whittle--Mat\'ern fields are not isotropic, and for many applications, this might be a desirable feature. 
However, the effect of the vertex conditions at vertices of degree 1 may not be desirable if these vertices are not actual 
endpoints of the graph but instead are induced from only observing a part of the graph (e.g., in Figure~\ref{fig:data}, some roads are cut by the observation window). 
To remove boundary effects, the Kirchhoff condition at vertices with degree 1, corresponding to a Neumann condition, can be replaced by boundary condition yielding a stationary solution at the boundary. 
		
In particular, for $\alpha=1$ the Robin boundary conditions $\kappa u + u' = 0$ results in a stationary model on an interval \citep{Daon2018mitigating}, 
and replacing the Kirchhoff vertex condition by a Robin boundary condition at vertices of degree 1 yields a well-defined model \citet{bolinetal_fem_graph}. 
Assuming such stationary vertex conditions means that $\frac{1}{2}\mv{r}(0,0)^{-1}$ is not removed for one of the endpoints in \eqref{eq:Qedge}. 
Thus, for $\alpha=1$, \eqref{eq:expprec} is then modified to
\begin{equation*}
	Q_{ij} =  2\kappa\tau^2\cdot \begin{cases}
		\frac{1}{2}\mathbb{I}\left(\deg(v_i) =1\right)+  \hspace{-0.15cm}\sum\limits_{e \in \mathcal{E}_{v_i}} 
		\hspace{-0.2cm}\left(	\frac{1}{2} + \frac{e^{-2\kappa \ell_e}}{1-e^{-2\kappa \ell_e}}  \right) \mathbb{I}\left(\bar{e}\neq \underline{e}\right)+	\tanh(\kappa \frac{l_e}{2} )  \mathbb{I}\left(\bar{e}= \underline{e}\right) &\hspace{-0.15cm}  i=j,\\
		\sum\limits_{e \in  \mathcal{E}_{v_i}  \cap  \mathcal{E}_{v_j} } -\frac{e^{-\kappa \ell_e}}{1-e^{-2\kappa \ell_e}} &\hspace{-0.15cm} \text{$i\neq j$.}
		\end{cases}
\end{equation*}
Similar corrections can be derived for higher $\alpha\in\mathbb{N}$. For example, 
for $\alpha = 2$, the boundary conditions  $\kappa u + u' \sim \pN(0,(2\kappa)^{-1})$ provide stationarity for the interval.
An example of the marginal variances of a Whittle--Mat\'ern field on a metric graph, with $\alpha = 2$, $\kappa = 4$, and $\tau = (4\kappa^3)^{-1}$, 
with and without the boundary correction is presented in Figure~\ref{fig:variances}. 
The effect of this boundary correction is investigated in the application in Section~\ref{sec:application}. 
	
\begin{figure}[t]
	\includegraphics[width=0.99\linewidth]{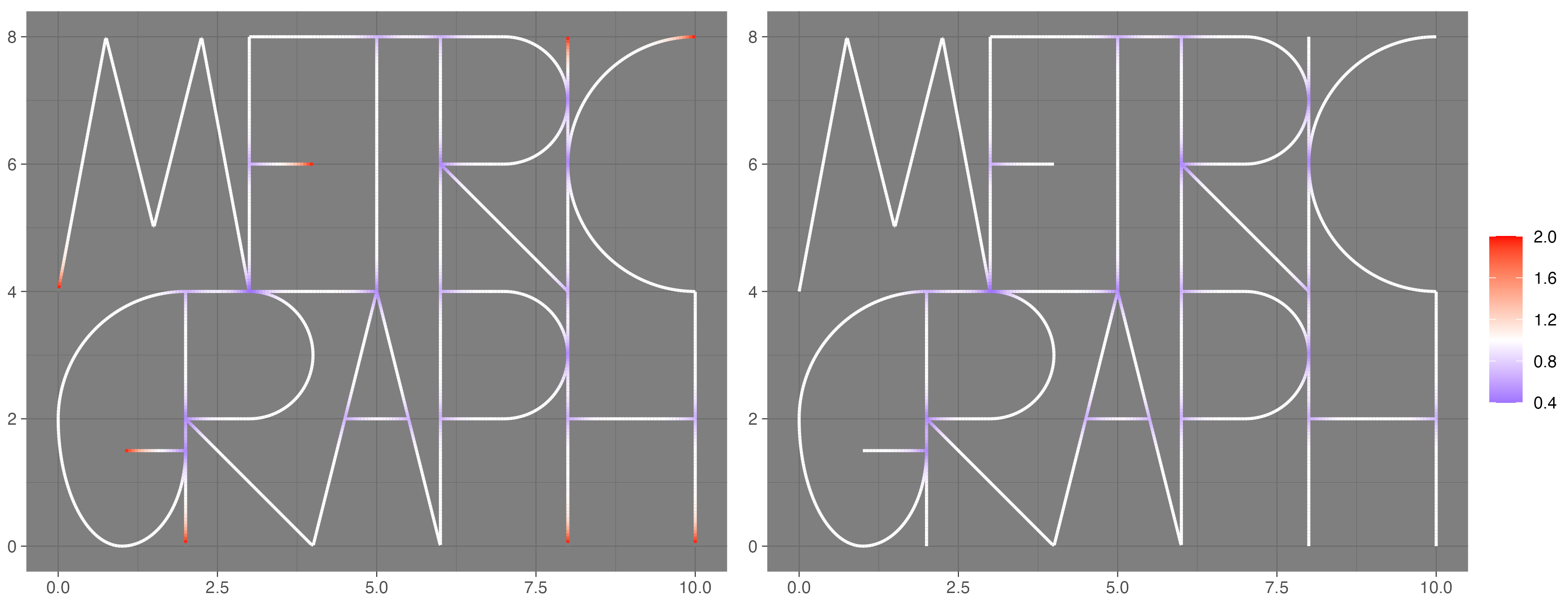}
	\caption{Marginal variances with (right) and without (left) boundary correction for the case of $\alpha=2$ on the metric graph that is used as logo for the \texttt{MetricGraph} package \cite{bsw_MetricGraph_cran}. }
	\label{fig:variances}
\end{figure}

\section{Comparison with models based on the graph Laplacian}\label{sec:graph_laplacian}
For $\alpha=1$, we have the following result regarding the limit of the precision matrix as $\kappa\rightarrow 0$. 
The result follows directly by the expression of $Q_{ij}$ in \eqref{eq:precQexp} and L'H\^ospital's rule.
			
\begin{Corollary}
Let $\alpha=1$ and $\tau^2  = 2\kappa$ and consider the matrix $\mv{Q}$ in \eqref{eq:precQexp}. Then,
$$
	\lim_{\kappa\rightarrow 0} 2\kappa Q_{ij} = 
	\begin{cases}
		\sum_{e \in \mathcal{E}_{v_i}} \frac{1}{\ell_e}  & i=j,\\
		\sum_{e\in\mathcal{E}_{v_i}\cap\mathcal{E}_{v_j}} -\frac{1}{\ell_e} &  i\neq j.
	\end{cases}
$$
\end{Corollary}
	
The matrix obtained in the corollary is a graph Laplacian matrix for a graph with edge weights $\nicefrac{1}{\ell_e}$. 
Thus, there are clear connections to the construction of Mat\'ern fields on graphs suggested by \citet{borovitskiy2021matern}. 
Their method does not define a stochastic process on the metric graph, but a multivariate distribution at the vertices. 
However, we now compare the two approaches in more detail. 
For simplicity, we assume that $\alpha=1$ so that the precision matrix at the vertices of the Whittle--Mat\'ern field is given by Corollary~\ref{cor:exp_Q}. 
	
The model proposed by \citet{borovitskiy2021matern} for a graph $\Gamma = (\mathcal{V},\mathcal{E})$ is a Gaussian process 
$\mv{u}_{\mathcal{V}}$ on $\mathcal{V}$ defined via
$(\hat{\kappa}^2\mv{I} + \mv{\Delta})^{\alpha/2}\mv{u}_V \sim \pN(\mv{0},\mv{I}),$
where $\mv{\Delta} = \mv{D} - \mv{W}$ is the (unweighted) graph Laplacian matrix, $\mv{W}$ is the adjacency matrix of the graph, 
and $\mv{D}$ is a diagonal matrix with elements $D_{ii} = d_i = \sum_{j=1}^{n_{\mathcal{V}}} W_{i,j}$. 
For $\alpha=1$, this means that $\mv{u}_{\mathcal{V}}$ is a multivariate normal distribution with precision matrix $\hat{\mv{Q}}$ with nonzero elements
\begin{equation*}
	\hat{Q}_{i,j} = \begin{cases}
	\hat{\kappa}^2 + d_i & i=j, \\
	-1								& i \sim j.
	\end{cases}
\end{equation*}
Comparing this expression to that in Corollary~\ref{cor:exp_Q} reveals that this defines a different distribution than that of the Whittle--Mat\'ern field on $\Gamma$ evaluated on ${\mathcal{V}}$. 
However, suppose that we take a graph where all edges have a length of 1 and subdivide each edge into sections of length $h$, 
obtaining a new graph $\Gamma_h$ with edges of length $h$ and where most vertices have degree $2$.
%
If we define $\hat{c} = e^{-\kappa h}/(1-e^{-2\kappa h})$ and set $\hat{\kappa}^2 = \hat{c}^{-1} + 2 e^{-\kappa h} - 2$, then
\begin{equation*}
	\hat{c}\hat{Q}_{i,j} = \begin{cases}
	Q_{i,i} &  \text{$i=j$ and $d_i = 2$ or $i\neq j$,} \\
	Q_{i,i} + 1 - \frac{d}{2} + \hat{c}(d-2) (e^{-\kappa h} -1) &  \text{$i=j$ and $d_i \neq 2$}.
	\end{cases}
\end{equation*}
Thus, the precision matrix $\hat{c}\hat{\mv{Q}}$ agrees with that of a Whittle--Mat\'ern field for all vertices except those with a degree different from $2$. Hence, the covariance matrices $\mv{\Sigma} = \mv{Q}^{-1}$ and $\hat{\mv{\Sigma}} = \hat{c}^{-1}\hat{\mv{Q}}^{-1}$ are similar 
if $h$ is small and most vertices have a degree of $2$.
	
For example, suppose that we only have one vertex with degree $3$ in the graph, whereas the other vertices have a degree of 2. 
Then, $\hat{c}\hat{\mv{Q}} = \mv{Q} + \mv{v}\mv{v}^{\top}$, where $\mv{v}$ is a vector with all zeros except for one element 
$v_i = (1 - \frac{d}{2} + \hat{c}(d-2) (e^{-\kappa h} -1))^{\nicefrac12}.$
Via the Sherman--Morrison formula, we obtain
$\mv{\Sigma} - \hat{\mv{\Sigma}} = v_i(1+v_i^2 \Sigma_{ii})^{-1} \mv{\Sigma}_i\mv{\Sigma}_i^\top,$
where $\mv{\Sigma}_i$ is the $i$th column of $\mv{\Sigma}$. 
Because $v_i\rightarrow 0$ as ${h\rightarrow 0}$, 
we observe that each element in the difference $\mv{\Sigma} - \hat{\mv{\Sigma}}$ converges to $0$ as ${h\rightarrow 0}$. 
Thus, the construction based on the graph Laplacian can be viewed as a finite-difference approximation of the Whittle--Mat\'ern field. 
However, for a nonsubdivided graph, the difference between the construction based on the graph Laplacian and the exact Whittle--Mat\'ern field can be considerable. 
In addition, subdividing the graph to get a good approximation induces a high computational cost. 
Finally, another advantage of the Whittle--Mat\'ern fields compared to the approach based on the graph Laplacian is that we can add vertices on existing edges, 
and remove vertices with a degree of 2 in the graph without changing the Whittle--Mat\'ern fields. 
This property is useful for applications because we do not change the model when adding observation or prediction locations to the graph. 
	
\section{Application}\label{sec:application}
We consider the traffic speed data in Figure \ref{fig:data}. 
The network was obtained from OpenStreetMap \citep{OpenStreetMap}, 
and the traffic speeds were obtained from the California performance measurement system database \citep{chen2001freeway}. 
We have $n= 325$ observations, which we assume follow the model \eqref{eq:likelihood}. 
That is, the observations are assumed to be noisy observations of a centered Gaussian process $u$ on $\Gamma$. 
We compare seven different models for $u$: 
Whittle--Mat\'ern fields with $\alpha=1$ and $\alpha=2$, with and without the stationary boundary correction of Section~\ref{sec:boundary}, 
a Gaussian process with an isotropic exponential covariance function using the resistance metric from \cite{anderes2020isotropic}, 
and the models from \cite{borovitskiy2021matern} based on the graph Laplacian for $\alpha=1$ and $\alpha=2$. 
Note that the latter two models are not defined on $\Gamma$, but only on the vertices, $\mathcal{V}$. 
Therefore, we add all observation locations as vertices in the graph for those models.  
In addition, the graph does not have Euclidean edges, so the isotropic model is not guaranteed to be valid.
		
All models are fitted to the data using numerical optimization of the respective log-likelihood function, 
in \texttt{R} using the \texttt{MetricGraph} \citep{bsw_MetricGraph_cran} package. 
The negative log-likelihood value for each model is presented in Table~\ref{tab:res}. 
All seven models have the same number of parameters, so according to measures such as Akaike or Bayesian information criteria (AIC or BIC), 
the Whittle--Mat\'ern model with $\alpha=2$ is the most appropriate.
	
\begin{figure}
	\centering
	\includegraphics[width=\linewidth]{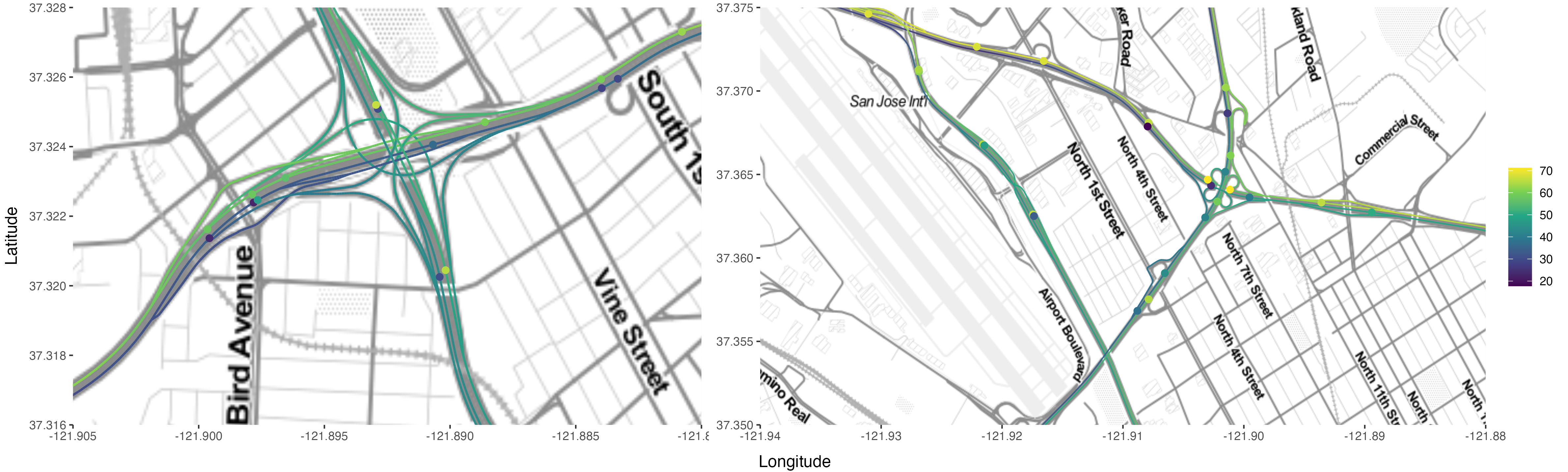}
	\caption{Kriging estimate on the two subareas of $\Gamma$ shown in Figure~\ref{fig:data} for the Whittle--Mat\'ern field with $\alpha=2$. 
		The estimate is not shown on the full graph because one cannot see the individual roads without zooming in on specific areas.}
	\label{fig:interpolation}
\end{figure}

Given the estimated parameters, we can perform kriging prediction to estimate the process at any location 
in the street network for all models but those based on the graph Laplacian. 
Figure~\ref{fig:interpolation} depicts an example of this, 
where the posterior mean of $u(s)$, given the observations, is shown for the Whittle--Mat\'ern model with $\alpha=2$. 
We can only perform prediction at the graph vertices for the models based on the graph Laplacian.
Thus, to perform kriging prediction at unobserved locations, we need to add those in the graph before fitting the model, 
which can be computationally demanding and changes the model. 

We perform a five-fold pseudo cross validation to compare the predictive quality of the seven models. 
The data are split into five groups at random, and  the data from the first group are predicted based on the fitted model and the data from the remaining four groups. 
This process is repeated  five times, each with a different group that is predicted. 
We evaluate the accuracy of the predictions using averages of five different scoring rules: the root mean square error (RMSE), 
mean absolute error (MAE), log-score (LS), continuous ranked probability score (CRPS), and the scaled CRPS (SCRPS; see \cite{bolin2022local} and \cite{gneiting2007strictly} for details about these scoring rules). 
The average scores are listed in Table~\ref{tab:res}.  
The Whittle--Mat\'ern model with $\alpha=2$ without boundary corrections fits best, 
but model with $\alpha=2$ and boundary corrections, and the model based on the graph Laplacian with $\alpha=2$ have similar results. 
This finding confirms the theoretical results in Proposition~\ref{prop:A-kriging}, 
that the smoothness of the field is crucial for the predictive accuracy. 
It also reveals that the higher degree of isotropy achieved by using the boundary corrections might not have a significant effect and
actually reduces the predictive performance in this case.

\begin{table}[t]\centering
   	\caption{Various scoring rules that measure the predictive performance of the models in the application.  
 	Models with boundary corrections are marked as ``adjusted.'' A lower value is better for all scores; the best values are indicated in bold. 
 	The first five scores are computed through five-fold cross-validation, whereas the negative log-likelihood is computed on the whole dataset.}
	\begin{tabular}{rrrrrrr}  \hline & RMSE & MAE & LS & CRPS & SCRPS &  neg. log-like \\   \hline 
	Isotropic exponential & 8.9327 & 6.5496 & 3.6378 & 4.9410 & 2.1527 & 1223.9716 \\ 
	Whittle--Mat\'ern $\alpha=1$ adjusted  & 8.9366 & 6.5516 & 3.6521 & 4.9996 & 2.1618 & 1221.3791 \\  
	Whittle--Mat\'ern $\alpha=1$ & 8.9417 & 6.5607 & 3.6522 & 5.0017 & 2.1619 & 1221.2284 \\   
	Graph Laplacian $\alpha=1$& 8.9370 & 6.5583 & 3.6525 & 5.0019 & 2.1620 & 1221.3832 \\ 
	Whittle--Mat\'ern $\alpha=2$  adjusted & 8.5448 & 6.1499 & 3.5850 & 4.7067 & \bf  2.1264 & 1208.1840 \\   
	Whittle--Mat\'ern $\alpha=2$  & \bf 8.5438 & \bf  6.1492 & \bf  3.5849 & \bf 4.7063 & \bf  2.1264 & \bf  1207.8678 \\ 
	Graph Laplacian $\alpha=2$  & 8.5897 & 6.1536 & 3.5980 & 4.7742 & 2.1348 & 1208.7028 \\   \hline
	\end{tabular}
	\label{tab:res}
\end{table}
 
\section{Discussion}\label{sec:discussion}
We comprehensively characterized the statistical properties of the Gaussian Whittle--Mat\'ern fields and demonstrated that their finite-dimensional 
distributions can be derived exactly for the Markov cases $\alpha\in\mathbb{N}$. 
In these cases, we obtained sparse precision matrices that facilitate their use in real applications to extensive datasets 
via computationally efficient implementations based on sparse matrices.
	
We argue that this class of models is a natural choice for applications where Gaussian random fields are needed to model data on metric graphs, 
and that the anisotropy of the models may be a desirable property for many applications. 
Also, as far as we know, no other constructions can currently provide differentiable Gaussian processes on general metric graphs. 
Having differentiable processes may be useful in many applications, 
particularly applications involving log-Gaussian Cox processes \citep{moller2022lgcp}, 
where it may be desirable to have a smoothly varying intensity function. 
Considering such applications is planned for future work. 
Another interesting aspect is to consider generalized Whittle--Mat\'ern fields \citep{bolinetal_fem_graph}, 
where we can allow the parameter $\kappa$ to be non-stationary over the graph. 
In the Markov case, a natural extension is to expand Theorem \ref{thm:representation} to allow for a non-stationary covariance. 
In the non-Markovian case, a natural extension is to use a finite element method combined with an 
approximation of the fractional operator, similarly to the methods by \citet{BKK2020}, \citet{BK2020rational} and \citet{xiong2022}, 
to approximate the random fields. Some initial work in this direction is presented in \citet{bolinetal_fem_graph}.
	 
We presented adjusted vertex conditions at vertices of degree 1 to remove boundary effects, 
and it would be interesting to investigate alternative vertex conditions for vertices with higher degrees, 
which could be used to make the marginal variances more stationary across the graph. 
Finally, considering non-Gaussian extensions of the Whittle--Mat\'ern fields similarly to the Type-G random fields on Euclidean domains by \citet{bw20} is also an interesting topic for future work.

\bibliographystyle{chicago}
\bibliography{../../Bib/unified_graph_bib}

\begin{thebibliography}{}

\bibitem[\protect\citeauthoryear{Anderes, M{\o}ller, Rasmussen, et~al.}{Anderes
  et~al.}{2020}]{anderes2020isotropic}
Anderes, E., J.~M{\o}ller, J.~G. Rasmussen, et~al. (2020).
\newblock Isotropic covariance functions on graphs and their edges.
\newblock {\em Ann. Statist.\/}~{\em 48\/}(4), 2478--2503.

\bibitem[\protect\citeauthoryear{Baddeley, Nair, Rakshit, and
  McSwiggan}{Baddeley et~al.}{2017}]{baddeley2017stationary}
Baddeley, A., G.~Nair, S.~Rakshit, and G.~McSwiggan (2017).
\newblock Stationary point processes are uncommon on linear networks.
\newblock {\em Stat\/}~{\em 6\/}(1), 68--78.

\bibitem[\protect\citeauthoryear{Berkolaiko and Kuchment}{Berkolaiko and
  Kuchment}{2013}]{Berkolaiko2013}
Berkolaiko, G. and P.~Kuchment (2013).
\newblock {\em Introduction to quantum graphs}, Volume 186 of {\em Mathematical
  Surveys and Monographs}.
\newblock American Mathematical Society, Providence, RI.

\bibitem[\protect\citeauthoryear{Bogachev}{Bogachev}{1998}]{Bogachev1998}
Bogachev, V.~I. (1998).
\newblock {\em Gaussian {M}easures}, Volume~62 of {\em Mathematical Surveys and
  Monographs}.
\newblock American Mathematical Society, Providence, RI.

\bibitem[\protect\citeauthoryear{Bolin and Kirchner}{Bolin and
  Kirchner}{2020}]{BK2020rational}
Bolin, D. and K.~Kirchner (2020).
\newblock The rational {SPDE} approach for {G}aussian random fields with
  general smoothness.
\newblock {\em J. Comput. Graph. Statist.\/}~{\em 29\/}(2), 274--285.

\bibitem[\protect\citeauthoryear{Bolin and Kirchner}{Bolin and
  Kirchner}{2023}]{bk-measure}
Bolin, D. and K.~Kirchner (2023).
\newblock Equivalence of measures and asymptotically optimal linear prediction
  for {Gaussian} random fields with fractional-order covariance operators.
\newblock {\em Bernoulli\/}~{\em 29}, 1476--1504.

\bibitem[\protect\citeauthoryear{Bolin, Kirchner, and Kov\'{a}cs}{Bolin
  et~al.}{2020}]{BKK2020}
Bolin, D., K.~Kirchner, and M.~Kov\'{a}cs (2020).
\newblock Numerical solution of fractional elliptic stochastic {PDE}s with
  spatial white noise.
\newblock {\em IMA J. Numer. Anal.\/}~{\em 40\/}(2), 1051--1073.

\bibitem[\protect\citeauthoryear{Bolin, Kov\'{a}cs, Kumar, and Simas}{Bolin
  et~al.}{2023}]{bolinetal_fem_graph}
Bolin, D., M.~Kov\'{a}cs, V.~Kumar, and A.~B. Simas (2023).
\newblock Regularity and numerical approximation of fractional elliptic
  differential equations on compact metric graphs.
\newblock Preprint, arXiv:2302.03995.

\bibitem[\protect\citeauthoryear{Bolin, Simas, and Wallin}{Bolin
  et~al.}{2023a}]{BSW2022}
Bolin, D., A.~B. Simas, and J.~Wallin (2023a).
\newblock Gaussian {W}hittle-{M}at\'ern fields on metric graphs.
\newblock {\em Bernoulli\/}.
\newblock In press.

\bibitem[\protect\citeauthoryear{Bolin, Simas, and Wallin}{Bolin
  et~al.}{2023b}]{BSW_Markov}
Bolin, D., A.~B. Simas, and J.~Wallin (2023b).
\newblock Markov properties of {G}aussian random fields on compact metric
  graphs.
\newblock Preprint, arxiv:2304.03190.

\bibitem[\protect\citeauthoryear{Bolin, Simas, and Wallin}{Bolin
  et~al.}{2023c}]{bsw_MetricGraph_cran}
Bolin, D., A.~B. Simas, and J.~Wallin (2023c).
\newblock {\em MetricGraph: Random fields on metric graphs}.
\newblock R package version 1.1.2.

\bibitem[\protect\citeauthoryear{Bolin, Simas, and Xiong}{Bolin
  et~al.}{2023}]{xiong2022}
Bolin, D., A.~B. Simas, and Z.~Xiong (2023).
\newblock Covariance-based rational approximations of fractional {SPDEs }for
  computationally efficient {B}ayesian inference.
\newblock {\em J. Comput. Graph. Statist.\/}.
\newblock In press.

\bibitem[\protect\citeauthoryear{Bolin and Wallin}{Bolin and
  Wallin}{2020}]{bw20}
Bolin, D. and J.~Wallin (2020).
\newblock Multivariate type {G} {M}at\'{e}rn stochastic partial differential
  equation random fields.
\newblock {\em J. R. Stat. Soc. Ser. B. Stat. Methodol.\/}~{\em 82\/}(1),
  215--239.

\bibitem[\protect\citeauthoryear{Bolin and Wallin}{Bolin and
  Wallin}{2021}]{bolin2021efficient}
Bolin, D. and J.~Wallin (2021).
\newblock Efficient methods for {Gaussian} {Markov} random fields under sparse
  linear constraints.
\newblock {\em Advances in Neural Information Processing Systems\/}~{\em 34}.

\bibitem[\protect\citeauthoryear{Bolin and Wallin}{Bolin and
  Wallin}{2022}]{bolin2022local}
Bolin, D. and J.~Wallin (2022).
\newblock Local scale invariance and robustness of proper scoring rules.
\newblock {\em Statist. Sci.\/}.

\bibitem[\protect\citeauthoryear{Borovitskiy, Azangulov, Terenin, Mostowsky,
  Deisenroth, and Durrande}{Borovitskiy et~al.}{2021}]{borovitskiy2021matern}
Borovitskiy, V., I.~Azangulov, A.~Terenin, P.~Mostowsky, M.~Deisenroth, and
  N.~Durrande (2021).
\newblock Mat{\'e}rn {G}aussian processes on graphs.
\newblock In {\em International Conference on Artificial Intelligence and
  Statistics}, pp.\  2593--2601. PMLR.

\bibitem[\protect\citeauthoryear{Chen, Petty, Skabardonis, Varaiya, and
  Jia}{Chen et~al.}{2001}]{chen2001freeway}
Chen, C., K.~Petty, A.~Skabardonis, P.~Varaiya, and Z.~Jia (2001).
\newblock Freeway performance measurement system: mining loop detector data.
\newblock {\em Transportation Research Record\/}~{\em 1748\/}(1), 96--102.

\bibitem[\protect\citeauthoryear{Cronie, Moradi, and Mateu}{Cronie
  et~al.}{2020}]{cronie2020}
Cronie, O., M.~Moradi, and J.~Mateu (2020).
\newblock Inhomogeneous higher-order summary statistics for point processes on
  linear networks.
\newblock {\em Stat. Comput.\/}~{\em 30\/}(5), 1221--1239.

\bibitem[\protect\citeauthoryear{Da~Prato and Zabczyk}{Da~Prato and
  Zabczyk}{2014}]{daPrato2014}
Da~Prato, G. and J.~Zabczyk (2014).
\newblock {\em Stochastic {E}quations in {I}nfinite {D}imensions\/} (Second
  ed.), Volume 152 of {\em Encyclopedia of Mathematics and its Applications}.
\newblock Cambridge: Cambridge University Press.

\bibitem[\protect\citeauthoryear{Daon and Stadler}{Daon and
  Stadler}{2018}]{Daon2018mitigating}
Daon, Y. and G.~Stadler (2018).
\newblock Mitigating the influence of the boundary on {PDE}-based covariance
  operators.
\newblock {\em Inverse Probl. Imaging\/}~{\em 12\/}(5), 1083--1102.

\bibitem[\protect\citeauthoryear{Dudley}{Dudley}{2018}]{dudley2018real}
Dudley, R.~M. (2018).
\newblock {\em Real analysis and probability}.
\newblock CRC Press.

\bibitem[\protect\citeauthoryear{Dunson, Wu, and Wu}{Dunson
  et~al.}{2022}]{dunson2020graph}
Dunson, D.~B., H.-T. Wu, and N.~Wu (2022).
\newblock Graph based {G}aussian processes on restricted domains.
\newblock {\em J. R. Stat. Soc. Ser. B. Stat. Methodol.\/}~{\em 84\/}(2),
  414--439.

\bibitem[\protect\citeauthoryear{Evans}{Evans}{2010}]{evans2010partial}
Evans, L.~C. (2010).
\newblock {\em Partial differential equations}, Volume~19.
\newblock American Mathematical Soc.

\bibitem[\protect\citeauthoryear{Gikhman and Skorokhod}{Gikhman and
  Skorokhod}{1974}]{gikhmanskorohod}
Gikhman, I. and A.~Skorokhod (1974).
\newblock {\em The theory of stochastic processes I}.
\newblock Springer.

\bibitem[\protect\citeauthoryear{Gneiting and Raftery}{Gneiting and
  Raftery}{2007}]{gneiting2007strictly}
Gneiting, T. and A.~E. Raftery (2007).
\newblock Strictly proper scoring rules, prediction, and estimation.
\newblock {\em J. Amer. Statist. Assoc.\/}~{\em 102\/}(477), 359--378.

\bibitem[\protect\citeauthoryear{Ibragimov and Rozanov}{Ibragimov and
  Rozanov}{2012}]{ibragimov_rozanov}
Ibragimov, I.~A. and Y.~A. Rozanov (2012).
\newblock {\em Gaussian random processes}, Volume~9.
\newblock Springer Science \& Business Media.

\bibitem[\protect\citeauthoryear{Janson}{Janson}{1997}]{janson_gaussian}
Janson, S. (1997).
\newblock {\em Gaussian {H}ilbert spaces}, Volume 129 of {\em Cambridge Tracts
  in Mathematics}.
\newblock Cambridge: Cambridge University Press.

\bibitem[\protect\citeauthoryear{Kaufman and Shaby}{Kaufman and
  Shaby}{2013}]{Kaufman2013}
Kaufman, C.~G. and B.~A. Shaby (2013).
\newblock The role of the range parameter for estimation and prediction in
  geostatistics.
\newblock {\em Biometrika\/}~{\em 100\/}(2), 473--484.

\bibitem[\protect\citeauthoryear{Kirchner and Bolin}{Kirchner and
  Bolin}{2022}]{kb-kriging}
Kirchner, K. and D.~Bolin (2022).
\newblock Necessary and sufficient conditions for asymptotically optimal linear
  prediction of random fields on compact metric spaces.
\newblock {\em Ann.\ Statist.\/}~{\em 50\/}(2), 1038--1065.

\bibitem[\protect\citeauthoryear{Lauritzen}{Lauritzen}{1996}]{lauritzen1996graphical}
Lauritzen, S.~L. (1996).
\newblock {\em Graphical models}, Volume~17.
\newblock Clarendon Press.

\bibitem[\protect\citeauthoryear{Lehmann}{Lehmann}{1999}]{lehmann1999elements}
Lehmann, E.~L. (1999).
\newblock {\em Elements of large-sample theory}.
\newblock Springer.

\bibitem[\protect\citeauthoryear{Lindgren, Bolin, and Rue}{Lindgren
  et~al.}{2022}]{lindgren2022spde}
Lindgren, F., D.~Bolin, and H.~Rue (2022).
\newblock The {SPDE} approach for {G}aussian and non-{G}aussian fields: 10
  years and still running.
\newblock {\em Spat. Stat.\/}~{\em 50}, Paper No. 100599.

\bibitem[\protect\citeauthoryear{Lindgren, Rue, and Lindstr\"{o}m}{Lindgren
  et~al.}{2011}]{lindgren11}
Lindgren, F., H.~Rue, and J.~Lindstr\"{o}m (2011).
\newblock An explicit link between {G}aussian fields and {G}aussian {M}arkov
  random fields: the stochastic partial differential equation approach.
\newblock {\em J.\ R.\ Stat.\ Soc.\ Ser.\ B Stat.\ Methodol.\/}~{\em 73\/}(4),
  423--498.
\newblock With discussion and a reply by the authors.

\bibitem[\protect\citeauthoryear{McLean}{McLean}{2000}]{mclean}
McLean, W. (2000).
\newblock {\em Strongly elliptic systems and boundary integral equations}.
\newblock Cambridge: Cambridge University Press.

\bibitem[\protect\citeauthoryear{M{\o}ller and Rasmussen}{M{\o}ller and
  Rasmussen}{2022}]{moller2022lgcp}
M{\o}ller, J. and J.~G. Rasmussen (2022).
\newblock Cox processes driven by transformed {G}aussian processes on linear
  networks.
\newblock Preprint, arXiv:2212.08402.

\bibitem[\protect\citeauthoryear{Od\v{z}ak and \v{S}\'{c}eta}{Od\v{z}ak and
  \v{S}\'{c}eta}{2019}]{Odzak2019Weyl}
Od\v{z}ak, A. and L.~\v{S}\'{c}eta (2019).
\newblock On the {W}eyl law for quantum graphs.
\newblock {\em Bull. Malays. Math. Sci. Soc.\/}~{\em 42\/}(1), 119--131.

\bibitem[\protect\citeauthoryear{Okabe and Sugihara}{Okabe and
  Sugihara}{2012}]{okabe2012spatial}
Okabe, A. and K.~Sugihara (2012).
\newblock {\em Spatial analysis along networks: statistical and computational
  methods}.
\newblock John Wiley \& Sons.

\bibitem[\protect\citeauthoryear{{OpenStreetMap contributors}}{{OpenStreetMap
  contributors}}{2017}]{OpenStreetMap}
{OpenStreetMap contributors} (2017).
\newblock {Planet dump retrieved from https://planet.osm.org}.
\newblock \url{ https://www.openstreetmap.org }.

\bibitem[\protect\citeauthoryear{Pitt}{Pitt}{1971}]{Pitt1971}
Pitt, L.~D. (1971).
\newblock A {M}arkov property for {G}aussian processes with a multidimensional
  parameter.
\newblock {\em Arch. Ration. Mech. Anal.\/}~{\em 43}, 367--391.

\bibitem[\protect\citeauthoryear{Porcu, White, and Genton}{Porcu
  et~al.}{2022}]{porcu2022}
Porcu, E., P.~A. White, and M.~G. Genton (2022).
\newblock Nonseparable space-time stationary covariance functions on networks
  cross time.
\newblock Preprint, arXiv:2208.03359.

\bibitem[\protect\citeauthoryear{Rue and Held}{Rue and
  Held}{2005}]{rue2005gaussian}
Rue, H. and L.~Held (2005).
\newblock {\em Gaussian {M}arkov random fields}, Volume 104 of {\em Monographs
  on Statistics and Applied Probability}.
\newblock Boca Raton, FL: Chapman \& Hall/CRC.
\newblock Theory and applications.

\bibitem[\protect\citeauthoryear{Sanz-Alonso and Yang}{Sanz-Alonso and
  Yang}{2022}]{Alonso2021}
Sanz-Alonso, D. and R.~Yang (2022).
\newblock The {SPDE} approach to {M}at\'ern fields: Graph representations.
\newblock {\em Statist. Sci.\/}~{\em 37}, 519--540.

\bibitem[\protect\citeauthoryear{Sriperumbudur, Fukumizu, and
  Lanckriet}{Sriperumbudur et~al.}{2011}]{sriperumbudur2011universality}
Sriperumbudur, B.~K., K.~Fukumizu, and G.~R. Lanckriet (2011).
\newblock Universality, characteristic kernels and rkhs embedding of measures.
\newblock {\em Journal of Machine Learning Research\/}~{\em 12\/}(7).

\bibitem[\protect\citeauthoryear{Stein}{Stein}{1999}]{stein99}
Stein, M.~L. (1999).
\newblock {\em Interpolation of {S}patial {D}ata: {S}ome {T}heory for
  {K}riging}.
\newblock Springer Series in Statistics. Springer-Verlag, New York.

\bibitem[\protect\citeauthoryear{Steinwart and Scovel}{Steinwart and
  Scovel}{2012}]{Steinwart2012}
Steinwart, I. and C.~Scovel (2012).
\newblock Mercer's theorem on general domains: on the interaction between
  measures, kernels, and {RKHS}s.
\newblock {\em Constr. Approx.\/}~{\em 35\/}(3), 363--417.

\bibitem[\protect\citeauthoryear{Thom\'{e}e}{Thom\'{e}e}{2006}]{Thomee}
Thom\'{e}e, V. (2006).
\newblock {\em Galerkin finite element methods for parabolic problems\/}
  (Second ed.), Volume~25 of {\em Springer Series in Computational
  Mathematics}.
\newblock Berlin: Springer-Verlag.

\bibitem[\protect\citeauthoryear{Wendland}{Wendland}{2004}]{wendland}
Wendland, H. (2004).
\newblock {\em Scattered data approximation}, Volume~17.
\newblock Cambridge university press.

\bibitem[\protect\citeauthoryear{Whittle}{Whittle}{1963}]{whittle63}
Whittle, P. (1963).
\newblock Stochastic processes in several dimensions.
\newblock {\em Bull.\ Internat.\ Statist.\ Inst.\/}~{\em 40}, 974--994.

\bibitem[\protect\citeauthoryear{Zhang}{Zhang}{2004}]{Zhang2004}
Zhang, H. (2004).
\newblock Inconsistent estimation and asymptotically equal interpolations in
  model-based geostatistics.
\newblock {\em J.\ Amer.\ Statist.\ Assoc.\/}~{\em 99\/}(465), 250--261.

\end{thebibliography}
\newpage

\begin{appendix}
\section{Strict positive-definiteness of $\varrho(\cdot,\cdot)$}\label{app:strposdef}
We begin this section by recalling some relevant auxiliary definitions. We refer the reader to \cite{sriperumbudur2011universality} for further details. 
Let $M_b(\Gamma)$ be the space of all finite signed Radon measures on $\Gamma$. Also, recall that $\Gamma$ is a compact metric space, so it is, in particular, locally compact and Hausdorff. We say that a measurable, symmetric and bounded kernel $k:\Gamma\times\Gamma\to\mathbb{R}$ is integrally strictly positive-definite if
$$\forall \mu \in M_b(\Gamma)\setminus \{0\},\quad \int_\Gamma \int_\Gamma k(x,y) d\mu(x) d\mu(y) > 0.$$
\begin{Remark}\label{rem:ispdimpliesspd}
Note that by taking $\mu = \sum_{j=1}^N a_j \delta_{s_j}$, where $N\in\mathbb{N}$, $a_1,\ldots,a_N \in \mathbb{R}\setminus\{0\}$ and $s_1,\ldots,s_N\in \Gamma$, and $\delta_{s}(\cdot)$ is the Dirac measure concentrated at $s\in\Gamma$, we obtain that a integrally strictly positive-definite kernel $k(\cdot,\cdot)$ is also strictly positive-definite.
\end{Remark}

A  positive semi-definite kernel $k:\Gamma\times\Gamma\to\mathbb{R}$ is called a $c$-kernel if it is bounded and for every $s\in\Gamma$, $k(\cdot,s)\in C(\Gamma)$. Furthermore, $k$ is $c$-universal if the reproducing kernel Hilbert space induced by $k$ is dense in $C(\Gamma)$ with respect to the $\|\cdot\|_{C(\Gamma)}$-norm. Observe that since $\Gamma$ is compact, the notion of $c$-universality agrees with the notion of $c_0$-universality in \cite{sriperumbudur2011universality}.

We have the following result (adapted to our context of compact metric graphs) whose proof can be found in \citet[Proposition 4]{sriperumbudur2011universality}:

\begin{Proposition}\label{prp:c0universal_sipd}
Let $k$ be a $c$-kernel on a compact metric graph $\Gamma$. Then, $k$ is $c$-universal if, and only if, $k$ is integrally strictly positive-definite.
\end{Proposition}

Let $\varrho(\cdot,\cdot)$ be the covariance function of $u$, where $u$ is given by the solution to \eqref{eq:Matern_spde}. Thus, by Theorem \ref{thm:regularity}, $\varrho \in C(\Gamma\times\Gamma)$. In particular, since $\Gamma$ is compact, $\varrho$ is a $c$-kernel. Further, by, e.g., \citet[Corollary 8.16]{janson_gaussian}, the reproducing kernel Hilbert space induced by $\varrho$ is the Cameron--Martin space associated to $u$. By \citet[Proposition 8]{BSW2022}, the Cameron--Martin space associated to $u$ is $\dot{H}^\alpha$.
Moreover, in \citet[Propositions 3 and 4]{BSW2022}, we have a characterization for $\dot{H}^k$ for any $k\in\mathbb{N}$. Thus, our goal in showing that $\varrho$ is strictly positive-definite is to use Proposition \ref{prp:c0universal_sipd} and Remark \ref{rem:ispdimpliesspd} to reduce the problem to showing that $\varrho$ is $c$-universal. Finally, to this end, we will find a set $\mathcal{A}(\Gamma)\subset \bigcap_{k\in\mathbb{N}} \dot{H}^k$ such that $\mathcal{A}(\Gamma)$ is dense in $C(\Gamma)$ with respect to the $\|\cdot\|_{C(\Gamma)}$ norm.

In view of the Stone-Weierstrass theorem \cite[Theorem 2.4.11]{dudley2018real}, all we need to do is find a set $\mathcal{A}(\Gamma)$ such that it is a sub-algebra of $C(\Gamma)$, it contains constant functions and separates points. We will now obtain some auxiliary function spaces in order to obtain such a sub-algebra.

Let $\<1\> = \textrm{span}\{1\}$ be the space of constant functions on $\Gamma$ and $D(\Gamma) = \oplus_{e\in\mathcal{E}} C_c^\infty(e)$ be the space of functions with the support in the union of the interiors of the edges whose restrictions to the edges are infinitely differentiable. From  \citet[Propositions 3 and 4]{BSW2022}, $\<1\>\subset \bigcap_{k\in\mathbb{N}} \dot{H}^k$ and by \citet[Proposition 10]{BSW_Markov}, $D(\Gamma) \subset \bigcap_{k\in\mathbb{N}} \dot{H}^k$. Therefore, we obtain that ${\<1\> + D(\Gamma) = \{c + f: c\in\mathbb{R}, f\in D(\Gamma)\} \subset \bigcap_{k\in\mathbb{N}} \dot{H}^k}$. It is easy to see that $\<1\> + D(\Gamma)$ is an algebra that contains the constant functions and separates the points in the interiors of the edges from all the other points of $\Gamma$. We, thus, need to obtain an additional space that separates the vertices of the metric graph from the remaining vertices.  

To define such a set, we introduce some additional definitions and assumptions. Given a compact metric graph $\Gamma$, let $\widetilde{\Gamma}$ be the compact metric graph obtained from $\Gamma$ by adding one vertex (of degree 2) at the interior of each edge that is a loop. Therefore, $\widetilde{\Gamma}$ has no loops and, by \citet[Proposition 2]{BSW2022}, the solution to $(\kappa^2-\Delta_\Gamma)^{\alpha/2}(\tau u) = \mathcal{W}$ on $\widetilde{\Gamma}$ has the same covariance function as the solution to \eqref{eq:Matern_spde}. Therefore, we can assume, without loss of generality, that the compact metric graph $\Gamma$ does not contain loops, although it may contain cycles. For each $v\in\mathcal{V}$, let
$$S(v,\Gamma) := \{s\in\Gamma: s = (t,e), e\in \mathcal{E}_v, e = [0,\ell_e], t\in (0,\ell_e)\} \cup\{v\}$$
be the star graph induced by $v$, with the outer vertices removed. That is, $S(v,\Gamma)$ contains $v$ and the edges incident to $v$, but does not contain the remaining vertices of these edges. Thus, $S(v,\Gamma)\subset\Gamma$ is an open set. Let, now, $\mathcal{S}_c(v,\Gamma)$ be the set of continuous functions $f:\Gamma\to\mathbb{R}$ such that the support of $f$ is compactly contained in $S(v,\Gamma)$. Let, additionally, for $v\in\mathcal{V}$,
$$N(v,\Gamma) := \{f\in C(\Gamma): \forall e,e'\!\!\in \mathcal{E}_v, f_e \in C^\infty(e), f_e^{(2k-1)}(v) = 0, f_e^{(2k)}(v) = f_{e'}^{(2k)}(v), k\in\mathbb{N}\}.$$
Now, define
$\mathcal{S}_c(\Gamma) = \{f\in C(\Gamma): \exists v\in\mathcal{V}, f\in \mathcal{S}_c(v,\Gamma) \cap N(v,\Gamma)\}$, which is
 a subalgebra of $C(\Gamma)$. Furthermore, by \citet[Propositions 3 and 4]{BSW2022}, $\mathcal{S}_c(\Gamma)\subset \bigcap_{k\in\mathbb{N}} \dot{H}^k$. Also observe that $\mathcal{S}_c(\Gamma)$ is nonempty since $0\in \mathcal{S}_c(\Gamma)$. We are now in a position to obain the subalgebra of $C(\Gamma)$ that contains constant functions and separate points:

\begin{Lemma}\label{lem:StoneWeierSpace}
Let $\Gamma$ be a compact metric graph and define 
$$\mathcal{A}(\Gamma) = \<1\> + D(\Gamma) + \mathcal{S}_c(\Gamma) = \{c+f+g: c\in\mathbb{R}, f\in D(\Gamma), g\in\mathcal{S}_c(\Gamma)\}.$$ 
Then, $\mathcal{A}(\Gamma)$ is dense in $C(\Gamma)$ with respect to the $\|\cdot\|_{C(\Gamma)}$ norm.
\end{Lemma}
\begin{proof}
	In view of Stone-Weierstrass theorem \cite[Theorem 2.4.11]{dudley2018real}, it is enough to show that $\mathcal{A}(\Gamma)$ is a subalgebra of $C(\Gamma)$ that contains the constant functions and separates points of $\Gamma$.
To this end, note that $\mathcal{A}(\Gamma)$ is clearly a vector space. Now, observe that given ${f_1,f_2 \in D(\Gamma)}$ and $g_1,g_2\in\mathcal{S}_c(\Gamma)$, then $f_1f_2, f_1g_1 \in D(\Gamma)$ and $g_1g_2\in \mathcal{S}_c(\Gamma)$. Further, let ${h_1,h_2\in \mathcal{A}(\Gamma)}$, so there exist $c_1,c_2\in\mathbb{R}$, $f_1,f_2\in D(\Gamma)$ and $g_1,g_2\in \mathcal{S}_c(\Gamma)$ such that ${h_i = c_i + f_i + g_i}$,${i=1,2}$. Hence, 
$$h_1h_2 = c_1c_2 + (c_1 f_2 + c_2f_1 + f_1f_2 + f_1g_2 + f_2g_1) + (c_1g_2 + c_2g_1 + g_1g_2)\in \mathcal{A}(\Gamma),$$
since $c_1c_2 \in \mathbb{R}$, $c_1 f_2 + c_2f_1 + f_1f_2 + f_1g_2 + f_2g_1\in D(\Gamma)$ and $c_1g_2 + c_2g_1 + g_1g_2\in\mathcal{S}_c(\Gamma)$. By construction, $\mathcal{A}(\Gamma)\subset C(\Gamma)$. Therefore, $\mathcal{A}(\Gamma)$ is a subalgebra of $C(\Gamma)$. Also, by construction, $\mathcal{A}(\Gamma)$ contains the constant functions. From the discussion in the beginning of this section, we already have that $\mathcal{A}(\Gamma)$ separates points from the interiors of edges from all the remaining points of the metric graph. 
Therefore, all we need to show is that $\mathcal{A}(\Gamma)$ separates the vertices of $\Gamma$ from the other vertices. 
To this end, we need to show that for each $v\in\mathcal{V}$, $\mathcal{S}_c(\Gamma)$ contains a nontrivial function with support on $S(v,\Gamma)$. 

We will now construct such functions. Take any $v\in\mathcal{V}$ and let $\delta := \min_{e\in \mathcal{E}_v} \ell_e>0$. 
Let $\psi \in C^\infty([0,\delta])$ be such that $\psi(x) = 1$ for $x\in [0,\delta/3]$ and that $\psi(x) = 0$ for $x\in [2\delta/3, \delta]$. We will now define a non-trivial function $f : \Gamma\to\mathbb{R}$ such that $f \in \mathcal{S}_c(\Gamma)$ and the support of $f$ is contained in $S(v,\Gamma)$. First, let $f(s) = 0$ if $s = (t,e)$ with $e\not\in\mathcal{E}_v$ and $t\in e$. Now, given any $e\in \mathcal{E}_v$ such that $\underline{e} = v$, define $f_e(t) = \psi(t)$ if $t\leq \delta$, and $f_e(t) = 0$ otherwise. Similarly, if $\bar{e} = v$, define $f_e(t) = \psi(\delta-t)$ if $t>\ell_e-\delta$ and $f_e(t) = 0$ otherwise. It is now clear that, by construction, $f$ is non-trivial, belongs to $\mathcal{S}_c(\Gamma)$ and has support in $S(v,\Gamma)$. 
\end{proof}


\begin{proof}[Proof of Proposition \ref{prp:strposdef}]
In view of Proposition \ref{prp:c0universal_sipd} and Remark \ref{rem:ispdimpliesspd} it is enough to show that $\varrho(\cdot,\cdot)$ is $c$-universal. 
Now, observe that by, e.g., \citet[Corollary 8.16]{janson_gaussian}, the reproducing kernel Hilbert space induced by $\varrho$ is the Cameron--Martin space associated to $u$. Further, by \citet[Proposition 8]{BSW2022}, the Cameron--Martin space associated to $u$ is $\dot{H}^\alpha$. 
Furthermore, it is clear that $\dot{H}^{\ceil{\alpha}} \subset \dot{H}^\alpha$, so ${\bigcap_{k\in\mathbb{N}} \dot{H}^k \subset \dot{H}^\alpha}$. 
Moreover, by the discussion above, we have ${\<1\>+D(\Gamma) \subset \bigcap_{k\in\mathbb{N}} \dot{H}^k}$ and ${\mathcal{S}_c(\Gamma) \subset \bigcap_{k\in\mathbb{N}} \dot{H}^k}$. Thus, ${\mathcal{A}(\Gamma)\subset \bigcap_{k\in\mathbb{N}} \dot{H}^k}$. 
Finally, by Lemma \ref{lem:StoneWeierSpace}, ${\mathcal{A}(\Gamma)}$ is dense in $C(\Gamma)$. 
Since ${\mathcal{A}(\Gamma)\subset \bigcap_{k\in\mathbb{N}} \dot{H}^k \subset \dot{H}^\alpha}$, 
we have that $\dot{H}^\alpha$ is dense in $C(\Gamma)$ with respect to the $\|\cdot\|_{C(\Gamma)}$ norm.
\end{proof}

\section{Proofs for the bridge representation}\label{app:proofs_edge}
First, we introduce additional notation. For two Hilbert spaces $E$ and $F$, 
we have the continuous embedding $E\hookrightarrow F$ if the inclusion map from $E$ to $F$ is continuous 
(i.e., $C>0$ exists such, that for every $f\in E$, $\|f\|_F \leq C \|f\|_E$). 
Further, we use the notation $E\cong F$ if a Hilbert space $Y$ exists and two spaces $\widetilde{E},\widetilde{F}\subset Y$ that are isometrically isomorphic to $E$ and $F$, respectively, such that  $\widetilde{E} \hookrightarrow \widetilde{F} \hookrightarrow \widetilde{E}$. Given $S\subset \Gamma$, we let $C_c(S)$ denote the set of continuous functions with support compactly contained in the interior of $S$.

We primarily work with weak differentiability in the $L_2(\Omega)$ sense, which we define here. 
A function ${v:e \to H(\Gamma), e\in\mathcal{E}}$ is weakly continuous in the $L_2(\Omega)$ sense if, for each ${w\in H(\Gamma)}$, 
the function ${s\mapsto \pE(w v(s))}$ is continuous, and $v$ is weakly differentiable at $s$ in the $L_2(\Omega)$ sense if 
$v'(s)\in H(\Gamma)$ exists such that, for each ${w\in H(\Gamma)}$ and each sequence $t_n\to s$ with ${t_n\neq s}$, ${\pE(w\nicefrac{(v(t_n)-v(s))}{(t_n-s)}) \to \pE(wv'(s))}$. 
Higher-order weak derivatives in the $L_2(\Omega)$ sense are defined inductively: 
for $k\geq 2$, $v$ has a $k$th order weak derivative at $s\in e$ if $v^{(k-1)}(t)$ exists for every $t\in e$ and it is weakly differentiable at $s$. 
Now, assume that $u$ has weak derivatives, in the $L_2(\Omega)$ sense, of orders $1,\ldots, p$ for $p\in\mathbb{N}$. 
Also assume that $u^{(j)}$, $j=0,\ldots,p$, is weakly continuous in the $L_2(\Omega)$ sense. 
Then, the weak derivatives are well-defined for each $s\in\Gamma$, and we say that $u$ is a differentiable Gaussian random field of order $p$. 
A Gaussian random field that is continuous in $L_2(\Omega)$ is said to be differentiable of order 0. 
					
Let $u$ be a differentiable Gaussian random field of order $\alpha-1$, $\alpha\geq 1$. 
Then, we let ${H_{\alpha,I}(\partial S) = \textrm{span}\{u_e(s), u_e'(s),\ldots, u_e^{(\alpha-1)}(s):
s\in \partial S, e\in \mathcal{E}_S \}}$ be the boundary Gaussian space on $S$ and 
${\mathcal{H}_{0,I}(S) = \mathcal{H}(S)\ominus \mathcal{H}_{\alpha,I}(\partial S)}$ be the interior Gaussian space on $S$, 
where $\mathcal{E}_S$ denotes the set of edges with a non-empty intersection with the interior of $S$.
		
Let $\varrho_M:\mathbb{R}\to\mathbb{R}$ be the Mat\'ern covariance function in \eqref{eq:matern_cov}. 
By Kolmogorov's extension theorem, a centered Gaussian process $U_{e}(t), t\in e$, with ${\Cov(U_{e}(t_1), U_e(t_2)) = \varrho_M(|t_1-t_2|)}$, $t_1,t_2\in e$ exists. 
We now show that the stationary Mat\'ern process $U_e$ on the interval $[0,\ell_e]$ has the same edge representation as the one for $u_e$ in Theorem \ref{thm:ReprTheoremEdge_Refined}. 
\begin{Theorem}\label{thm:EdgeReprStationary}
	Let $U_e(\cdot)$ be defined as above. We have the following representation:
	$$U_e(t) = V_{\alpha,0}(t) + \sum_{j=1}^{2\alpha} s_j(t)  \left(B^{\alpha}U_e\right)_j = V_{\alpha,0}(t) + \mv{s}^\top(t) B^\alpha U_e,\quad t\in e,$$
	where $V_{\alpha,0}(\cdot)$ is a Whittle--Mat\'ern bridge process on $e$ with parameters $(\kappa,\tau,\alpha)$ on the interval $[0,\ell_e]$, $V_{\alpha,0}(\cdot)$ is independent of $B^{\alpha}U_e$, and $\mv{s}(\cdot)$ is given by \eqref{eq:edge_repr_Solution}.
\end{Theorem}

\begin{proof}
Let 
$r_{B,\ell_e}(\cdot,\cdot)$ be the covariance function of the Whittle--Mat\'ern bridge process given by \eqref{eq:cov_func_whittle_matern_bridge} (i.e., $r_{B,\ell_e}(\cdot,\cdot)$ is the covariance function of $V_{\alpha,0}$). Then, 
$$\varrho_M(t_1-t_2) = r_{B,\ell_e}(t_1,t_2) + \begin{bmatrix}
	\mv{r}_1(t_1,0) & \mv{r}_1(t_1,\ell_e) 
\end{bmatrix}
\begin{bmatrix}
	\mv{r}(0,0) & \mv{r}(0,\ell_e) \\
	\mv{r}(\ell_e,0) & \mv{r}(\ell_e,\ell_e)
\end{bmatrix}^{-1}
\begin{bmatrix}
	\mv{r}_1(0,t_2) \\
	\mv{r}_1(\ell_e,t_2) 
\end{bmatrix}.$$
Thus,
\begin{equation}\label{eq:decomp_statio_matern_cov}
	\varrho_M(t_1-t_2) = r_{B,\ell_e}(t_1,t_2) + h(t_1,t_2),
\end{equation}
where 
\begin{equation}\label{eq:cov_boundary_part_statiomatern}
h(t_1,t_2) = 	\mv{s}(t_1)^\top \Cov(B^\alpha U_e) \mv{s}(t_2),
\end{equation}
which proves the result.	
\end{proof}

Let $\alpha\in\mathbb{N}$, and recall that $H_0^\alpha(e)$ is the closure of $C_c^\infty(e)$ with respect to the Sobolev norm $\|\cdot\|_{H^\alpha(e)}$, with 
inner product 
$$(u,v)_{\alpha,e} = (L^{\alpha/2}u, L^{\alpha/2}v)_{L_2(e)}, \quad u,v\in H^\alpha_0(e).$$

\begin{proof}[Proof of Lemma~\ref{lem:CM_edge_repr_bridge}]
We prove the result using the representation in Theorem~\ref{thm:EdgeReprStationary}. 
Thus, we begin by connecting the process $V_{\alpha,0}(\cdot)$ to the stationary Mat\'ern process on the interval $[0,\ell_e]$.
Let $V_{\alpha,0}(\cdot)$ be the process obtained from a stationary Mat\'ern process on $[0,\ell_e]$ from the representation in Theorem \ref{thm:EdgeReprStationary}. Then, $V_{\alpha,0}(\cdot)$ has covariance function $r_{B,\ell_e}(\cdot,\cdot)$ and is independent of $B^\alpha U_e$.
We now show that, for every $t\in e$, $B^\alpha r_{B,\ell_e}(\cdot, t) = \mv{0}$. 
As $V_{\alpha,0}(t)$ is independent of $B^\alpha U_e$, we have 
$$
	\pE(V_{\alpha,0}(t) U_e^{(k)}(0)) = 0 \quad \hbox{and}\quad \pE(V_{\alpha,0}(t) U_e^{(k)}(\ell_e)) = 0, \quad k=0,\ldots,\alpha-1.
$$
Therefore, for each $t,s\in e$ and $x\in\{0,\ell_e\}$,
\begin{align*}
	\partial_1^k r_{B,\ell_e}(x,t) &= \partial_1^k \pE(V_{\alpha,0}(t) U_e(x)) = \pE(V_{\alpha,0}(t) U_e^{(k)}(x)) = 0, \quad k=0,\ldots, \alpha-1,
\end{align*}
where we applied the fact that $U^{(k)}_e(\cdot)$ is the $L_2(\Omega)$ derivative to interchange the derivative and expectation signs, 
and the well-known fact that the stationary Mat\'ern process admits $L_2(\Omega)$ derivatives of orders $k=1,\ldots,\alpha-1$. 
This result indicates that, for each $t\in e$, ${B^\alpha r_{B,\ell_e}(\cdot, t) = \mv{0}}$. 
The same argument shows that, for every $t\in e$, $r_{B,\ell_e}(\cdot,t)$ is $\alpha-1$ times differentiable. 
In contrast, let  $h(\cdot,\cdot)$ be given by \eqref{eq:cov_boundary_part_statiomatern}. 
From the explicit expression \eqref{eq:cov_boundary_part_statiomatern}, we can readily check that, for every $t\in e$, $h(\cdot,t)$ is a solution to $(\kappa^2-\Delta)^\alpha h = 0$. 
Thus, by standard elliptic regularity, $h(\cdot,t) \in C^{\infty}(e)$. In particular, for every $t\in e$, $h(\cdot,t)\in H^\alpha(e)$.

Next, 
observe that if $\alpha=1$, 
$\varrho_M(\cdot-t)$ is the exponential covariance function, in which one can readily check that $\varrho_M(\cdot-t)$ is absolutely continuous 
(using the fundamental theorem of calculus on the intervals $[0,t]$ and $[t,\ell_e]$); 
thus, $\varrho_M(\cdot - t) \in H^1(e)$. 
If $\alpha>1$, it follows by \citet[Expression~(15), p. 32]{stein99} that for every $t\in e$, $\varrho_M(\cdot-t)\in C^{\alpha}(e)$. Thus, for every $t\in e$ and every $\alpha\in\mathbb{N}$, $\varrho_M(\cdot - t) \in H^\alpha(e)$. 

Therefore, since for every $t\in e$, $h(\cdot,t)$ and $\varrho_M(\cdot-t)$ belong to $H^\alpha(e)$, 
we have, by \eqref{eq:decomp_statio_matern_cov}, that for every $t\in e$, $r_{B,\ell_e}(\cdot,t)\in H^\alpha(e)$. 
Further, $B^\alpha r_{B,\ell_e}(\cdot,t) =\mv{0}$ for every $t\in e$. 
Hence, by \citet[Theorem 3.40]{mclean}, $r_{B,\ell_e}(\cdot,t) \in H_0^\alpha(e)$ for every $t\in e$. 

Let $\mathcal{C}_{\alpha,0} : L_2(e)\to L_2(e)$ be the integral operator induced  by $r_{B,\ell_e}(\cdot,\cdot)$:
$$
	(\mathcal{C}_{\alpha,0} f)(t) = \int_{e} r_{B,\ell_e}(s,t) f(s) ds,\quad t\in e,\quad f\in L_2(e).
$$
Recall that $r_{B,\ell_e}(\cdot,\cdot)$ is the covariance function of $V_{\alpha,0}(\cdot)$, 
which implies (by Mercer's theorem, \citet{Steinwart2012}) that $\mathcal{C}_{\alpha,0}$ is the covariance operator of $V_{\alpha,0}(\cdot)$.	
Let $(\mathcal{H}_{\alpha,0}(e), \<\cdot,\cdot\>_{\mathcal{H}_{\alpha,0}(e)})$ be the Cameron--Martin space associated to $V_{\alpha,0}(\cdot)$, which coincides with the completion of $\mathcal{C}_{\alpha,0}(L_2(e))$ with respect to the inner product ${\<\widetilde{f},\widetilde{g}\>_{\mathcal{H}_{\alpha,0}(e)} = (\mathcal{C}_{\alpha,0} f, g)_{L_2(e)}},$ where ${\widetilde{f},\widetilde{g}\in \mathcal{C}_{\alpha,0}(L_2(e))}$, with $\widetilde{f} = \mathcal{C}_{\alpha,0}f$ and $\widetilde{g} = \mathcal{C}_{\alpha,0}g$, for some $f,g\in L_2(e)$ \citep[see, e.g.][p.49]{Bogachev1998}.

We can now prove the statement of the lemma in four steps. 
First, we demonstrate that the range of $\mathcal{C}_{\alpha,0}$ is contained in $H_0^{\alpha}(e)$. 
Second, we establish that the range of $\mathcal{C}_{\alpha,0}$ contains $C_c^\infty(e)$.  
Third, we identify the Cameron--Martin inner product with $(\cdot,\cdot)_{\alpha,e}$. 
Finally, we conclude that the Cameron--Martin associated with $V_{\alpha,0}$ is equal to $(H_0^\alpha(e), (\cdot,\cdot)_{\alpha,e})$.

\noindent \textbf{Step 1.} We already showed that $r_{B,\ell_e}(\cdot, t)\in H_0^{\alpha}(e)$, $t\in e$. 
Let $\alpha>1$ (the case $\alpha=1$ can be easily obtained through explicit expressions). 
For $0\leq k \leq \alpha-1$, we have $r_{B,\ell_e}(\cdot, t) \in C^\alpha(e)$. 
Hence, as $e$ is compact, for every $f\in L_2(e)$, we obtain
\begin{equation}\label{eq:diff_under_integral_CM}
	(\mathcal{C}_{\alpha,0} f)^{(k)}(t) = \int_{e} \partial_2^k r_{B,\ell_e}(s,t) f(s) ds,\quad t\in e, k=0,\ldots,\alpha,
\end{equation}
where $\partial_2^k g(\cdot,\cdot)$ denotes the $k$th derivative of $g$ with respect to the second variable.
Next, $(t,s)\mapsto \partial_t^k \varrho_M^{(\alpha)}(t-s)$ and $(t,s)\mapsto \partial_2^kh(s,t)$ are 
bounded on $e\times e$, where $\partial_t^k$ is the $k$th derivative with respect to the variable $t$ and $k=0,\ldots,\alpha$. 
Indeed, the boundedness of $\partial_t^k\varrho_M(\cdot-\cdot)$, $k=0,\ldots,\alpha$, 
directly follows from the fact that $\varrho_M(\cdot)\in C^\alpha(e)$ and $e\times e$ is compact, 
whereas the boundedness of $\partial_2^k h(\cdot,t)$, $t\in e$, follows from \eqref{eq:cov_boundary_part_statiomatern} and the fact that $e\times e$ is compact. 
Thus, $r_{B,\ell_e}^{(k)}(\cdot,\cdot)$ is bounded on $e\times e$, $k=0,\ldots,\alpha$. 
Therefore, by \eqref{eq:diff_under_integral_CM} and the boundedness of $r_{B,\ell_e}^{(k)}(\cdot,\cdot)$, $\mathcal{C}_{\alpha,0}f \in H^\alpha(e)$ 
for every $f\in L_2(e)$.

Moreover, as $B^\alpha r_{B,\ell_e}(\cdot,t) = \mv{0}, t\in e$, it follows by simple evaluation that ${B^\alpha \mathcal{C}_{\alpha,0} f = \mv{0}}$. 
Thus, for every $f\in L_2(e) \Rightarrow \mathcal{C}_{\alpha,0} f \in H_0^\alpha(e)$. 
That is, $\mathcal{C}_{\alpha,0}(L_2(e)) \subset H^\alpha_0(e)$.

\noindent \textbf{Step 2.} 
It is well-known that $\varrho_M$
is the free-field Green function of $L^\alpha$ in the distributional sense. 
That is, $\varrho_M$ solves
$L^\alpha \varrho_M(\cdot- t) = \delta_t,$
where $\delta_t(\cdot)$ is the Dirac delta measure. 
Thus, for any $g\in C_c^\infty(\mathbb{R})$, 
\begin{equation}\label{eq:greenfunctionsolution}
	\int_{-\infty}^\infty \varrho_M(s - t) L^\alpha g(s) ds = g(t),
\end{equation}
where $L^\alpha = (\kappa^2-\Delta)^\alpha$ acts on the variable $s$. 
Let $g\in C_c^\infty(e)$ and define $\widehat{g} = L^{\alpha} g \in C^\infty_c(e)$. 
For every $t\in e$, $h(\cdot,t)\in C^\infty(e)$, ${L^\alpha h(\cdot,t) = 0}$ and $\varrho_M(\cdot-t) = h(\cdot,t) + r_{B,\ell_e}(\cdot,t)$. 
Thus, as $g\in C_c^\infty(e)$, we have that for every $t\in e$,
\begin{equation}\label{eq:invCovOperatorStationary}
\begin{aligned}
	\mathcal{C}_{\alpha,0} \widehat{g}(t) &= \int_e  r_{B,\ell_e}(s,t) L^\alpha g(s) ds = \int_e  r_{B,\ell_e}(s,t) L^\alpha g(s) ds + \int_e L^\alpha h(s,t)  g(s) ds\\
	&= \int_e  r_{B,\ell_e}(s,t) L^\alpha g(s) ds + \int_e  h(s,t) L^\alpha g(s) ds\\
	&= \int_e \varrho_M(s-t) L^\alpha g(s) ds = g(t).
\end{aligned}
\end{equation}
Therefore, the image of $\mathcal{C}_{\alpha,0}$ contains $C_c^\infty(e)$; thus, 
$C_c^\infty(e) \subset \mathcal{C}_{\alpha,0}(L_2(e)) \subset H^\alpha_0(e).$

\noindent \textbf{Step 3.} Equation~\eqref{eq:invCovOperatorStationary} directly shows that, for any $g\in C_c^\infty(e)$, $\mathcal{C}_{\alpha,0} L^\alpha g = g$. 
Similarly, together with differentiation under the sign of the integral, we also obtain that $L^\alpha \mathcal{C}_{\alpha,0} g = g$ for every $g\in C_c^\infty(e)$. 
Further, $C_c^\infty(e)\subset \mathcal{C}_{\alpha,0}(L_2(e))$; thus, given $\widetilde{f}, \widetilde{g} \in \mathcal{C}_{\alpha,0}(L_2(e))$,  
functions $f,g \in L_2(e)$ exist such that $\widetilde{f} = \mathcal{C}_{\alpha,0} f$ and $\widetilde{g} = \mathcal{C}_{\alpha,0}g$, 
so that $L^\alpha \widetilde{f} = f$ and $L^\alpha \widetilde{g} = g$. 
Let $\<\cdot,\cdot\>_{\mathcal{H}_{\alpha,0}(e)}$ be the inner product of $\mathcal{H}_{\alpha,0}(e)$. 
We have that
$$
	\<\widetilde{f},\widetilde{g}\>_{\mathcal{H}_{\alpha,0}(e)} = (\mathcal{C}_{\alpha,0} f, g)_{L_2(e)} = (\widetilde{f}, L^\alpha \widetilde{g})_{L_2(e)} = (\widetilde{f}, \widetilde{g})_{\alpha,e}.
$$
This demonstrates that $\<\cdot,\cdot\>_{\mathcal{H}_{\alpha,0}(e)}$ coincides with $(\cdot, \cdot)_{\alpha,e}$ on $C_c^\infty(e)\times C_c^\infty(e)$. 
Next, the norm $\|\cdot\|_{\alpha,e}$ induced by $(\cdot,\cdot)_{\alpha,e}$ is equivalent to the Sobolev norm $\|\cdot\|_{H^\alpha(e)}$ on $H^\alpha_0(e)$ (see, e.g., \citet[Lemma 3.1]{Thomee}).
Because $H_0^\alpha(e)$ is the closure of $C_c^\infty(e)$ with respect to the Sobolev norm $\|\cdot\|_{H^\alpha(e)}$, it is also the closure of $C_c^\infty(e)$ with respect to the norm $\|\cdot\|_{\alpha,e}$. In particular, $H_0^\alpha(e)$ is closed with respect to $\|\cdot\|_{\alpha,e}$. Therefore, $\<\cdot,\cdot\>_{\mathcal{H}_{\alpha,0}(e)}$ coincides with $(\cdot, \cdot)_{\alpha,e}$ on $H_0^\alpha(e)\times H_0^\alpha(e)$.

\noindent \textbf{Step 4.} 
We claim that $\mathcal{H}_{\alpha,0}(e) = (H_0^\alpha(e), (\cdot,\cdot)_{\alpha,e})$. Indeed, $\<\cdot,\cdot\>_{\mathcal{H}_{\alpha,0}(e)}$ coincides with $(\cdot, \cdot)_{\alpha,e}$ on $H_0^\alpha(e)\times H_0^\alpha(e)$ and $H_0^\alpha(e)$ is the closure of $C_c^\infty(e)$ with respect to the norm $\|\cdot\|_{\alpha,e}$. 
This shows that $C_c^\infty(e)$ is dense in $H^\alpha_0(e)$ with respect to the Cameron--Martin inner product. 
Finally, as $C_c^\infty(e) \subset \mathcal{C}_{\alpha,0}(L_2(e)) \subset H^\alpha_0(e)$, it directly follows (from the first observation in Step 4 along with the previous considerations) that $\mathcal{H}_{\alpha,0}(e) = (H_0^\alpha(e), (\cdot,\cdot)_{\alpha,e})$. 
\end{proof}

We have the following result from \citet[][Propositions 3, 4 and 8]{BSW2022}. 
			
\begin{Proposition}\label{prp:charHdot}
	We have $\dot{H}^1 \cong H^1(\Gamma)$, $\dot{H}^2 \cong \widetilde{H}^2(\Gamma)\cap \mathcal{K}_1(\Gamma) \cap C(\Gamma)$, and
	$$
	\dot{H}^k \cong \left\{f\in \widetilde{H}^k(\Gamma): D^{2\left\lfloor \nicefrac{k}{2}\right\rfloor} f \in C(\Gamma), \,\forall m = 0,\ldots, \left\lfloor\nicefrac{(k-2)}{2} \right\rfloor, D^{2m} f\in\dot{H}^2 \right\}.
	$$
	Furthermore, the norms $\|\cdot\|_k$ and $\|\cdot\|_{\widetilde{H}^k(\Gamma)}, k\in\mathbb{N},$ are equivalent and $\dot{H}^\alpha$ is the Cameron--Martin space associated to the solution $u$ of \eqref{eq:Matern_spde}.
\end{Proposition}

We need the following two technical lemmata
to prove Theorem~\ref{thm:ReprTheoremEdge_Refined}, 

\begin{Lemma}\label{lem:charHdot_edge}
There exists an isometric isomorphism between $(\mathcal{H}_{0,I}(e), (\cdot, \cdot)_\alpha)$
and $(H_0^\alpha(e), (\cdot,\cdot)_{\alpha,e})$.
\end{Lemma}

\begin{proof}
First, by \citet[Theorem 7]{BSW_Markov}, we have $\mathcal{H}_{0,I}(e) \cong \dot{H}_0^\alpha(e)$,
where $\dot{H}_0^\alpha(e)$ is the completion of $C_c(e)\cap \dot{H}^\alpha(e)$ 
with respect to the Sobolev norm $\|\cdot\|_{H^\alpha(e)}$. 
Next, from Proposition \ref{prp:charHdot}, $\dot{H}^\alpha(e) \cap C_c(e) = H^\alpha(e)\cap C_c(e)$.
Finally, it follows from the well-known approximation by smooth functions in Sobolev spaces 
(because we are on an edge; see, e.g., \cite{evans2010partial}) that $C_c^\infty(e)$ is dense in $H^\alpha(e)\cap C_c(e)$
with respect to the Sobolev norm $\|\cdot\|_{H^\alpha(e)}$. 
Therefore, as $\dot{H}_0^\alpha(e)$ is the completion of $C_c(e)\cap \dot{H}^\alpha(e)$ 
with respect to the Sobolev norm $\|\cdot\|_{H^\alpha(e)}$, it follows that $\dot{H}_0^\alpha(e)$ 
is the completion of $C^\infty_c(e)$ with respect to the Sobolev norm $\|\cdot\|_{H^\alpha(e)}$, which is $H_0^\alpha(e)$. 
In particular, $\mathcal{H}_{+,0}(e)\cong H_0^\alpha(e)$. 
It remains to establish that the bilinear form $(\cdot,\cdot)_{\alpha,e}$ is isometric to $(\cdot,\cdot)_\alpha$. 
Indeed, given $u\in H_0^\alpha(e)$ and $v\in C_c^\infty(e)$, we can extend $u$ and $v$ to $\Gamma$ by defining
them to be zero on $\Gamma\setminus e$. 
Let $\widetilde{u}$ and $\widetilde{v}$ be these extensions. 
By \citet[Theorem 7]{BSW_Markov}, $\widetilde{u}, \widetilde{v} \in \mathcal{H}_{0,I}(e)$, and because $\alpha\in\mathbb{N}$,
$$
	(u,v)_{\alpha,e} = (u, L^\alpha v)_{L_2(e)} = (\widetilde{u}, L^\alpha \widetilde{v})_{L_2(\Gamma)} = (\widetilde{u}, \widetilde{v})_\alpha.
$$
Now, $C_c^\infty(e)$ is dense in $H_0^\alpha(e)$ with respect to the Sobolev norm $\|\cdot\|_{H^\alpha(e)}$
and from Proposition \ref{prp:charHdot}, the norms $\|\cdot\|_\alpha = \sqrt{(\cdot, \cdot)_\alpha}$ and the 
Sobolev norm $\|\cdot\|_{\widetilde{H}^\alpha(\Gamma)}$ are equivalent. 
Hence, the bilinear form is continuous with respect to the Sobolev norm $\|\cdot\|_{H^\alpha(e)}$. 
Thus, we can uniquely extend $(\cdot,\cdot)_{\alpha,e}$ to $H_0^\alpha(e)\times H_0^\alpha(e)$. 
Finally, for $u,v\in H_0^\alpha(e)$ we let $\widetilde{u}$ and $\widetilde{v}$ be their
extensions to $\Gamma\setminus e$ (as zero). 
Then, $(u,v)_{\alpha,e} = (\widetilde{u}, \widetilde{v})_\alpha$, and because $H_0^\alpha(e)\cong \mathcal{H}_{0,I}(e)$, this concludes the proof.
\end{proof}

\begin{Lemma}\label{lem:sol_ODE_edge}
Let $e = [0,\ell_e]$, with $0<\ell_e<\infty$ and let $\mv{S}_e(\cdot)$ be the vector-valued function defined by \eqref{eq:edge_repr_Solution}, 
where $\alpha\in\mathbb{N}, \kappa>0$ and $\tau>0$. 
Further, let $s_1(\cdot), \ldots s_{2\alpha}(\cdot)$ be the components of $\mv{S}_e(\cdot)$. 
Then, $\{s_j(\cdot): j = 1,\ldots,2\alpha\}$ is a set of linearly independent solutions of the linear equation 
$(\kappa^2 - \Delta)^\alpha s(\cdot) = 0$ and $B^\alpha \mv{S}_e = \mv{I}$.
\end{Lemma}

\begin{proof}
Begin by defining
$$
	\mv{\rho}(\cdot) = (\varrho_M(\cdot),\ldots, \varrho_M^{(\alpha-1)}(\cdot), \varrho_M(\cdot-\ell_e),\ldots,\varrho_M^{(\alpha-1)}(\cdot-\ell_e)).
$$ 
By the explicit expressions of $\varrho_M(\cdot)$, we can readily check that the coordinates of $\mv{\rho}$ are $2\alpha$ 
linearly independent solutions to the linear differential equation $(\kappa^2-\Delta)^\alpha s = 0$. 
Recall the definition of $\mv{r}(\cdot,\cdot)$ in \eqref{eq:R_matrix_edge_repr} and define
$$	
	\mv{M} = \begin{bmatrix}
	\mv{r}(0,0) & \mv{r}(0,\ell_e) \\
	\mv{r}(\ell_e,0) & \mv{r}(\ell_e,\ell_e)
	\end{bmatrix}.
$$
The operator $B^\alpha$ is linear and $B^\alpha \mv{\rho} = \mv{M}$, 
where for a function ${\mv{v}(\cdot)=(v_1(\cdot),\ldots, v_m(\cdot))}$, 
$m\in\mathbb{N}$, $B^\alpha \mv{v}$ is the matrix whose $j$th row is given by $B^\alpha v_j$, $j=1,\ldots, m$. 
Let $\mv{s}(\cdot) = \mv{\rho}(\cdot)\mv{A}$, where the matrix $\mv{A}$ is to be determined. 
Then, the conditions $B^\alpha s_j = e_j$ to ensure $B^\alpha \mv{S}_e = \mv{I}$ 
translate to our current functions as $\mv{I} = B^\alpha \mv{s} =  (B^\alpha \mv{\rho}) \mv{A}=  \mv{M}\mv{A}$. 
Therefore, $\mv{A} = \mv{M}^{-1}$ and $\mv{s}(\cdot) =  \mv{\rho}(\cdot)\mv{M}^{-1}$, which concludes the proof.
\end{proof}

\begin{proof}[Proof of Theorem \ref{thm:ReprTheoremEdge_Refined}]
Let $\dot{H}_0^\alpha(e) = C_c(e)\cap \dot{H}^\alpha(e)$, where $\dot{H}^\alpha(e) = \{h|_e: h\in \dot{H}^\alpha\}$. 
From \citet[Theorem 9]{BSW_Markov}, we obtain an edge representation for $u$ and 
Lemma~\ref{lem:sol_ODE_edge} provides as explicit expression for $\mv{s}_e(\cdot)$. 
Further, by Lemma \ref{lem:CM_edge_repr_bridge}, 
$(H_0^\alpha(e),  (\cdot,\cdot)_{\alpha,e})$ is the Cameron--Martin space associated with the Whittle--Mat\'ern bridge process. 
Thus, to conclude the proof it is enough to show that ${(\dot{H}_0^\alpha(e), (\cdot,\cdot)_{\alpha,e})\cong (H_0^\alpha(e),  (\cdot,\cdot)_{\alpha,e})}$. 
By \citet[Theorem~8]{BSW_Markov}, $\mathcal{H}_{0,I}(e)\cong \dot{H}^\alpha_0(e)$; therefore, 
it is enough to show that we have the identification $(\mathcal{H}_{0,I}(e), (\cdot,\cdot)_{\alpha,e})\cong (H_0^\alpha(e),  (\cdot,\cdot)_{\alpha,e})$. 
This follows directly from Lemma~\ref{lem:charHdot_edge}.
\end{proof}

\begin{proof}[Proof of Theorem~\ref{cor:conditional_dists}]
If the degree of a vertex of the edge $e$ is 1, then by Theorem~\ref{thm:regularity}.\ref{thm:regularity:item:Kirchhoff}, for $0<k\leq\alpha-1$ such that $k$ is odd,  $u_e^{(k)} = 0$ a.s. Thus, we can only condition on odd-order derivatives being zero for the vertex.
Now, the results follows from Theorem \ref{thm:ReprTheoremEdge_Refined} and Theorem \ref{thm:EdgeReprStationary}, since $v_{\alpha,0}(\cdot)$ and $V_{\alpha,0}(\cdot)$ are independent of $B^\alpha u_e$ and $B^\alpha U_e$, respectively. 
\end{proof}

\begin{proof}[Proof of Proposition \ref{prp:Whittle_Matern_bridge_prop}]
The expression for the covariance function of the Whittle--Mat\'ern bridge process follows directly from Definition \ref{def:WMB} and conditioning of normal random vectors.
From the proof of Theorem \ref{thm:ReprTheoremEdge_Refined}, we have that the process $v_{\alpha,0}(\cdot)$ in the representation given in Theorem \ref{thm:ReprTheoremEdge_Refined} is a Whittle--Mat\'ern bridge process. The remaining properties of the Whittle--Mat\'ern bridge process are, thus, direct consequences of Theorem \ref{thm:ReprTheoremEdge_Refined}.
\end{proof}
		
\section{Proofs for the conditional representation}\label{app:proofs_conditional}
This section aims to provide a conditional representation for Whittle--Mat\'ern fields on metric graphs, 
obtained as solutions to \eqref{eq:Matern_spde} for any $\alpha\in\mathbb{N}$. 
Thus, we extend the results in \cite[Section~6]{BSW_Markov} for integer $\alpha\geq 1$ when $\kappa$ and $a$ are constant functions, with $a\equiv 1$. 
The construction of the conditioning and several proofs are similar to their counterparts in \cite{BSW_Markov}; however, we sometimes repeat them for clarity.

The strategy to obtain the conditional representation is as follows. 
We start by defining a process, such that the Cameron--Martin space associated with a field, denoted by $\widetilde{u}(\cdot)$, 
consisting of independent copies of this process on each edge, 
contains the Cameron--Martin space associated with the solution to \eqref{eq:Matern_spde}, denoted by $u(\cdot)$. 
Then, we establish that $u(\cdot)$ can be obtained by projecting $\widetilde{u}(\cdot)$ to some Gaussian space. 
Thus, we can write $\widetilde{u}(\cdot)$ as $u(\cdot)$ plus the projection of $\widetilde{u}(\cdot)$ to a space which is orthogonal to $u(\cdot)$. 
Hence, $u(\cdot)$ can be obtained from $\widetilde{u}(\cdot)$ by conditioning the additional term to be equal to zero. 
Finally, we demonstrate that conditioning this additional term to zero is equivalent to conditioning $\widetilde{u}(\cdot)$ to satisfy Kirchhoff conditions: 
continuity at the vertices, the sum of odd-order directional derivatives being equal to zero, and continuity of the even-order derivatives. 

We define the Cameron--Martin version of the boundaryless Whittle--Mat\'ern process on an interval as follows.
In Appendix \ref{app:bdlessaux}, we show that this process coincides with the boundaryless Whittle--Mat\'ern process of Definition \ref{def:cov_based_bdlessWM}.

\begin{Definition}\label{def:bdlessWM}
	We define the Cameron--Martin (CM) boundaryless Whittle--Mat\'ern process on $[0,\ell]$, $\ell\in (0,\infty)$, 
	with characteristics $\alpha\in\{1,2,3,\ldots\}$, $\kappa>0$, and $\tau>0$ 
	as the centered Gaussian process with the Cameron--Martin space given by $(H^\alpha([0,\ell]), \|\cdot\|_{\alpha,\kappa,\tau})$, 
	where the norm is induced by the inner product
	$$
		\<f,g\>_{\alpha,\kappa,\tau} = \tau^2 \sum_{m=0}^\alpha {\alpha \choose m} \kappa^{2(\alpha-m)} \int_0^\ell f^{(m)}(x) g^{(m)}(x) dx,\quad f,g\in H^\alpha([0,\ell]),
	$$
	if $\alpha$ is even, and
	$\<f,g\>_{\alpha,\kappa,\tau} =\kappa^2\<f,g\>_{\alpha-1,\kappa,\tau} + \<f',g'\>_{\alpha-1,\kappa,\tau}$, $f,g\in H^\alpha([0,\ell]),$
	if $\alpha$ is odd, where $\<f,g\>_{0,\kappa,\tau} = \tau^2(f,g)_{L_2(\Gamma)}$.
\end{Definition}


\begin{Lemma}\label{lem:bdlessWM}
	The CM boundaryless Whittle--Mat\'ern process on $[0,\ell]$, $0<\ell<\infty$, with characteristics $\alpha\in\{1,2,3,\ldots\}$, $\kappa >0$ and $\tau>0$, given in Definition \ref{def:bdlessWM}, exists.
\end{Lemma}

\begin{proof}
By the trace theorem \cite[Theorem 1, p.272]{evans2010partial}, for every $x\in [0,\ell]$, ${C_x>0}$ exists such that, for every $f\in H^\alpha([0,\ell])$,
$$
	|f(x)|\leq C_x \|f\|_{H^{1}([0,x])} \leq C_x \|f\|_{H^1([0,\ell])} \leq \widetilde{C}_x \|f\|_{\alpha, \kappa, \tau},
$$
where $\alpha\in\{1,2,3,\ldots\}$, $\kappa >0$, $\tau>0$, and $\widetilde{C}_x>0$ is a constant that may depend on $x, \kappa, \tau, \alpha$, but not on $f$. 
Therefore, the point evaluation on $(H^\alpha([0,\ell]), \|\cdot\|_{\alpha,\kappa,\tau})$ is continuous, 
which proves that $(H^\alpha([0,\ell]), \|\cdot\|_{\alpha,\kappa,\tau})$ is a reproducing kernel Hilbert space. 
	
Let $\varrho_{\alpha,\kappa,\tau}(\cdot,\cdot)$ be the (unique) reproducing kernel from $(H^\alpha([0,\ell]), \|\cdot\|_{\alpha,\kappa,\tau})$. 
The CM boundaryless Whittle--Mat\'ern process on $[0,\ell]$, $0<\ell<\infty$, with  $\alpha\in\mathbb{N}$, $\kappa >0$ and $\tau>0$ is, therefore, a centered Gaussian process with covariance function $\varrho_{\alpha,\kappa,\tau}(\cdot,\cdot)$, which exists by Kolmogorov's extension theorem.
\end{proof}

\begin{Proposition}\label{prp:bdlessWM_Markov}
The CM boundaryless Whittle--Mat\'ern process on $[0,\ell]$, $0<\ell<\infty$, with characteristics $\alpha\in\{1,2,3,\ldots\}$, $\kappa >0$ and $\tau>0$, 
is a Markov random field of order $\alpha$ in the sense of \cite[Definition 2]{BSW_Markov}.
\end{Proposition}
\begin{proof}
Observe that $(H^\alpha([0,\ell]), \|\cdot\|_{\alpha,\kappa,\tau})$ is local. 
Therefore, by \cite[Theorem 1]{BSW_Markov}, it directly follows that the CM boundaryless Whittle--Mat\'ern process is a Markov random field.
Then, we can mimic the proof of \cite[Theorem 6]{BSW_Markov} to find that the CM boundaryless Whittle--Mat\'ern process on $[0,\ell]$ with characteristic $\alpha\in\mathbb{N}$ is a Markov random field of order $\alpha$.
\end{proof}

Further technical results regarding the CM boundaryless Whittle--Mat\'ern processes are presented in Appendix \ref{app:bdlessaux}.

\begin{Remark}
By Definition \ref{def:bdlessWM} and \citet[Remark 2 and Lemma 3]{BSW_Markov}, for $\alpha=1$, 
the CM boundaryless Whittle--Mat\'ern process is given by a Whittle--Mat\'ern process on $[0,\ell]$ with Neumann boundary conditions. 
However, this is not the case for $\alpha>1$.
\end{Remark}

Fix $\alpha\in\mathbb{N}$, $\kappa>0$ and $\tau>0$ and let $\{\widetilde{u}_e:e\in\mathcal{E}\}$ be a family of independent CM boundaryless 
Whittle--Mat\'ern processes on $e=[0,\ell_e]$ with characteristics $\alpha$, $\kappa$ and $\tau$. 
We introduce a process on $\Gamma$ denoted by $\widetilde{u}$ that acts on each edge $e\in\mathcal{E}$ as $\widetilde{u}_e$. 
More precisely, we define, for $s=(t,e)\in\Gamma$, the field $\widetilde{u}$ as $\widetilde{u}(s) := \widetilde{u}_e(t)$. 

\begin{Proposition}\label{prp:CM_indep_field_bdless_WM}
Let $\widetilde{u}$ be defined as above. The Cameron--Martin space associated with $\widetilde{u}$ is given by $\widetilde{H}^\alpha(\Gamma)$ with the inner product
$$
	\<f,g\>_{\alpha,\Gamma} = \sum_{e\in\mathcal{E}} \<f_e,g_e\>_{\alpha,\kappa,\tau}.
$$
\end{Proposition}
	
\begin{proof}
By Definition \ref{def:bdlessWM}, the Cameron--Martin space associated with $\widetilde{u}_e$ is $H^\alpha(e)$ with the inner product $\<f_e,g_e\>_{\alpha,\kappa,\tau}$. The proof is completed by applying \cite[Lemma~8 in Appendix A]{BSW_Markov}.
\end{proof}

\begin{Remark}\label{rem:CM_kirk_bdless}
Let $\alpha\in\mathbb{N}$ and recall the definitions of $\dot{H}^\alpha$ and the norm $\|\cdot\|_{\alpha}$ in Section~\ref{sec:construction}. 
Let $\|\cdot\|_{\alpha,\Gamma}$ be the norm in $\widetilde{H}^\alpha(\Gamma)$ associated with $\<\cdot,\cdot\>_{\alpha,\Gamma}$. 
By integration by parts, if $f\in \dot{H}^\alpha$, then $\|f\|_{\alpha} = \|f\|_{\alpha,\Gamma}.$
\end{Remark}

We start by obtaining a representation of $u$, the solution to \eqref{eq:Matern_spde}, in terms of $\widetilde{u}$.

\begin{Proposition}\label{prp:CondRepr_WM_alpha}
Let $\widetilde{u}$ be the field obtained by joining independent CM boundaryless Whittle--Mat\'ern processes (see Definition \ref{def:bdlessWM}) on each edge. Therefore, the Whittle--Mat\'ern field given by the solution of \eqref{eq:Matern_spde} with positive integer $\alpha$ can be obtained as
$$
	u(s) = \sum_{e\in\mathcal{E}} v_{e,0}(s) + \mv{s}_{e}^\transp(s)T_{\alpha,e} \widehat{B}^\alpha\widetilde{u},
$$
where $v_{e,0}(\cdot)$ is the Whittle--Mat\'ern bridge with parameters $(\kappa,\tau,\alpha)$ defined on the edge $e$, $\mv{s}_{e}$ is given by \eqref{eq:edge_repr_Solution} and for each $e\in\mathcal{E}$, $T_{\alpha,e}$ is a matrix.
\end{Proposition}

\begin{proof}
Recall that $H^\alpha_0([0,\ell])$ is the closure of $C_c^\infty([0,\ell])$ with respect to the Sobolev norm. 
That is, $H^\alpha_0([0,\ell])$ is the space of functions $f:[0,\ell]\to\mathbb{R}$ that admit $\alpha$ weak derivatives satisfying $f^{(k)}(0) = f^{(k)}(\ell) = 0$, 
for $k=0,\ldots,\alpha-1$. 

For each $e\in\mathcal{E}$, we can apply Proposition~\ref{prp:EdgeReprBdlessProcInterval} in Appendix~\ref{app:bdlessaux} to $\widetilde{u}_e$ 
to obtain that ${\widetilde{u}_e(t) = v_{e,0}(t) + \mv{s}_e^\transp(t) B^\alpha \widetilde{u}_e}$ for $t\in e$,
where $v_{e,0}(\cdot)$ is the Whittle--Mat\'ern bridge process on the edge $e$ and $\mv{s}_e$ is given by \eqref{eq:edge_repr_Solution}. 
Further, $v_{e,0}$ is independent of $B^\alpha \widetilde{u}_e$. 
We can obtain a representation on the entire metric graph $\Gamma$ from these edge representations on all edges. 
To show this, we extend $v_{e,0}(\cdot)$ and $\mv{s}_e(\cdot)$ as zero on $\Gamma\setminus e$. 
Thus, for $s=(t,e)$, $v_{e',0}(s) = 0$ if $e\neq e'$ and $v_{e',0}(s) = v_{e,0}(t)$ if $e=e'$. 
We handle the functions $\mv{s}_e(\cdot)$ similarly. 
Therefore, 
\begin{equation}\label{eq:bridge_repr_indep_bdless}
	\widetilde{u}(s) = \sum_{e\in\mathcal{E}} v_{e,0}(s) + \mv{s}_{e}^\transp(s) B^\alpha \widetilde{u}_{e}, \quad s\in\Gamma.
\end{equation}

By Remark \ref{rem:CM_kirk_bdless}, $(\dot{H}^\alpha, (\cdot,\cdot)_{\alpha})$ is a closed subspace of $(\widetilde{H}^\alpha(\Gamma), \<\cdot,\cdot\>_{\alpha,\Gamma})$. 
Let $H_{\widetilde{u}}(\Gamma)$ be the Gaussian space induced by $\widetilde{u}$, that is, the $L_2(\Omega)$-closure of $\textrm{span}\{\widetilde{u}(t,e): e\in\mathcal{E}, t\in e\}$, and let $\Pi_K:\widetilde{H}^\alpha(\Gamma)\to \dot{H}^\alpha$ be the $\<\cdot, \cdot\>_{\alpha,\Gamma}$-orthogonal projection from $\widetilde{H}^\alpha(\Gamma)$ onto $\dot{H}^\alpha$. 
Define the isometric isomorphism $\Phi:H_{\widetilde{u}}(\Gamma) \to \widetilde{H}^\alpha(\Gamma)$
associated to the linear Gaussian space $H_{\widetilde{u}}(\Gamma)$ by
$\Phi(h)(t,e) = \pE(h\widetilde{u}_e(t))$, for $s = (t,e)\in\Gamma$.


We now show that the Whittle--Mat\'ern field with positive integer $\alpha$ can be obtained as a projection of $\widetilde{u}$ on a suitable space.
To this end, we introduce the following additional notation. 
We let $\widetilde{\varrho}(\cdot,\cdot)$ denote the covariance function of the field $\widetilde{u}$, 
which is also the reproducing kernel of $(\widetilde{H}^\alpha(\Gamma), \<\cdot,\cdot\>_{\alpha,\Gamma})$. 
We now project the reproducing kernel on $\dot{H}^\alpha$ to obtain a corresponding reproducing kernel there. 
More precisely, we define the function  $\varrho:\Gamma\times \Gamma\to \mathbb{R}$ as
$\varrho(s,\cdot) = \varrho((t,e),\cdot) = \Pi_K(\widetilde{\varrho}(s,\cdot)) = \Pi_K(\widetilde{\varrho}((t,e),\cdot)), s = (t,e)\in\Gamma$.
By the same arguments as in \cite[Section 6]{BSW_Markov}, 
we find that $\varrho(\cdot,\cdot)$ is a reproducing kernel for $(\dot{H}^\alpha,(\cdot,\cdot)_{\alpha})$, 
and in a similar manner, this allows us to define a new field:
$$
	u(s) := \Phi^{-1}(\varrho(s,\cdot)), \quad s\in\Gamma.
$$
Therefore, the very definition of $u$ directly shows that $u$ is a centered Gaussian random field with Cameron--Martin space 
$(\dot{H}^\alpha,(\cdot,\cdot)_{\alpha})$ and Gaussian space 
$$
	{H_u = \Phi^{-1}\circ \Pi_K(\widetilde{H}^\alpha(\Gamma)) = \Phi^{-1}(\dot{H}^\alpha)}.
$$
Thus, $u$ is a Whittle--Mat\'ern field on $\Gamma$. Indeed, this follows from the uniqueness of the Cameron--Martin spaces. 
Furthermore, by the same lines as in \cite[Section~6]{BSW_Markov}, we have the following representation for Whittle--Mat\'ern fields with $\alpha\in\mathbb{N}$:
\begin{equation}\label{eq:repr_gWM_alpha}
	u(s) = \pE(\widetilde{u}(s) | \sigma(H_u)),\quad s\in\Gamma.
\end{equation}

We  now show that the conditional expectation above only affects the values of $\widetilde{u}$ at the vertices of $\Gamma$.
Let $\varrho_{0,e}(\cdot,\cdot)$ be the covariance function of $v_{e,0}$, satisfying $\varrho_{0,e}(\cdot,t)\in H_0^\alpha(e)$, $t\in e$. 
We claim that $\Phi(v_{e,0}(t))|_e \in H^\alpha_0(e)$, where $t\in e$. 
To prove this claim, observe that, because  $B^\alpha \widetilde{u}_e$ and $v_{e,0}$ are independent, 
\begin{align*}
	\Phi(v_{e,0}(t))((t',e)) &= \pE(v_{e,0}(t)\widetilde{u}_e(t')) = \pE\left(v_{e,0}(t)(v_{e,0}(t') + \mv{s}_e^\transp(t') B^\alpha \widetilde{u}_e)\right)\\
	&= \pE(v_{e,0}(t)v_{e,0}(t')) = \varrho_{0,e}(t,t').
\end{align*}
Next, observe that, for $e\neq e'$, $\widetilde{u}_e$ is independent of $\widetilde{u}_{e'}$. 
Thus, for every $t\in e$, ${\Phi^{(k)}(v_{e,0}(t)) = 0}$ in $\Gamma\setminus e$, $k=0,\ldots,\alpha-1$. 
In other words, $\Phi^{(k)}(v_{e,0}(t))=0$ for every boundary point $t$ of the edge $e$. 
In particular, $\Phi^{(k)}(v_{e,0}(t)), t\in e$ is continuous for $k=0,\ldots,\alpha-1$. 
Therefore, for any $t\in e$, $\Phi(v_{e,0}(t)) \in \dot{H}^\alpha_L(\Gamma)$. 
Now, let $H_{0,\widetilde{u}}$ be the $L_2(\Omega)$ closure of $\textrm{span}\{v_{e,0}(s): s\in\Gamma\}$.
Then, $\Phi(H_{0,\widetilde{u}}) \subset \dot{H}^\alpha_L(\Gamma)$ gives us that
\begin{equation}\label{eq:bridge_contained_kirk_alpha}
	H_{0,\widetilde{u}}\subset H_u.
\end{equation}
Therefore, we can use \eqref{eq:bridge_repr_indep_bdless}, \eqref{eq:repr_gWM_alpha}, and \eqref{eq:bridge_contained_kirk_alpha} to obtain, for $s\in\Gamma$, that 
\begin{align*}
	u(s) &= \pE(\widetilde{u}(s) | \sigma(H_u)) = \pE\left(\sum_{e\in\mathcal{E}} v_{e,0}(s) + \mv{s}_{e}^\transp(t) B^\alpha \widetilde{u}_{e}  \Big| \sigma(H_u)\right) \\
		&= \sum_{e\in\mathcal{E}} v_{e,0}(s) + \mv{s}_{e}^\transp(s) \pE(B^\alpha \widetilde{u}_{e}|\sigma(H_u)).
\end{align*}

Following the notation introduced in \citet[Section 6]{BSW_Markov}, we let $H_\mathcal{G}$ denote the orthogonal complement of $H_{0,\widetilde{u}}$ in $H_u$ 
(i.e., $H_u = H_{0,\widetilde{u}} \oplus H_\mathcal{G}$).
Similarly as in \cite[Section 6]{BSW_Markov}, we have $H_\mathcal{G} = \textrm{span}\{F_j \widehat{B}^\alpha \widetilde{u}: j\in J\}$,
where $\widehat{B}^\alpha \widetilde{u} = \{B^\alpha \widetilde{u}_e: e\in\mathcal{E}\}$, $F_j$ denotes a row vector and $J$ is a finite set containing the indexes. 
This implies, in particular, that the linear space $H_\mathcal{G}$ is generated by a set of linear combinations of $\widehat{B}^\alpha\widetilde{u}$ (i.e., of vertex elements).
On the other hand, for each $e\in\mathcal{E}$, $B^\alpha\widetilde{u}_e$ is independent of $\sigma(H_{0,\widetilde{u}})$. Thus,
\begin{align*}
	u(s) &= \sum_{e\in\mathcal{E}} v_{e,0}(s) + \mv{s}_{e}^\transp(s) \pE(B^\alpha \widetilde{u}_{e}|\sigma(H_\mathcal{G})) \\
	&= \sum_{e\in\mathcal{E}} v_{e,0}(s) + \mv{s}_{e}^\transp(s) \pE(B^\alpha \widetilde{u}_{e}|\sigma(F_j \widehat{B}^\alpha\widetilde{u}:j\in J)), \, \, s\in\Gamma.
\end{align*}

The above expression allows using \citet[Theorem 9.1 and Remark 9.2]{janson_gaussian}. Therefore, matrices $\widehat{T}_{\alpha,e}, e\in\mathcal{E}$ exist such that, 
for every $e\in\mathcal{E}$,
$$
	\pE(B^\alpha\widetilde{u}_e|\sigma(F_j \widehat{B}^\alpha\widetilde{u}:j\in J)) = \widehat{T}_{\alpha,e} \widehat{F} \widehat{B}^\alpha \widetilde{u} := T_{\alpha,e} \widehat{B}^\alpha\widetilde{u},
$$
where $\widehat{F} = [F_1,\ldots,F_{|J|}]$, $\widehat{T}_{\alpha,e}, e\in\mathcal{E}$ and $T_{\alpha,e} = \widehat{T}_{\alpha,e} \widehat{F}$. 
\end{proof}

The following result is a conditional representation for $u$ similar to the one obtained in \cite[Theorem 10]{BSW_Markov}. 

\begin{Proposition}\label{prp:Wm_prp_kirch_cond_alpha}
Let $\widetilde{u}$ be the field obtained by joining independent CM boundaryless Whittle--Mat\'ern processes (see Definition \ref{def:bdlessWM}) on each edge and let $K_{\alpha,e}, e\in\mathcal{E},$ be the matrices given in \eqref{eq:decomp_indepWM_kirk_alpha} and $K_\alpha$ be their vertical concatenation. The Whittle--Mat\'ern field given by the solution of \eqref{eq:Matern_spde} with $\alpha\in\mathbb{N}$ can be obtained by conditioning the field $\widetilde{u}$ on $K_\alpha\widehat{B}^\alpha\widetilde{u} = 0$.
\end{Proposition}

\begin{proof}
The representation obtained in \eqref{eq:repr_gWM_alpha} shows that , for $s\in\Gamma$,
\begin{align*}
	\widetilde{u}(s) &= \pE(\widetilde{u}(s)| \sigma(H_u)) + \left(\widetilde{u}(s) -\pE(\widetilde{u}(s)| \sigma(H_u))\right)
	= u(s) + (\widetilde{u}(s) - u(s)).
\end{align*}
This decomposition proves that $u(s)$ and $\widetilde{u}(s) - u(s)$ are orthogonal for every $s\in\Gamma$, and their Gaussianity implies that they are independent. 
To simplify the notation in the bridge representation \eqref{eq:bridge_repr_indep_bdless}, we let $P_{\alpha,e}$ be the matrix satisfying $P_{\alpha,e} \widehat{B}^\alpha\widetilde{u} = B^\alpha\widetilde{u}_e$. Therefore, \eqref{eq:bridge_repr_indep_bdless} can be written as
$$
	\widetilde{u}(s) = \sum_{e\in\mathcal{E}} v_{e,0}(s) + \mv{s}_{e}^\transp(s) P_{\alpha,e} \widehat{B}^\alpha \widetilde{u}, \quad s\in\Gamma.
$$
Combining this expression with Proposition \ref{prp:CondRepr_WM_alpha} shows that 
$
\widetilde{u}(s) - u(s) = \sum_{e\in\mathcal{E}}   \mv{s}_{e}^\transp(s) K_{\alpha,e} \widehat{B}^\alpha \widetilde{u},
$
for every $s\in\Gamma$, where $K_{\alpha,e} = P_{\alpha,e} - T_{\alpha,e}$. That is, 
\begin{equation}\label{eq:decomp_indepWM_kirk_alpha}
	\widetilde{u}(s) = u(s) + \sum_{e\in\mathcal{E}}   \mv{s}_{e}^\transp(s) K_{\alpha,e} \widehat{B}^\alpha \widetilde{u},\quad s\in\Gamma.
\end{equation}
Here, $u(s)$ and $\{K_{\alpha,e} \widehat{B}^\alpha \widetilde{u}:e\in\mathcal{E}\}$ are independent for every $s\in\Gamma$. 
This independence, the fact that $u(s)$ is a centered Gaussian field, and \eqref{eq:decomp_indepWM_kirk_alpha}, shows that, for every $s\in\Gamma$,
\begin{align*}
	\pE(\widetilde{u}(s) | \sigma(K_{\alpha,e} \widehat{B}^\alpha \widetilde{u}:e\in\mathcal{E})) &= \pE\left(u(s) + K_{\alpha,e} \widehat{B}^\alpha \widetilde{u} \Big| \sigma(K_{\alpha,e} \widehat{B}^\alpha \widetilde{u}:e\in\mathcal{E})\right) \\
	&= \pE\left(u(s) \Big| \sigma(K_{\alpha,e} \widehat{B}^\alpha \widetilde{u}:e\in\mathcal{E})\right) + \sum_{e\in\mathcal{E}}   \mv{s}_{e}^\transp(s) K_{\alpha,e} \widehat{B}^\alpha \widetilde{u}\\
	&= \pE(u(s)) + \sum_{e\in\mathcal{E}}   \mv{s}_{e}^\transp(s) K_{\alpha,e} \widehat{B}^\alpha \widetilde{u}= \sum_{e\in\mathcal{E}}   \mv{s}_{e}^\transp(s) K_{\alpha,e} \widehat{B}^\alpha \widetilde{u}.
\end{align*}
We now follow the steps in \citet[Section 6]{BSW_Markov}. 
That is, we first define $H_{K\widehat{B}\widetilde{u}}$ as the Gaussian space spanned by $\{K_{\alpha,e} \widehat{B}^\alpha \widetilde{u}:e\in\mathcal{E}\}$. 
Second, we let $\Pi_B$ denote the $L_2(\Omega)$-orthogonal projection from $H_{\widetilde{u}}$ onto $H_{K\widehat{B}\widetilde{u}}$, 
and similarly, let $\Pi_B^\perp$ be the $L_2(\Omega)$-orthogonal projection onto the orthogonal complement of $H_{K\widehat{B}\widetilde{u}}$ on $H_{\widetilde{u}}$. Then, $\Pi_B^\perp (\widetilde{u}(s)) = u(s)$ and
$\Pi_B(\widetilde{u}(s)) = \sum_{e\in\mathcal{E}}   \mv{s}_{e}^\transp(s) K_{\alpha,e} \widehat{B}^\alpha \widetilde{u}$. 
The fact that $\Pi_B^\perp (\widetilde{u}(s)) = u(s)$ allows using the result in \citet[Remark 9.10]{janson_gaussian} to obtain that $u(s)$ is the field obtained by conditioning $\widetilde{u}(s)$ on $K_{\alpha,e}\widehat{B}^\alpha\widetilde{u} = 0, e\in\mathcal{E}$. Finally, we can now unify the conditional set by letting $K_\alpha$ denote the vertical concatenation of $K_{\alpha,e}, e\in\mathcal{E}$, so that the system $\{K_{\alpha,e}\widehat{B}^\alpha\widetilde{u} = 0, e\in\mathcal{E}\}$ is equivalent to the linear equation $K_\alpha\widehat{B}^\alpha\widetilde{u} = 0$. 
\end{proof}

We are now in a position to prove Theorem \ref{thm:representation}:

\begin{proof}[Proof of Theorem \ref{thm:representation}]
Proposition \ref{prp:Wm_prp_kirch_cond_alpha} shows that the solution to \eqref{eq:Matern_spde} can be obtained by conditioning $\widetilde{u}(s)$ on $K_\alpha\widehat{B}^\alpha\widetilde{u} = 0$ for any $s\in\Gamma$. 
Because $\widehat{B}^\alpha\widetilde{u}$ consists solely of values of the field $\widetilde{u}$ at the vertices of the metric graph, 
we can conclude that $K\widehat{B}\widetilde{u} = 0$ enforces continuity of the even order weak derivatives (in the $L_2(\Omega)$ sense) of conditioned field, 
 $u$, at the vertices, 
 and enforces the Kirchhoff condition $\sum_{e\in \mathcal{E}_v} \partial_e^{(m)} u(v) = 0$, where $m$ is a positive odd integer number less than or equal to $\alpha-1$. 
 To make this more precise, observe that \eqref{eq:decomp_indepWM_kirk_alpha} allows us to conclude that $K_\alpha\widehat{B}^\alpha\widetilde{u} = 0$ if and only if
\begin{equation}\label{eq:Matrix_cond_kirk_alpha}
\begin{cases}
	\widetilde{u}^{(2k)}_e(v) = \widetilde{u}^{(2k)}_{e'}(v), &k = 0, \ldots, \floor{(\alpha-1)/2},\quad v\in e\cap e',\quad e,e'\in\mathcal{E},\\
	\sum_{e\in\mathcal{E}_v} \partial_e^{(2k-1)} u(v) = 0, &k = 1,\ldots,\floor{\alpha/2}, \quad v\in\mathcal{V}, \hbox{if $\alpha>1$}.
\end{cases}
\end{equation}
As in \citet[Section 6]{BSW_Markov}, we have that, if $\widetilde{K}_\alpha$ is another matrix that enforces the conditions in  \eqref{eq:Matrix_cond_kirk_alpha} (i.e., if  $\widetilde{K}_\alpha$ is another matrix such that  \eqref{eq:Matrix_cond_kirk_alpha} holds if, and only if ${\widetilde{K}_\alpha\widehat{B}^\alpha \widetilde{u} = 0}$), then $\sigma(K_\alpha\widehat{B}^\alpha\widetilde{u}) = \sigma(\widetilde{K}_\alpha\widehat{B}^\alpha\widetilde{u})$.
The result follows from a direct combination of the above with Proposition \ref{prp:bdless_thm_cov_123}.
\end{proof}

\section{The CM boundaryless Whittle--Mat\'ern process}\label{app:bdlessaux}
We have the following 
representation for CM boundaryless Whittle--Mat\'ern processes.

\begin{Proposition}\label{prp:EdgeReprBdlessProcInterval}
Let, $u$ be a CM boundaryless Whittle--Mat\'ern process on $[0,\ell]$, where $0<\ell<\infty$, with characteristics $\alpha\in\mathbb{N}$, $\kappa >0$, and $\tau>0$. Then, we have the representation:
$$
u(t) = v_{\alpha,0}(x) + \sum_{j=1}^{2\alpha} s_j(x)  \left(B^{\alpha}u\right)_j = v_{\alpha,0}(x) + \mv{s}^\transp(t) B^\alpha u,\quad x\in [0,\ell_e],
$$
where $\mv{s}(\cdot) = (s_1(\cdot),\ldots,s_{2\alpha}(\cdot))$ is given by \eqref{eq:edge_repr_Solution} and
$v_{\alpha,0}(\cdot)$ is the Whittle--Mat\'ern bridge process which is independent of $B^\alpha u$.
\end{Proposition}

\begin{proof}
The proof is analogous to the proof of \cite[Theorem 9]{BSW_Markov}, where we replace in this proof 
(and in the proofs of the auxiliary results needed for this proof, such as \cite[Theorem 8 and Proposition 9]{BSW_Markov}), 
the space $\dot{H}^\alpha_L(\Gamma)$ with $H^\alpha([0,\ell])$ and the inner product $(\cdot,\cdot)_{\dot{H}^\alpha_L(\Gamma)}$ with $\<\cdot,\cdot\>_{\alpha,\kappa,\tau}$. 
Finally, the identification of $v_{\alpha,0}(\cdot)$ as a Whittle--Mat\'ern bridge process follows from Lemmata~\ref{lem:CM_edge_repr_bridge} and \ref{lem:charHdot_edge}, and the explicit expression for $\mv{s}(\cdot)$ follows from Lemma \ref{lem:sol_ODE_edge}.
\end{proof}

For convenience, we recall the definition of strictly positive-definite matrix-valued functions.

\begin{Definition}\label{def:str_pos_def_mat}
	Let $\mv{f}:X\times X\to M_k(\mathbb{R})$, where $X$ is a non-empty set and $M_k(\mathbb{R})$ is the set of $k\times k$ real-valued matrices. We say that $\mv{f}$ is strictly positive-definite if for every $N\in\mathbb{N}$, $x_1,\ldots,x_N\in X$ and $\mv{c}_1,\ldots,\mv{c}_N \in \mathbb{R}^k$, with at least one $\mv{c}_i\neq 0$, for some $i=1,\ldots,N$, we have
	$\sum_{i,j=1}^N \mv{c}_i^\top \mv{f}(x_i,x_j) \mv{c}_j > 0.$
\end{Definition}

We now prove a theorem that allows us to obtain the covariance function of CM boundaryless Whittle--Mat\'ern processes. 
We state this theorem in a very general format, then apply it to the CM boundaryless Whittle--Mat\'ern processes.

\begin{Theorem}\label{thm:CondDens}
Let $(\Omega, \mathcal{F},\mathbb{P})$ be a complete probability space and let $\mv{u} : \mathbb{R}\times \Omega \rightarrow \mathbb{R}^d$ 
be a stationary Gaussian Markov process of order 1 \cite[Sections 5 and 10]{Pitt1971} with a strictly positive-definite (matrix-valued) covariance function 
$\mv{r}\left(t_1, t_2\right)  = \Cov \left[ \mv{u}\left(t_1\right), \mv{u}\left(t_2\right) \right]$, ${t_1,t_2\in\mathbb{R}}$. 
Then, for each $T>0$, the function $\tilde{\mv{r}}_T(\cdot,\cdot)$ defined in \eqref{eq:covmod} is a strictly positive-definite covariance function on the domain $[0, T]$. Further, the family $\{\tilde{\mv{r}}_T \}_{T\in \mathbb{R}_+}$ is the unique family of covariance functions on $\{[0,T]\}_{T\in \mathbb{R}_+}$ satisfying the following conditions: 
\begin{enumerate}
	\item If $\tilde{ \mv{u} }$ is a centered Gaussian process on $[0,T]$ with covariance $\tilde{\mv{r}}_T\left(\cdot,\cdot\right)$, then, for any ${m\in\mathbb{N}}$, 
	any $\mv{t}\in\mathbb{R}^m,  \mv{t}=(t_1,\ldots,t_m)$, and any $\mv{u}_0, \mv{u}_T\in\mathbb{R}^d$,
	$$
		\tilde{ \mv{u} }(\mv{t})| \{\tilde{ \mv{u}}(0)=\mv{u}_0, \tilde{ \mv{u}}(T) = \mv{u}_T \} \stackrel{d}{=} \mv{u}(\mv{t})| \{\mv{u}(0)=\mv{u}_0,  \mv{u}(T) = \mv{u}_T\},
	$$
	where $\tilde{ \mv{u} }(\mv{t}) = (\tilde{\mv{u}}(t_1), \ldots,\tilde{\mv{u}}(t_m)), \mv{u}(\mv{t})= (\mv{u}(t_1), \ldots,\mv{u}(t_m)).$\label{thm:ConDens1}
	\item $\tilde{\mv{r}}_T\left(0,0\right) = \tilde{\mv{r}}_T\left(T,T\right)$.\label{thm:ConDens2}
	\item Let $T_1+T_2=T$ and $\tilde{\mv{u}}_{T_1}, \tilde{\mv{u}}_{T_2}$, and $\tilde{\mv{u}}_{T}$ be three independent centered Gaussian processes  with covariance functions $\tilde{\mv{r}}_{T_1}\left(\cdot,\cdot\right),\tilde{\mv{r}}_{T_2}\left(\cdot,\cdot\right)$, and $\tilde{\mv{r}}_{T}\left(\cdot,\cdot\right)$, respectively. 
	Define
	\begin{equation}\label{eq:cond3thmconddens}
		\tilde{ \mv{u}}^*(t )  \stackrel{d}{=}  \left[\Big(\mathbb{I}\left(t \leq T_1\right)\tilde{ \mv{u}}_{T_1}(t) +  \mathbb{I}\left(t \geq T_1\right)\tilde{ \mv{u}}_{T_2}(t-T_1)\Big)  | \tilde{ \mv{u}}_{T_1}(T_1) = \tilde{ \mv{u}}_{T_2}(0)\right], 
	\end{equation}
	where $t\in [0,T]$, as the process with the same finite-dimensional distributions as the conditional distribution on the right-hand side of \eqref{eq:cond3thmconddens}  (i.e, the process obtained by joining $\tilde{\mv{u}}_{T_1}(t)$ and $\tilde{\mv{u}}_{T_2}(t)$ on $[0,T]$).
	Then $\tilde{ \mv{u}}^*    \stackrel{d}{=} \tilde{ \mv{u}}_{T}.$\label{thm:ConDens3}
\end{enumerate}
\end{Theorem}

To prove the theorem, we require the following lemma and corollary. 
The result in the lemma is known in the literature as adjusting the $c$-marginal  \cite[p.~134]{lauritzen1996graphical}. 
Due to its importance, we state it in a slightly more general form than in \cite{lauritzen1996graphical} and prove it.

\begin{Lemma}\label{lem:conditional}
Assume that 
\begin{align*}
	\mv{X}= \begin{bmatrix}
	\mv{X}_A \\
	\mv{X}_B
	\end{bmatrix}
	\sim \pN\left(\mv{0},\mv{\Sigma}\right), \quad 
	\text{with} \quad 
	\mv{\Sigma} = \begin{bmatrix}
	\mv{\Sigma}_{AA} & \mv{\Sigma}_{AB}\\
	\mv{\Sigma}_{BA} & \mv{\Sigma}_{BB}
	\end{bmatrix},
\end{align*}
where $\mv{X}_A\in\mathbb{R}^{n_A}$ and $\mv{X}_B\in\mathbb{R}^{n_B}$.
Let $\mv{Q} = \mv{\Sigma}^{-1}$ (with the corresponding block structure) and fix a symmetric and nonnegative definite $n_B\times n_B$ matrix $\mv{H}$. 
Then, if $\mv{X}_{B}^* \sim \pN\left(0,\mv{H} \right)$ and $\mv{X}^*_A = \mv{X}_{A|B}  + \mv{\Sigma}_{AB}\mv{\Sigma}^{-1}_{BB} \mv{X}^*_B$,
where 
${\mv{X}_{A|B} \sim \pN \left( \mv{0}, \mv{\Sigma}_{AA} - \mv{\Sigma}_{AB}\mv{\Sigma}^{-1}_{BB} \mv{\Sigma}_{BA}  \right)}$, we have 
$\mv{X}^*= 	[(\mv{X}^*_A)^\top, (\mv{X}^*_B)^\top]^\top \sim \pN\left(\mv{0},\mv{\Sigma}^*\right)$,
where 
\begin{align}\label{eq:inverseprecmatstr}
	\mv{\Sigma}^*= 	
	\begin{bmatrix}
	\mv{\Sigma}_{AA} - \mv{\Sigma}_{AB}\mv{\Sigma}^{-1}_{BB} \mv{\Sigma}_{BA} + \mv{\Sigma}_{AB}\mv{\Sigma}^{-1}_{BB} \mv{H}\mv{\Sigma}^{-1}_{BB}\mv{\Sigma}_{BA}  &\quad  \mv{\Sigma}_{AB}\mv{\Sigma}^{-1}_{BB} \mv{H}\\
	\mv{H}\mv{\Sigma}^{-1}_{BB} \mv{\Sigma}_{BA} &\quad \mv{H}
	\end{bmatrix}
\end{align}
and
\begin{align}\label{eq:PrecLemma}
	\left( \mv{\Sigma}^* \right)^{-1} = \mv{Q}^* = 
	\begin{bmatrix}
	\mv{Q}_{AA} && \mv{Q}_{AB} \\
	\mv{Q}_{BA} && \mv{Q}_{BB} + \mv{H}^{-1} - \mv{\Sigma}^{-1}_{BB} 
	\end{bmatrix} = \begin{bmatrix}
	\mv{Q}_{AA} && \mv{Q}_{AB} \\
	\mv{Q}_{BA} && \mv{H}^{-1}  + \mv{Q}_{BA}\mv{Q}^{-1}_{AA} \mv{Q}_{AB}
	\end{bmatrix} .
\end{align}
	
If $\mv{H}$ is singular where the projection onto the nonnull space is $\mv{P} = \mv{H}^\dagger \mv{H}$, then
\begin{align*}
	\mv{Q}^*  = 
	\begin{bmatrix}
	\mv{Q}_{AA} && \mv{Q}_{AB}\mv{P} \\
	\mv{P}\mv{Q}_{BA} && \mv{H}^{\dagger}  + \mv{P}\left( \mv{Q}_{BB} -  \mv{\Sigma}_{BB}^{-1} \right)\mv{P}
	\end{bmatrix} =
	\begin{bmatrix}
	\mv{Q}_{AA} && \mv{Q}_{AB}\mv{P} \\
	\mv{P}\mv{Q}_{BA} && \mv{H}^{\dagger}  + \mv{P}\left(  \mv{Q}_{BA}\mv{Q}^{-1}_{AA} \mv{Q}_{AB} \right)\mv{P}
	\end{bmatrix},
\end{align*}
where $\mv{H}^\dagger$ denotes the pseudo-inverse of $\mv{H}$.
\end{Lemma}
\begin{proof}
The expression for $\mv{\Sigma}^*$ follows directly by the definitions of $\mv{X}_A^*$ and $\mv{X}_B^*$; thus, 
we must only establish that $\mv{Q}^*$ has the desired expression. By the Schur complement, 
\begin{align*}
	\mv{Q}_{AA}^* &=  \left(\mv{\Sigma}^*_{AA} - \mv{\Sigma}^*_{AB} \left(\mv{\Sigma}^{*}_{BB}\right)^{-1} \mv{\Sigma}^*_{BA}\right)^{-1}= 
	\left(\mv{\Sigma}_{AA} - \mv{\Sigma}_{AB} \left(\mv{\Sigma}_{BB}\right)^{-1} \mv{\Sigma}_{BA}\right)^{-1} = \mv{Q}_{AA}.
\end{align*}
Then, using that $\mv{\Sigma}_{AB}  \mv{\Sigma}^{-1}_{BB} =  -\mv{Q}_{AA} ^{-1} \mv{Q}_{AB}$ (see \cite{rue2005gaussian} pp.21 and 23), we have
\begin{align*}
	\mv{Q}_{BA}^*  &=  -\mv{Q}_{AA} \mv{\Sigma}_{AB}  \mv{\Sigma}^{-1}_{BB} \mv{H}\mv{H}^{-1} =
	-\mv{Q}_{AA} \mv{\Sigma}_{AB}  \mv{\Sigma}^{-1}_{BB} = \mv{Q}_{AB} .
\end{align*}
Again using the Schur complement and that $\mv{\Sigma}_{AB}  \mv{\Sigma}^{-1}_{BB} =  -\mv{Q}_{AA} ^{-1} \mv{Q}_{AB}$, we obtain
\begin{align*}
	\mv{Q}^*_{BB}  &=  \mv{H}^{-1}  +  \mv{H}^{-1}  \mv{\Sigma}^*_{BA} \mv{Q}^*_{AA}\mv{\Sigma}^*_{AB}  \mv{H}^{-1} 
	=  \mv{H}^{-1}  + \mv{\Sigma}_{BB}^{-1} \mv{\Sigma}_{BA} \mv{Q}_{AA}\mv{\Sigma}_{AB} \mv{\Sigma}_{BB}^{-1}  \\
	&=  \mv{H}^{-1}  + \mv{Q}_{BA}\mv{Q}^{-1}_{AA} \mv{Q}_{AB} 
	=  \mv{H}^{-1}    + \mv{Q}_{BB} -  \mv{\Sigma}_{BB}^{-1}.
\end{align*}
Finally, if $\mv{H}$ is singular, then 
\begin{align*}
	\mv{Q}^*_{BB}  &=   \mv{H}^{\dagger}  +   \mv{H}^{\dagger}  \mv{\Sigma}^*_{BA} \mv{Q}^*_{AA}\mv{\Sigma}^*_{AB}  \mv{H}^{\dagger}  
	=   \mv{H}^{\dagger}   +\mv{P} \mv{\Sigma}_{BB}^{-1} \mv{\Sigma}_{BA} \mv{Q}_{AA}\mv{\Sigma}_{AB} \mv{\Sigma}_{BB}^{-1}\mv{P}  \\
	&=  \mv{H}^{\dagger}    +\mv{P}\left( \mv{Q}_{BB} -  \mv{\Sigma}_{BB}^{-1} \right)\mv{P}.
\end{align*}
\end{proof}

The lemma shows that one can change the conditional distribution of a Gaussian random variable, $\mv{X}_{B}$, without affecting the conditional distribution of $\mv{X}_{A}|\mv{X}_B$. 
The following particular case is needed for the proof of the theorem. 
The result follows directly from Lemma~\ref{lem:conditional} and the Woodbury matrix identity.

\begin{Corollary}\label{cor:conditional}
Assume the setting of Lemma~\ref{lem:conditional} and that $\mv{H}^{-1} = \mv{\Sigma}^{-1}_{BB} + \mv{C}$, where $\mv{C}$ is a symmetric and non-singular matrix. Then, $\mv{Q}^*  = \mv{Q} + \diag(\mv{0}, \mv{C})$ and 
\begin{align*}
	\mv{X}^*= 	\begin{bmatrix}
	\mv{X}^*_A \\
	\mv{X}^*_B
	\end{bmatrix}
	\sim \pN\left(\mv{0},
	\begin{bmatrix}
	\mv{\Sigma}_{AA} - \mv{\Sigma}_{AB}\left(\mv{\Sigma}_{BB} + \mv{C}^{-1} \right)^{-1}\mv{\Sigma}_{BA}  &\qquad  \mv{\Sigma}_{AB}\mv{\Sigma}^{-1}_{BB} \mv{H}\\
	\mv{H}\mv{\Sigma}^{-1}_{BB}\mv{\Sigma}_{BA}  &\qquad \mv{H}
	\end{bmatrix}
	\right).
\end{align*}
\end{Corollary}

The proof of Lemma \ref{lem:conditional} also provides the following corollary about precision matrices with certain structures:

\begin{Corollary}\label{cor:precmatstructure}
Using the notation from Lemma \ref{lem:conditional}, let $\mv{Q}^*$ be a symmetric and invertible matrix of the form 
\begin{equation}\label{eq:precmatstructS}
	\mv{Q}^*  = \begin{bmatrix}
	\mv{Q}_{AA} &\quad \mv{Q}_{AB} \\
	\mv{Q}_{BA} &\quad \mv{S}  + \mv{Q}_{BA}\mv{Q}^{-1}_{AA} \mv{Q}_{AB}
	\end{bmatrix} ,
\end{equation}
where $\mv{S}$ is symmetric and strictly positive-definite. Then, the inverse $\mv{\Sigma}^* = (\mv{Q}^*)^{-1}$ is given by
\begin{align}\label{eq:inverseprecmatstrS}
	\mv{\Sigma}^*= 	
	\begin{bmatrix}
	\mv{\Sigma}_{AA} - \mv{\Sigma}_{AB}\mv{\Sigma}^{-1}_{BB} \mv{\Sigma}_{BA}   +  \mv{\Sigma}_{AB}\mv{\Sigma}^{-1}_{BB} \mv{S}^{-1}\mv{\Sigma}^{-1}_{BB}\mv{\Sigma}_{BA}  &\qquad  \mv{\Sigma}_{AB}\mv{\Sigma}^{-1}_{BB} \mv{S}^{-1}\\
	\mv{S}^{-1}\mv{\Sigma}^{-1}_{BB}\mv{\Sigma}_{BA}  & \qquad\mv{S}^{-1}
	\end{bmatrix}
\end{align}
\end{Corollary}

\begin{proof}
The expression for $\mv{\Sigma}^* = (\mv{Q}^*)^{-1}$ is obtained by comparing \eqref{eq:precmatstructS} and \eqref{eq:inverseprecmatstrS} with \eqref{eq:PrecLemma} and \eqref{eq:inverseprecmatstr}, respectively, where we obtain that $\mv{H}^{-1}=\mv{S}$. 
\end{proof}

\begin{proof}[Proof of Theorem~\ref{thm:CondDens}]	
We use the following matrix notation throughout the proof: 
Suppose that $\mv{\Sigma}$ is the covariance matrix of $\mv{U}= (\mv{u}(t_1),\ldots, \mv{u}(t_n))^\top$ for some ${t_1,\ldots, t_n \in \mathbb{R}}$, 
then we let $\mv{\Sigma}^{t_it_j}$ denote the submatrix that is the covariance matrix of $(\mv{u}(t_i),\mv{u}(t_j))^\top$. 
We write $\mv{Q} = \mv{\Sigma}^{-1}$ for the precision matrix of $\mv{U}$ and $\mv{Q}^{t_it_j} := (\mv{\Sigma}^{t_it_j})^{-1}$, 
whereas, for a matrix $\mv{M}$, $\mv{M}_{t_i t_j}$ denotes the submatrix obtained from $\mv{M}$ with respect to the indices $t_i$ and $t_j$.

We derive the result by showing that any covariance function $\tilde{\mv{r}}$ satisfying (i), (ii) and (iii) must be of the form \eqref{eq:covmod}. 
Finally, we show that $\tilde{\mv{r}}_T$ is, indeed, a covariance function.

Fix $t, s \in (0,T)$ and let $\mv{Q}$ be the precision matrix of $\mv{U}=(\mv{u}(t), \mv{u}(s), \mv{u}(0), \mv{u}(T))^\top$ and $\widetilde{\mv{Q}}$ denote the precision matrix of $\tilde{\mv{U}}=(\tilde{\mv{u}}(t), \tilde{\mv{u}}(s), \tilde{\mv{u}}(0),  \tilde{\mv{u}}(T))^\top$. 
In addition, let $A = \{t,s\}$ and $B = \{0,T\}$. 
Then, (i) and the Markov property of $\mv{u}$ imply that 
$\widetilde{\mv{Q}}_{AA} = 	\mv{Q}_{AA} $ and ${\widetilde{\mv{Q}}_{AB} = \mv{Q}_{AB}}$. 
Hence, $\widetilde{\mv{Q}}$ is of the form \eqref{eq:precmatstructS}, so by Corollary \ref{cor:precmatstructure}, 
its inverse is given by \eqref{eq:inverseprecmatstrS}, directly implying that $\tilde{\mv{r}}$ is given by
\begin{align*}
	\tilde{\mv{r}}\left(s,t\right) =&\,  \mv{r}(s,t) - \begin{bmatrix}\mv{r}(s, 0) & \mv{r}(s,T)\end{bmatrix}  \left(\mv{\Sigma}^{0T}\right)^{-1} \begin{bmatrix}\mv{r}(0,t)\\
	\mv{r}(T,t)
	\end{bmatrix}  \\
	& \,\,\,\,+\begin{bmatrix}\mv{r}(s, 0) & \mv{r}(s,T)\end{bmatrix}  \left(\mv{\Sigma}^{0T}\right)^{-1}\mv{H}^{0T} \left(\mv{\Sigma}^{0T}\right)^{-1} \begin{bmatrix}\mv{r}(0,t)\\
	\mv{r}(T,t)
	\end{bmatrix},
\end{align*}
for some strictly positive-definite matrix $\mv{H}^{0T}$.
	
We now show that (ii) and (iii) provide an explicit form for $\mv{H}^{0T}$. 
Fix $T_1\in (0,T)$, let $\widetilde{\mv{u}}^*$ be obtained from \eqref{eq:cond3thmconddens}, 
where $T_2 = T-T_1$, and let $\widetilde{\mv{u}}_3$ be obtained from $\tilde{\mv{r}}_T$. 
We obtain an explicit form for $\mv{H}^{0T}$ by obtaining conditions for equality of the densities of 
$\widetilde{\mv{U}}^*= [\widetilde{\mv{u}}^*\left(0 \right),\widetilde{\mv{u}}^*\left(T_1 \right), \widetilde{\mv{u}}^*\left(T \right)]$ and  
$\widetilde{\mv{U}}_3= [\widetilde{\mv{u}}_3\left(0 \right),\widetilde{\mv{u}}_3\left(T_1 \right), \widetilde{\mv{u}}_3\left(T \right)]$, 
which we denote by $f_{\widetilde{\mv{U}}^*}$ and $f_{\widetilde{\mv{U}}_3}$, respectively.
Let $ \mv{C}^{0T} =  \left(\mv{H}^{0T} \right)^{-1} -\mv{Q} ^{0T}$. Then, by Lemma  \ref{lem:conditional}, we have that
\begin{align*}
	f_{\widetilde{\mv{U}}^*}( \mv{x}) &\propto f_{\widetilde{\mv{u}}_1(0),\widetilde{\mv{u}}_1(T_1)}( \mv{x}_0,\mv{x}_{T_1})  f_{\widetilde{\mv{u}}_1(0),\widetilde{\mv{u}}_1(T_2)}( \mv{x}_{T_1},\mv{x}_{T})\\
	&
	\propto \exp \left(  -0.5
	\begin{bmatrix}
	\mv{x}_0 \\ 
	\mv{x}_{T_1}
	\end{bmatrix}^{\top}  \left( \mv{H}^{0T_1} \right)^{-1}
	\begin{bmatrix}
	\mv{x}_0 \\ 
	\mv{x}_{T_1}
	\end{bmatrix}  
	-0.5 \begin{bmatrix}
	\mv{x}_{T_1} \\ 
	\mv{x}_{T}
	\end{bmatrix}^{\top} \left( \mv{H}^{0T_2} \right)^{-1}
	\begin{bmatrix}
	\mv{x}_{T_1} \\ 
	\mv{x}_{T}
	\end{bmatrix}
	\right) \\
	&= \exp \left( -0.5 \mv{x}^{\top}
	\mv{Q}^* \mv{x}
	\right),
\end{align*}
where
$$
	\mv{Q}^* =
	\begin{bmatrix}
	\mv{Q}^{0T_1}_{00} + \mv{C}^{0T_1}_{00} &\qquad 	\mv{Q}^{0T_1}_{0T_1} + \mv{C}^{0T_1}_{0T_1} &\qquad \mv{0} \\
	\mv{Q}^{0T_1}_{0T_1} + \mv{C}^{0T_1}_{0T_1} & \qquad	\mv{Q}^{0T_1}_{T_1T_1}+ \mv{Q}^{0T_2}_{00} +  \mv{C}^{0T_1}_{00}+ \mv{C}^{0T_2}_{T_2T_2} &\qquad  	\mv{Q}^{0T_2}_{0T_2} + \mv{C}^{0T_2}_{0T_2} \\
	\mv{0} &\qquad \mv{Q}^{0T_2}_{0T_2} + \mv{C}^{0T_2}_{0T_2} &\qquad  \mv{Q}^{0T_2}_{T_2T_2} +\mv{C}^{0T_2}_{T_2T_2}
	\end{bmatrix} .
$$
Again, by Lemma \ref{lem:conditional}, the density of $\widetilde{\mv{U}}_3$ is
$f_{\widetilde{\mv{U}}_3}( \mv{x})  
	= \exp \left( -0.5 \mv{x} ^{\top}
	\widetilde{\mv{Q}}   \mv{x}
	\right),$
where
\begin{equation}\label{eq:qtildeconddens}
	\widetilde{\mv{Q}} =   \mv{Q}^{0T_1T} +
	\begin{bmatrix}
	\mv{C}^{0T}_{00} & \mv{0} &  \mv{C}^{0T}_{0T} \\
	\mv{0} &  \mv{0}  &   \mv{0} \\
	\mv{C}^{0T}_{0T}&   \mv{0}  &   \mv{C}^{0T}_{TT}
	\end{bmatrix},
\end{equation}
and $\mv{Q}^{0T_1T}$ is the precision matrix of $[\mv{u}(0), \mv{u}(T_1),\mv{u}(T)]$.
Now the densities are equal if and only if $\mv{Q}^*  = \widetilde{\mv{Q}}_3 $, which establishes three conditions on $\mv{C}^{0T}$: 
First, $\mv{C}^{0T}_{0T}  = \mv{0}$ for all $T$. 
Second, due to the Markov property of $\mv{u}$, $\mv{Q}^{0T_1T}_{00}= \mv{Q}^{0T_1}_{00} $; hence, $\mv{C}^{0T_1}=\mv{C}^{0T}$ for all $T_1$ and $T$. 
Thus, $\mv{C}^{0T}_{00} =: \mv{C}_0$ is a matrix independent of $T$. 
The same reasoning gives that $\mv{C}^{0T}_{TT} =: \mv{C}_{1}$ is independent of $T$. 
Finally, the Markov property and the stationarity of $\mv{u}$ implies that 
$$
	\widetilde{\mv{Q}}_{T_1T_1} = \mv{Q}^{0T_1T}_{T_1T_1} = \mv{r}(0,0)^{-1}+\left( \mv{Q}^{0T_1}_{0T_1} \right)^{\top} \left( \mv{Q}^{0T_1}_{00}\right)^{-1} \mv{Q}^{0T_1}_{0T_1} + 	\left(\mv{Q}^{0T_2}_{0T_2} \right)^{\top} \left( \mv{Q}^{0T_2}_{T_2T_2}\right)^{-1}  \mv{Q}^{0T_2}_{0T_2},
$$
and by construction,
$$
	\mv{Q}^*_{T_1T_1} =2\mv{r}(0,0)^{-1}+\left( \mv{Q}^{0T_1}_{0T_1} \right)^{\top} \left( \mv{Q}^{0T_1}_{00}\right)^{-1}\!\!\! \mv{Q}^{0T_1}_{0T_1} + 	\left(\mv{Q}^{0T_2}_{0T_2} \right)^{\top} \left( \mv{Q}^{0T_2}_{T_2T_2}\right)^{-1}\!\!\!  \mv{Q}^{0T_2}_{0T_2} + \mv{C}_0+\mv{C}_1.
$$
Hence, $\mv{C}_0 + \mv{C}_1 = - \mv{r}(0,0)^{-1}$.
By combining \eqref{eq:qtildeconddens}, the stationarity of $\mv{u}$, and (ii), we obtain $\mv{C}_0=\mv{C}_1$. 
More precisely, we invert the right-hand side of \eqref{eq:qtildeconddens} and use the stationarity of $\mv{u}$ to conclude that if $\mv{C}_0\neq \mv{C}_1$, 
then $\tilde{\mv{r}}_T(0,0)\neq \tilde{\mv{r}}_T(T,T)$. 
Thus,
$$
	\mv{C}^{0T} = - \frac{1}{2}\begin{bmatrix}
	\mv{r}(0,0)^{-1} & \mv{0} \\
	\mv{0} & \mv{r}(0,0)^{-1} 
	\end{bmatrix}.
$$
The desired expression for the covariance of $\tilde{ \mv{u}}(s)$ on $[0,T]$ is obtained by applying Corollary~\ref{cor:conditional}.
Finally, from the Schur complement, the matrix 
$$\begin{bmatrix}
	\mv{r}(0,0)  & 	-\mv{r}(0,T) \\
	-\mv{r}(T,0)  & 	\mv{r}(0,0)
	\end{bmatrix}$$ 
is strictly positive-definite; hence, $\tilde r$ is a covariance function.			
\end{proof}

\begin{proof}[Proof of Proposition~\ref{prop:multivariate_covariante}]
The result follows directly from Theorem~\ref{thm:CondDens}.
\end{proof}

\begin{proof}[Proof of Proposition~\ref{cor:precfuncconddens}]
From Lemma \ref{lem:conditional} and the proof of Theorem~\ref{thm:CondDens}, it follows that $[\widetilde{\mv{u}}(0), \widetilde{\mv{u}}(\ell)] \sim \pN(0, \mv{H}^{0\ell})$, where $\mv(\mv{H}^{0\ell})^{-1} = \mv{Q}^{0\ell} + \mv{C}^{0\ell}$, $\mv{Q}^{0\ell}$ is the precision matrix of $[{\mv{u}}(0), {\mv{u}}(\ell)]$ and 		
$$
	\mv{C}^{0\ell} = - \frac{1}{2}\begin{bmatrix}
	\mv{r}(0,0)^{-1} & \mv{0} \\
	\mv{0} & \mv{r}(0,0)^{-1} 
	\end{bmatrix},
$$
which proves the result.
\end{proof}

To connect Definitions \ref{def:bdlessWM} and \ref{def:cov_based_bdlessWM}, we require the following lemma.

\begin{Lemma}\label{lem:prop_bdlessproc}
Let $\widetilde{u}_T(\cdot)$ be a CM boundaryless Whittle--Mat\'ern process with parameters $(\kappa,\tau,\alpha)$, with $\alpha\in\mathbb{N}$, on the interval $[0,T]$ and let $\mv{r}(\cdot,\cdot)$ be given by \eqref{eq:R_matrix_edge_repr}. 
Then, the family $\{\widetilde{\mv{r}}_T(\cdot,\cdot): T>0\}$ of multivariate covariance functions of the multivariate processes ${\widetilde{\mv{u}}_T(\cdot) = [\widetilde{u}_T(\cdot), \widetilde{u}_T'(\cdot),\ldots, \widetilde{u}_T^{(\alpha-1)}(\cdot)]}$, satisfies conditions (i), (ii), and (iii) in Theorem \ref{thm:CondDens}, where the derivatives are weak  in the $L_2(\Omega)$ sense. Further, let $\mv{u}$ be a stationary Gaussian random field with covariance function $\mv{r}(\cdot,\cdot)$. Then, $\mv{u}$ is a Gaussian Markov random field of order 1.
\end{Lemma}

\begin{proof}
	We start by showing that $\mv{u}(\cdot)$ is a Markov random field. Let $u(\cdot)$ be a Mat\'ern random field on $\mathbb{R}$, that is, $u(\cdot)$ is Gaussian process on $\mathbb{R}$ with covariance function \eqref{eq:matern_cov}. Then, $\mv{u} \stackrel{d}{=} (u, u',\ldots, u^{(\alpha-1)})$, where the derivatives are taken weakly in the $L_2(\Omega)$ sense. It is well-known that $u(\cdot)$ is a Gaussian Markov random field of order $\alpha$, since $\alpha\in\mathbb{N}$. Indeed, it follows, e.g. from \citet[Proposition 10.2]{Pitt1971}. Finally, observe that for each Borel set $B\subset [0,T]$,
	$$\mathcal{F}_+^{u}(B) := \bigcap_{\varepsilon>0} \sigma(u(s): s\in B_\varepsilon) = \bigcap_{\varepsilon>0} \sigma(\mv{u}(s): s\in B_\varepsilon) =: \mathcal{F}_+^{{\mv{u}}}(B),$$
	where $B_\varepsilon := \{s\in [0,T]: \exists z\in B, d(s,z) < \varepsilon\}.$
	Since $u(\cdot)$ is a Gaussian Markov random field of order $\alpha$ and $\mv{u} \stackrel{d}{=} (u, u',\ldots, u^{(\alpha-1)})$, it follows that ${\mv{u}}(\cdot)$ is a Gaussian Markov random field of order 1. We refer the reader to \citet{Pitt1971} and \cite{BSW_Markov} for further details on Gaussian Markov random fields.

	We will now check conditions (i), (ii) and (iii) in Theorem \ref{thm:CondDens}.
By Proposition \ref{prp:bdlessWM_Markov}, $\widetilde{u}(\cdot)$ admits weak derivatives in the $L_2(\Omega)$ sense up to order $\alpha-1$, 
implying that the vector $\widetilde{\mv{u}}(\cdot) = [\widetilde{u}(\cdot), \widetilde{u}'(\cdot),\ldots, \widetilde{u}^{(\alpha-1)}(\cdot)]$ is well-defined.
Condition~(i) follows directly by the edge representations given in Theorem \ref{thm:ReprTheoremEdge_Refined} and Proposition \ref{prp:EdgeReprBdlessProcInterval}.
To verify condition (ii), observe that we have the following symmetry: 
$$
	({H}^\alpha([0,T]), \|\cdot\|_{\alpha,\kappa,\tau}) = (\breve{H}^\alpha([0,T]), \|\cdot\|_{\alpha,\kappa,\tau}), 
$$
where 
$
	\breve{H}^\alpha([0,T]) = \{f(T - \cdot): f\in {H}^\alpha([0,T])\}.
$ 
We also have 
\begin{equation}\label{eq:symmcond_bdless}
	\<f, g\>_{\alpha,\kappa,\tau} = \<f(T-\cdot), g(T-\cdot)\>_{\alpha,\kappa,\tau}.
\end{equation}
	
By Definition \ref{def:bdlessWM}, $({H}^\alpha([0,T]), \|\cdot\|_{\alpha,\kappa,\tau})$ is the Cameron--Martin space associated with $\widetilde{u}(\cdot)$. 
Therefore, if we let $\widetilde{\varrho}(\cdot,\cdot)$ be the covariance function of $\widetilde{u}(\cdot)$, then for every $t\in [0,T]$, 
we have  $\widetilde{\varrho}(\cdot, t)\in H^\alpha([0,T])$. 
Now, take $t_1,t_2\in [0,T]$, and use \eqref{eq:symmcond_bdless} to obtain 
$$
	\widetilde{\varrho}(T-t_2, t_1) = \<\widetilde{\varrho}(T - \cdot,t_1), \widetilde{\varrho}(\cdot,t_2)\>_{\alpha,\kappa,\tau} = \<\widetilde{\varrho}(\cdot,t_1), \widetilde{\varrho}(T-\cdot,t_2)\>_{\alpha,\kappa,\tau} = \widetilde{\varrho}(T-t_1, t_2).
$$
Differentiating the above expression $k$ times with respect to $t_1$ and $j$ times with respect to $t_2$, $j,k\in\{0,\ldots,\alpha-1\}$, gives that 
$
	\partial_{t_1}^k \partial_{t_2}^j \widetilde{\varrho}(T-t_2, t_1) = \partial_{t_1}^k \partial_{t_2}^j\widetilde{\varrho}(T-t_1, t_2).
$
Condition (ii) now follows from taking $t_1=T$ and $t_2 = 0$.

Finally, we verify condition (iii). To this end, we require the following additional notation. 
Take any $T_1 \in (0,T)$ and let $e_1 = [0,T_1]$, $e_2 = [T_1,T]$, 
$\widetilde{\mathcal{E}} = \{e_1, e_2\}$, $\widetilde{\mathcal{V}} = \{0,T_1,T\}$, and $\widetilde{\Gamma} = e_1\cup e_2$. 
Further, let $\widehat{u}_1(\cdot)$ and $\widehat{u}_2(\cdot)$ be two independent CM boundaryless Whittle--Mat\'ern processes with parameters $(\kappa,\tau,\alpha)$, with $\alpha\in\mathbb{N}$, on the edges $e_1$ and $e_2$, respectively. 
We consider the field $\widehat{u}(\cdot)$ on $\widetilde{\Gamma}$, such that $\widehat{u}|_{e_i}(\cdot) = \widehat{u}_i(\cdot)$. 
To prove condition (iii), we must establish that if we condition $(\widehat{u}(\cdot),\widehat{u}'(\cdot),\ldots,\widehat{u}^{\alpha-1}(\cdot))$ on 
$\widehat{u}_{e_1}^{(j)}(T_1) = \widehat{u}_{e_2}^{(j)}(T_1)$, for $j=0,\ldots,\alpha-1$ (i.e., on continuity of $\widehat{u}(\cdot)$ and its derivatives at $T_1$), 
then we obtain a process $\widetilde{u}(\cdot)$ that is a CM boundaryless process on $[0,T]$.
	
By Definition \ref{def:bdlessWM} and \citet[Lemma 8 in Appendix~A]{BSW_Markov}, 
the Cameron--Martin space associated with $\widehat{u}(\cdot)$ is given by $\widetilde{H}^\alpha(\widetilde{\Gamma})$, 
endowed with the inner product
${\<f,g\>_{\alpha,\kappa,\tau,\widetilde{\Gamma}} = \<f_{e_1}, g_{e_1}\>_{\alpha,\kappa,\tau,1} + \<f_{e_2}, g_{e_2}\>_{\alpha,\kappa,\tau,2}}$,
where $\<\cdot,\cdot\>_{\alpha,\kappa,\tau,i}$ is the inner product $\<\cdot,\cdot\>_{\alpha,\kappa,\tau}$ acting on $e_i$, $i=1,2$. 
Further, let $\<\cdot,\cdot\>_{\alpha,\kappa,\tau,T}$ be the inner product $\<\cdot,\cdot\>_{\alpha,\kappa,\tau}$ acting on $[0,T]$. 
Note that 
\begin{equation}\label{eq:Halpha_condcont_bdless}
	H^{\alpha}([0,T]) = \widetilde{H}^{\alpha}(\widetilde{\Gamma}) \cap \{f\in\widetilde{H}^\alpha(\widetilde{\Gamma}): f_{e_1}^{(k)}(T_1) = f_{e_2}^{(k)}(T_1), k=0,\ldots,\alpha-1\},
\end{equation} 
and if $f,g\in H^{\alpha}([0,T])$, then $\<f,g\>_{\alpha,\kappa,\tau,\widetilde{\Gamma}} = \<f,g\>_{\alpha,\kappa,\tau,T}$. 
Thus, $(H^\alpha([0,T]), \<\cdot,\cdot\>_{\alpha,\kappa\tau,T})$ is a closed subspace of $(\widetilde{H}^\alpha(\widetilde{\Gamma}), \<\cdot,\cdot\>_{\alpha,\kappa\tau,\widetilde{\Gamma}})$. 
Now, let $H_{\widehat{u}}([0,T])$ be the linear Gaussian space generated by $\widehat{u}$ 
(i.e., the completion of $\textrm{span}\{\widehat{u}(t):t\in [0,T]\}$ with respect to the $L_2(\Omega)$ norm). 
Recall that $(\widetilde{H}^\alpha(\widetilde{\Gamma}), \<\cdot,\cdot\>_{\alpha,\kappa\tau,\widetilde{\Gamma}})$ is the Cameron--Martin space associated with the field $\widehat{u}(\cdot)$, so that 
$$
	\widetilde{H}^\alpha(\widetilde{\Gamma}) = \{h(s) = \pE(\widehat{u}(s)v): s\in \Gamma\hbox{ and } v\in H_{\widehat{u}}(\widetilde{\Gamma})\}
$$
and $\Phi:H_{\widehat{u}}(\widetilde{\Gamma})\to \widetilde{H}^\alpha(\widetilde{\Gamma})$ given by $\Phi(v)(s) = \pE(\widehat{u}(s)v)$, for $v\in H_{\widehat{u}}(\widetilde{\Gamma})$ is an isometric isomorphism. 
Let $\Pi_C:\widetilde{H}^\alpha(\widetilde{\Gamma})\to H^\alpha([0,T])$ be the orthogonal projection onto $H^\alpha([0,T])$ with respect to the inner product $\<\cdot,\cdot\>_{\alpha,\kappa,\tau,\widetilde{\Gamma}}$. 
In addition, let $\widehat{\varrho}(\cdot,\cdot)$ be the reproducing kernel of $\widetilde{H}^\alpha(\widetilde{\Gamma})$ 
and observe that, following the same arguments as in Appendix \ref{app:proofs_conditional}, 
$(H^\alpha([0,T]), \<\cdot,\cdot\>_{\alpha,\kappa,\tau,T})$ is also a reproducing kernel Hilbert space, 
with reproducing kernel ${\widetilde{\varrho}(s,\cdot) = \Pi_C(\widehat{\varrho}(s,\cdot))}$, $s\in\widetilde{\Gamma}$. 
By defining the process $\widetilde{u}(s) = \Phi^{-1}(\widetilde{\varrho}(t,\cdot)), s\in \widetilde{\Gamma}$, 
it follows by construction that $\widetilde{u}(\cdot)$ has the Cameron--Martin space given by $(H^\alpha([0,T]), \<\cdot,\cdot\>_{\alpha,\kappa,\tau,T})$, 
that is, by Definition \ref{def:bdlessWM}, $\widetilde{u}(\cdot)$ is a CM boundaryless Whittle--Mat\'ern field with parameters $(\kappa,\tau,\alpha)$ on the interval $[0,T]$. Next, let $H_{\widetilde{u}}([0,T]) = \Phi^{-1}(H_{\widehat{u}}(\widetilde{\Gamma}))$ be the linear Gaussian space associated with $\widetilde{u}(\cdot)$. 
Then, $\widetilde{u}(s) = \pE(\widehat{u}(s)|\sigma(H_{\widetilde{u}}([0,T])))$, for $s\in \widetilde{\Gamma}$.
We can now proceed as in Appendix \ref{app:proofs_conditional}, using \eqref{eq:Halpha_condcont_bdless}, 
to conclude that $\widetilde{u}(\cdot)$ is obtained from $\widehat{u}(\cdot)$ by conditioning on 
$\widehat{u}_{e_1}^{(k)}(T_1) = \widehat{u}_{e_2}^{(k)}(T_1), k =0,\ldots,\alpha-1$, 
where the derivatives are weak derivatives in the $L_2(\Omega)$ sense. 
This proves part of condition (iii). 
It remains to be established that $\widetilde{u}^{(k)}(\cdot)$ is obtained from $\widehat{u}^{(k)}(\cdot)$ by conditioning on 
$\widehat{u}_{e_1}^{(k)}(T_1) = \widehat{u}_{e_2}^{(k)}(T_1), k =0,\ldots,\alpha-1$, 
where the derivatives are weak in the $L_2(\Omega)$ sense. 
However, observe that, directly from the definition of weak derivatives in the $L_2(\Omega)$ sense, we obtain
\begin{equation}\label{eq:cond_exp_bdless_weakderiv}
	\widetilde{u}^{(k)}(s) = \pE(\widehat{u}^{(k)}(s)|\sigma(H_{\widetilde{u}}([0,T]))), \quad s\in \widetilde{\Gamma}, \quad k=0,\ldots,\alpha-1.
\end{equation}
Further, by taking derivatives on the edge representation for the CM boundaryless Whittle--Mat\'ern processes (Proposition \ref{prp:EdgeReprBdlessProcInterval}), 
we obtain edge representations for the weak derivatives (in the $L_2(\Omega)$ sense) of $\widetilde{u}(\cdot)$ and $\widehat{u}(\cdot)$. 
The corresponding results for the weak derivatives of $\widetilde{u}(\cdot)$ and $\widehat{u}(\cdot)$ follow from applying the same arguments as in 
Appendix~\ref{app:proofs_conditional} combined with \eqref{eq:cond_exp_bdless_weakderiv} and the edge representation for the weak derivatives. 
Furthermore, by the same arguments, for any $a_0,\ldots,a_{\alpha-1}\in\mathbb{R}$, 
the linear combination $\sum_{k=0}^{\alpha-1} a_k \widetilde{u}^{(k)}$ can be obtained by conditioning $\sum_{k=0}^{\alpha-1} a_k \widehat{u}^{(k)}$ on $\widehat{u}_{e_1}^{(k)}(T_1) = \widehat{u}_{e_2}^{(k)}(T_1), k =0,\ldots,\alpha-1$. 
Therefore, $(\widetilde{u}(\cdot), \widetilde{u}'(\cdot),\ldots,\widetilde{u}^{(\alpha-1)}(\cdot))$ can be obtained from $(\widehat{u}(\cdot), \widehat{u}'(\cdot),\ldots,\widehat{u}^{(\alpha-1)}(\cdot))$ by conditioning on $\widehat{u}_{e_1}^{(k)}(T_1) = \widehat{u}_{e_2}^{(k)}(T_1), k =0,\ldots,\alpha-1$, which proves condition (iii).
\end{proof}

We have a final auxiliary lemma we need to be able to apply Theorem \ref{thm:CondDens}. 

\begin{Lemma}\label{lem:strpos_def_stationary}
The function $\mv{r}(\cdot,\cdot)$ in \eqref{eq:R_matrix_edge_repr} is strictly positive-definite (see Definition \ref{def:str_pos_def_mat}).
\end{Lemma}

\begin{proof}
Let $N\in\mathbb{N}$ and take $t_1,\ldots, t_N\in\mathbb{R}$. Define the cross-covariance matrix
$$\mv{\Sigma} = \begin{bmatrix}
\mv{\Sigma}_{11}& \mv{\Sigma}_{12}& \cdots & \mv{\Sigma}_{1N}\\
\mv{\Sigma}_{21}& \mv{\Sigma}_{22}& \cdots & \mv{\Sigma}_{2N}\\
\vdots & \vdots & \ddots & \vdots\\
\mv{\Sigma}_{N1}& \mv{\Sigma}_{N2}& \cdots & \mv{\Sigma}_{NN}
\end{bmatrix},$$
where $\mv{\Sigma}_{pj}$ has $(k,l)$th element given by
$$
	(\mv{\Sigma}_{pj})_{kl} = \frac{\pd^{l-1}}{\pd t_j^{l-1}}\frac{\pd^{k-1}}{\pd t_p^{k-1}}\varrho_M(t_p-t_j) = (-1)^{l-1} \varrho_M^{(k-1+l-1)}(t_p - t_j),
$$
where $p,j=1,\ldots,N$ and $k,l = 1,\ldots, \alpha.$ 
Let $f(\cdot)$ be the spectral density of $\varrho_M(\cdot)$ and observe that, by e.g., \citet{stein99}, 
$$
	\int_{-\infty}^\infty \omega^{2(\alpha-1)} f(\omega)\,d\omega < \infty,
$$
since $\varrho_M(\cdot)$ is $\alpha-1$ times differentiable in the $L_2(\Omega)$-sense. 
Thus, by the dominated convergence theorem, 
$$
	\varrho_M^{(k)}(t) = \int_{-\infty}^\infty (i\omega)^k e^{it\omega} f(\omega) \, d\omega,\quad k=0,\ldots,2(\alpha-1).
$$ 
Take $\mv{a}_1,\ldots, \mv{a}_N \in \mathbb{R}^\alpha$, with $\mv{a}_p = (a_{p,1},\ldots,a_{p,\alpha})^\top, p=1,\ldots,N$. Then,  
\begin{align}
\sum_{p,j=1}^N \mv{a}_p^\top \mv{\Sigma}_{pj} \mv{a}_j &= \sum_{p,j=1}^N \sum_{k,l=1}^{\alpha} a_{p,k}(\mv{\Sigma}_{p,j})_{kl}a_{j,l}
= \sum_{p,j=1}^N \sum_{k,l=1}^{\alpha} a_{p,k} (-1)^{l-1} \varrho_M^{(k-1+l-1)}(t_p-t_j) a_{j,l}\nonumber\\
&= \sum_{p,j=1}^N \sum_{k,l=1}^{\alpha} a_{p,k} (-1)^{l-1} \int_{-\infty}^\infty (i\omega)^{k+l-2} e^{i(t_p-t_j)\omega} f(\omega) \, d\omega a_{j,l}\nonumber\\
&= \int_{-\infty}^\infty \sum_{k,l=1}^\alpha  (i\omega)^{k-1} (-i\omega)^{l-1} \left(\sum_{p,j=1}^N a_{p,k} e^{it_p\omega} e^{-it_j \omega} a_{jl} \right) f(\omega)\, d\omega\nonumber\\
&= \int_{-\infty}^\infty \sum_{k,l=1}^\alpha h_k(\omega) \overline{h_l(\omega)} f(\omega)\, d\omega = \int_{-\infty}^\infty |\mv{h}(\omega)^\top \mv{1}|^2 f(\omega)\,d\omega \geq 0,\label{eq:ineq_pos_str_pos_statio}
\end{align}
where $h_k:\mathbb{R}\to\mathbb{C}$, $h_k(\omega) = (i\omega)^{k-1} \sum_{p=1}^\alpha a_{p,k} e^{it_p\omega}$, $k=1,\ldots,\alpha$, $\mv{h}(\cdot) = (h_1(\cdot), \ldots, h_\alpha(\cdot))$ and $\mv{1} = (1,\ldots,1)^\top \in\mathbb{R}^\alpha$.
Thus, to conclude the proof we must show that the inequality in \eqref{eq:ineq_pos_str_pos_statio} is  strict. We prove this claim by contraposition, 
that is, we assume that the integral in \eqref{eq:ineq_pos_str_pos_statio} is equal to zero, and then show that this implies that for every $p=1,\ldots,N$, 
 $\mv{a}_p=\mv{0}$. 

Now, observe that $f(\omega)>0$ a.e., so this integral is equal to zero if, and only if, $|\mv{h}(\omega)^\top \mv{1}|^2 = 0$ a.e., which in turn implies that $\mv{h}(\omega)^\top \mv{1} = 0$ a.e. Hence, $\sum_{k=1}^\alpha h_k(\omega) = 0$ a.e. By taking the inverse Fourier transform (in the space of tempered distributions), we obtain, in particular, that
\begin{equation}\label{eq:identity_dirac_deriv}
	\forall f\in C_c^\infty(\mathbb{R}), \sum_{k=1}^\alpha \sum_{p=1}^N a_{p,k} \delta_{t_p}^{(k-1)}(f) = 0 \Rightarrow \forall f\in C_c^\infty(\mathbb{R}),\sum_{k=1}^\alpha \sum_{p=1}^N a_{p,k} f^{(k-1)}(t_p) = 0,
\end{equation}
where $\delta_{t_p}(\cdot)$ is Dirac's delta measure concentrated at $t_p$, and $\delta_{t_p}^{(k-1)}(\cdot)$ is the $k$th derivative, in the distribution sense, of Dirac's delta measure, $k=1,\ldots,\alpha$.

Take any $j\in\{1,\ldots,N\}$, and let $\epsilon_j := \min\{|t_j - t_p|: p=1,\ldots, N, p\neq j\} > 0$. Now, let $f_{j,0} \in C_c^\infty(\mathbb{R}),j=1,\ldots,N,$ be such that the support of $f_{j,0}$ is contained in ${B_{\epsilon_j/2}(t_j) := \{y\in\mathbb{R}: |y-t_j|<\epsilon_j/2\}}$, with $f_{j,0}(t_j)>0$ and $f_{j,0}(\cdot)$ being constant in a neighborhood of $t_j$. Then, by applying \eqref{eq:identity_dirac_deriv} to $f_{j,0}$, we obtain that for every $j=1,\ldots, N$, $a_{j1} = 0$. 
Now, proceed in a similar manner. By taking $j=1,\ldots,N$, and $k=1,\ldots,\alpha-1$, let $f_{j,k}(\cdot)$ be such that the support of $f_{j,k}$ is contained in $B_{\epsilon_j/2}(t_j)$, $f_{j,k}^{(k)}(t_j)>0$ and $f_{j,k}^{(k)}(\cdot)$ is constant on a neighborhood of $t_j$. Then, by applying \eqref{eq:identity_dirac_deriv} to $f_{j,k}$, in an inductive manner with respect to $k$ (that is, we first apply to $k=1$, and use that $a_{j1}=0$, to obtain $a_{j2}=0$, then we do the same for $k=2$, and so on), we obtain that for every $j=1,\ldots, N$, $a_{j2} = 0,\ldots, a_{j\alpha} = 0$. That is, we obtain that for every $p=1,\ldots,N$, $\mv{a}_p=0$. Thus, if at least one $\mv{a}_p\neq \mv{0}$, $p=1,\ldots,N$, then the integral in \eqref{eq:ineq_pos_str_pos_statio} is strictly positive. 
\end{proof}

Finally, we can demonstrate that Definitions \ref{def:bdlessWM} and \ref{def:cov_based_bdlessWM} are equivalent.

\begin{Proposition}\label{prp:bdless_thm_cov_123}
Let $\widetilde{u}(\cdot)$ be a CM boundaryless Whittle--Mat\'ern process from Definition \ref{def:bdlessWM}, with parameters $(\kappa,\tau,\alpha)$, with $\alpha\in\mathbb{N}$, on the interval $[0,T]$. Then, the process $\widetilde{\mv{u}}(\cdot) = [\widetilde{u}(\cdot), \widetilde{u}'(\cdot),\ldots, \widetilde{u}^{(\alpha-1)}(\cdot)]$ has multivariate covariance function given by \eqref{eq:covmod}, where $\mv{r}(\cdot,\cdot)$ is given by \eqref{eq:R_matrix_edge_repr}, and the derivatives are weak derivatives in the $L_2(\Omega)$ sense.
\end{Proposition}
\begin{proof}
The result follows by combining Lemma \ref{lem:prop_bdlessproc} and Lemma \ref{lem:strpos_def_stationary} with Theorem~\ref{thm:CondDens}.
\end{proof}

\section{Proofs from Sections~\ref{sec:inference} and \ref{sec:prediction}}\label{app:inference}

\begin{proof}[Proof of Proposition \ref{Them:piAXsoft}]
Note that $\pi_{\mv{Y} | \mv{KU}}\left(  \mv{y} | \mv{b} \right)= \pi_{\mv{Y} | \mv{U}_\A^*}(\mv{y}|\mv{b}^*)$ and
\begin{align}
	\pi_{\mv{Y} | \mv{U}^*_{\A}}\left(  \mv{y} | \mv{b}^* \right) &=  \int  \pi_{\mv{U}^*_{\Ac},\mv{Y}| \mv{U}^*_{\A} }\left(\mv{u}^*_{\Ac}, \mv{y}| \mv{b}^* \right) d\mv{u}^*_{\Ac} \notag\\
	&=  \int  \pi_{\mv{Y}|\mv{U}^*_{\Ac}, \mv{U}^*_{\A}}(\mv{y}|\mv{u}_{\Ac}^*, \mv{b}^*) \pi_{\mv{U}^*_{\Ac}| \mv{U}^*_{\A}}(\mv{u}_{\Ac}^*| \mv{b}^*) d\mv{u}^*_{\Ac}.	\label{eq:decomposePI}
\end{align}
The goal is to derive an explicit form of the density by evaluating the integral in \eqref{eq:decomposePI}. 
To shorten the notation, we let $q(\mv{K},\mv{x})$ denote the quadratic form $\mv{x}^\top\mv{K}\mv{x}$ for a matrix $\mv{K}$ and a vector $\mv{x}$. 
First, it is straightforward to see that 
\begin{align}\label{eq:proof_cond1}
	\pi_{\mv{Y}|\mv{U}_\Ac^*,\mv{U}_{\A}^*}(\mv{y}|\mv{u}_{\Ac}^*,\mv{b}^*) 
	&=  \frac{1}{(2\pi)^{\frac{n}{2}}|\mv{\Sigma}|^{1/2}}\exp \left( -\frac{1}{2}q\left(\mv{\Sigma}^{-1}, \mv{y}- \mv{B}^*
	\begin{bmatrix}
		\mv{b}^* \\ 
		\mv{u}^*_{\Ac} 
	\end{bmatrix}
	\right)\right) \\
	&\propto \exp\left(-\frac{1}{2}\mv{u}^{*\top}_{\Ac}\mv{B}^{*\top}_{\Ac} \mv{\Sigma}^{-1} \mv{B}^*_{\Ac}\mv{u}^*_{\Ac}+ \mv{y}^{\top}\mv{\Sigma}^{-1}\mv{B}^*_{\Ac}\mv{u}^*_{\Ac} \right),\nonumber
\end{align}
as a function of $\mv{u}_{\Ac}^*$.
From the proof of Theorem 2 in \cite{bolin2021efficient}, we obtain
\begin{equation}\label{eq:Xconditional}
	\mv{U}_{\Ac}^*|\mv{U}^*_\A= \mv{b}^* \sim  \pN\left(
	\mv{\mu}^*_{\Ac}-\left(\mv{Q}_{\Ac\Ac}^{*}\right)^{-1} \mv{Q}_{\Ac\A}^{*} \left( \mv{b}^* -  \mv{\mu}^*_{\A}\right)
	, (\mv{Q}_{\Ac\Ac}^*)^{-1} \right).
\end{equation}
By these expressions, 
\begin{align*}
	&\pi_{\mv{Y}|\mv{U}^*_{\Ac}, \mv{U}^*_{\A}}(\mv{y}|\mv{u}_{\Ac}^*, \mv{b}^*) \pi_{\mv{U}^*_{\Ac}| \mv{U}^*_{\A}}(\mv{u}_{\Ac}^*| \mv{b}^*) =\\
	&= \exp\left (-\frac{1}{2}\mv{u}^{*\top}_{\Ac} \mv{B}^{*\top}_{\Ac}\mv{\Sigma}^{-1} \mv{B}^*_{\Ac} \mv{u}^*_{\Ac}+ \left( \mv{B}^{*\top}_{\Ac} \mv{\Sigma}^{-1}\mv{y} \right)^{\top}\mv{u}^*_{\Ac} \right) \frac{ |\mv{Q}^*_{\Ac\Ac}|^{1/2}}{\left(2\pi\right)^{\nicefrac{(n+m-k)}{2}} |\mv{\Sigma}|^{1/2}} \cdot \\
	&\quad\exp\left(-\frac{1}{2} \mv{u}^{*\top}_{\Ac} \mv{Q}^*_{\Ac\Ac} \mv{u}^*_{\Ac}+ \left(  \mv{Q}^*_{\Ac\Ac} \widetilde{\mv{\mu}}^*_{\Ac}\right)^{\top} \mv{u}^*_{\Ac}   \right)
	\exp \left(- \frac{1}{2}\left[  \mv{y}^{\top}\mv{\Sigma}^{-1}\mv{y}+  \widetilde{\mv{\mu}}_{\Ac}^{*\top}  \mv{Q}^*_{\Ac\Ac}  \widetilde{\mv{\mu}}^*_{\Ac} \right] \right) \\
	&=\pi_{\mv{U}^*_{\Ac}|\mv{Y}, \mv{U}^*_{\A}}(\mv{u}_{\Ac}^*| \mv{y},\mv{b}^*)  
	\frac{\exp \left( \frac{1}{2}  \widehat{\mv{\mu}}_{\Ac}^{*\top} \widehat{\mv{Q}}^*_{\Ac\Ac} \widehat{\mv{\mu}}^*_{\Ac} \right) }{|\mv{Q}^*_{\Ac\Ac}|^{-1/2}|\mv{\Sigma}|^{1/2} \left(2\pi\right)^{\nicefrac{n}{2}} }
	\exp \left(- \frac{1}{2}\left[  \mv{y}^{\top} \mv{\Sigma}^{-1}\mv{y} +  \mv{\mu}_{\Ac}^{*\top}  \mv{Q}^*_{\Ac\Ac}\mv{\mu}^*_{\Ac}\right]\right).
\end{align*}
Inserting this expression in \eqref{eq:decomposePI} and evaluating the integral, noting  that $\pi_{\mv{U}^*_{\Ac}|\mv{Y}, \mv{U}^*_{\A}}(\mv{u}_{\Ac}^*| \mv{y},\mv{b}^*)$ integrates to 1, gives the desired result. 
\end{proof}

\begin{proof}[Proof of Proposition \ref{Them:piXgby}]
Note that 
$$
\pi_{\mv{U}_{\Ac}^*|\mv{Y},\mv{U}^*_{\A} } (\mv{u}^*_{\Ac}| \mv{y},\mv{b}^*) \propto
\pi_{\mv{Y}|\mv{U}_\A^*,\mv{U}_{\Ac}^*}(\mv{y}|\mv{b}^*, \mv{u}_{\Ac}^* ) \pi(\mv{u}^*_{\Ac}|\mv{b}^*).
$$ 
By \eqref{eq:proof_cond1} and \eqref{eq:Xconditional} we have  
$\pi_{\mv{U}^*_{\Ac}|\mv{U}^*_\A}(\mv{u}^*_{\Ac}|\mv{b}^*) \propto
\exp \left( -\frac{1}{2}  \left( \mv{u}^*_{\Ac}- \widetilde{\mv{\mu}}^*_{\Ac} \right)^\top  \mv{Q}^*_{\Ac\Ac}  \left( \mv{u}^*_{\Ac}- \widetilde{\mv{\mu}}^*_{\Ac} \right) \right),$ 
as a function of $\mv{u}_{\Ac}^*$,
where $\widetilde{\mv{\mu}}^*_{\Ac} =  \mv{\mu}^*_{\Ac}-\left(\mv{Q}_{\Ac\Ac}^{*} \right)^{-1}\mv{Q}_{\Ac\A}^{*} \left( \mv{b}^* -  \mv{\mu}^*_{\A}\right)$.
Thus, it follows that
\begin{align*}
	\pi_{\mv{U}_{\Ac}^*|\mv{Y},\mv{U}^*_{\A} } (\mv{u}^*_{\Ac}| \mv{y},\mv{b}^*)  \propto
	&\exp\left (-\frac{1}{2}\mv{u}^{*\top}_{\Ac} \mv{B}^{*\top}_{\Ac}\mv{\Sigma}^{-1} \mv{B}^*_{\Ac}\mv{u}^*_{\Ac}+ \left( \mv{B}^{*\top}_{\Ac} \mv{\Sigma}^{-1}\mv{y} \right)^{\top}\mv{u}^*_{\Ac} \right) \times \\  
	\times& \exp\left(-\frac{1}{2} \mv{u}^{*\top}_{\Ac} \mv{Q}^*_{\Ac\Ac} \mv{u}^*_{\Ac}+ \left(  \mv{Q}^*_{\Ac\Ac} \widetilde{\mv{\mu}}^*_{\Ac}\right)^{\top} \mv{u}^*_{\Ac}   \right) \\
	\propto 	&\exp\left(- \frac{1}{2} \left(\mv{u}_{\Ac}^*-\widehat{\mv{\mu}}_{\Ac}^* \right)^\top \widehat{\mv{Q}}_{\Ac\Ac}^*
	\left(\mv{u}_{\Ac}^*-\widehat{\mv{\mu}}_{\Ac}^* \right)\right).
\end{align*}
Finally, the relation $\mv{U}= \mv{T}^\top\mv{U}^*$ completes the proof.
\end{proof}
	\end{appendix}  
\end{document}